\documentclass{statsoc}
\makeatletter
\AtBeginDocument{%
  \pdfpagewidth=\paperwidth
  \pdfpageheight=\paperheight
}
\makeatother

\usepackage{etoolbox}

\makeatletter
\patchcmd{\@makecaption}
  {\parbox}
  {\advance\@tempdima-\fontdimen2}
  {}{}
\makeatother

\usepackage{longtable}
\usepackage{booktabs}

\usepackage{amsmath,amssymb,amsfonts,mathtools,bm}
\usepackage{bbm}
\usepackage{mathrsfs}

\usepackage{xcolor}
\usepackage{enumitem}

\usepackage{xr-hyper} 

\usepackage[colorlinks=true,linkcolor=blue,citecolor=blue,urlcolor=blue]{hyperref} 

\usepackage[nameinlink,capitalise,noabbrev]{cleveref}

\usepackage{array}

\newcolumntype{L}[1]{>{\raggedright\arraybackslash}p{#1}}

\usepackage{tikz}
\usetikzlibrary{positioning,arrows.meta}
\usepackage{adjustbox} 
\usetikzlibrary{arrows.meta,positioning,fit,calc}

\usepackage{amsmath}
\usepackage{tikz}
\usetikzlibrary{arrows.meta, positioning, fit, calc, shapes.geometric, backgrounds}

\usepackage{aliascnt}

\newtheorem{theorem}{Theorem}[section]

\newaliascnt{lemma}{theorem}
\newtheorem{lemma}[lemma]{Lemma}
\aliascntresetthe{lemma}

\newaliascnt{proposition}{theorem}
\newtheorem{proposition}[proposition]{Proposition}
\aliascntresetthe{proposition}

\newaliascnt{corollary}{theorem}
\newtheorem{corollary}[corollary]{Corollary}
\aliascntresetthe{corollary}

\newaliascnt{definition}{theorem}
\newtheorem{definition}[definition]{Definition}
\aliascntresetthe{definition}

\newaliascnt{assumption}{theorem}
\newtheorem{assumption}[assumption]{Assumption}
\aliascntresetthe{assumption}

\newaliascnt{remark}{theorem}
\newtheorem{remark}[remark]{Remark}
\aliascntresetthe{remark}

\crefname{theorem}{Theorem}{Theorems}
\crefname{lemma}{Lemma}{Lemmas}
\crefname{proposition}{Proposition}{Propositions}
\crefname{corollary}{Corollary}{Corollaries}
\crefname{definition}{Definition}{Definitions}
\crefname{remark}{Remark}{Remarks}
\crefname{assumption}{Assumption}{Assumptions}
\crefname{equation}{equation}{equations}
\crefname{section}{Section}{Sections}
\crefname{subsection}{Section}{Sections}

\newcommand{\bbE}{\mathbb{E}}
\newcommand{\bbP}{\mathbb{P}}
\newcommand{\dd}{\mathrm{d}}

\newcommand{\norm}[1]{\left\lVert #1 \right\rVert}
\newcommand{\bbR}{\mathbb{R}}
\newcommand{\1}{\mathbbm{1}}

\newcommand{\bbN}{\mathbb{N}}

\newcommand{\CovX}{\mathbf{X}}

\newcommand{\Fobs}{\mathcal H}

\newcommand{\Tsig}{\mathscr T}
\newcommand{\Nproc}{N}
\newcommand{\Nzero}{N_0}
\newcommand{\None}{N_1}

\newcommand{\Tmax}{\mathfrak{T}}

\DeclareMathOperator*{\argmax}{arg\,max}

\newif\ifinappendix
\inappendixfalse

\let\oldappendix\appendix
\renewcommand{\appendix}{%
  \oldappendix
  \inappendixtrue
}

\crefformat{section}{\ifinappendix Supplement~#2#1#3\else Section~#2#1#3\fi}
\Crefformat{section}{\ifinappendix Supplement~#2#1#3\else Section~#2#1#3\fi}
\crefformat{subsection}{\ifinappendix Supplement~#2#1#3\else Section~#2#1#3\fi}
\Crefformat{subsection}{\ifinappendix Supplement~#2#1#3\else Section~#2#1#3\fi}

\title[Spatiotemporal point process causal inference]%
{Causal inference for spatiotemporal point processes in the presence of outcome spillover and carryover}

\author[Kresin {\it et al.}]{Conor Kresin}\coaddress{Conor Kresin, Department of Mathematics and Statistics, University of Otago,
P.O. Box 56, Dunedin 9054, New Zealand.}
\address{Department of Mathematics and Statistics, University of Otago, Dunedin, New Zealand.}
\email{conor.kresin@otago.ac.nz}

\author{Duncan A. Clark}
\address{Department of Statistics, Williams College, Williamstown, MA, USA.}

\author{Louis Davis}
\address{Department of Statistics, Stanford University, Stanford, CA, USA.}

\author[Kresin {\it et al.}]{Martin L.\ Hazelton}
\address{Department of Mathematics and Statistics, University of Otago, Dunedin, New Zealand.}

\begin{document}

\begin{abstract}
We develop a framework for causal inference with continuous spatiotemporal point-process outcomes under cell-level interventions and outcome spillover. Potential outcomes are indexed by full treatment allocations, and the observed post-treatment process is represented as an unlabelled superposition of latent control and treatment components. On the observed design support, expected post-treatment event counts in any spacetime region under a given treatment allocation are identified under consistency, exchangeability, and positivity; off-support contrasts are identified relative to a regime-stable structural point-process model. Estimation is likelihood-based and implemented with stochastic EM. To understand when this is feasible, we analyse a predictable blockwise hard-EM surrogate and show nonasymptotic contraction of estimation error to a statistical floor governed by locally ambiguous regions. This yields plug-in guarantees for cell-level and global causal functionals, and clarifies the additional array conditions needed for unnormalised growing-window contrasts. The framework covers history dependent spatiotemporal point processes including Poisson and Hawkes models, with applications to settings such as epidemiology, seismology, and finance. We provide an application assessing the causal effect of injecting wastewater into the ground on seismic activity in Oklahoma.
\end{abstract}

\section{Introduction}
Economists, epidemiologists, and social scientists are perennially interested in measuring the effect of a treatment, policy implementation, or intervention. The central difficulty in estimating causal effects stems from the absence of counterfactual outcomes \citep{rubin2005causal}, prompting the development of various paradigms \citep{pearl2009causal,splawa1990application} to determine when and how these unobserved outcomes can be estimated and used for subsequent causal inference.

Many phenomena of interest are most naturally characterised as continuous spatiotemporal data, and often such processes take the form of point patterns defined by the complex interaction in both time and space domains \citep{daley2003introduction}. 
It is not immediate how to adapt such data to a potential-outcomes setting without aggregation; doing so discards exact locations, introduces the modifiable areal unit problem, and obscures within-cell spatiotemporal structure \citep{fotheringham1991modifiable}. A related problem is that outcomes for continuous spatiotemporal data are often spatially or temporally dependent \citep{reich2021review}, generating cascades of events that potentially obfuscate the impact of an intervention. Such dependence is termed \textit{spillover} in the spatial context or \textit{carryover} in the temporal context, and if unaccounted for, results in inconsistent estimation and invalid inference \citep{papadogeorgou2019adjusting}.

Standard methods for causal inference that assume units are independent can be seriously misleading, especially when causal conclusions have policy implications \citep{benjamin2018spillover,lee2021network}.  There has been much recent work in proposing general frameworks to deal with these types of problems \citep{papadogeorgou2022causal,jiang2023instrumental} and in particular spatial spillover \citep{papadogeorgou2023spatial}. In this paper, we develop a model-based framework for continuous spatiotemporal point processes under cell-level interventions with outcome spillover and carryover (henceforth, we refer to both as ``spillover''). Such spillover occurs in spatiotemporal data wherein points can interact (for instance contagious disease processes feature self-excitement) or when points are spatially dispersed.  We represent the observed post-treatment process as an unlabelled superposition of latent control and treatment components. Estimation is likelihood-based via stochastic expectation maximisation (EM), the practical feasibility of which is examined theoretically through analysis of a predictable blockwise hard-EM surrogate. This conditional-intensity formulation retains exact event times and locations rather than relying on coarse aggregation.

\subsection{Overview of our approach}

Before proceeding, we briefly summarise our framework in informal terms. We observe events $(t_i,x_i)$ in continuous time and space. For example, such events could comprise infection in a city after a neighbourhood-level vaccination rollout, or earthquake locations after wastewater regulation. Treatment is applied in some parts of the spatiotemporal region but not others, and because nearby events can ``excite'' further events in space and time, outcomes in one region may depend on treatment in neighbouring regions. In the vaccination example, cases in a treated neighbourhood may still be driven by infections just across the boundary in an untreated one. To handle this outcome spillover while retaining exact event times and locations, we take the following steps.
\begin{enumerate}[label=(\roman*), nosep]
    \item We partition the post-treatment spatiotemporal window into finitely many cells (e.g.\ city neighbourhoods), which serve as the units of analysis, and define potential outcomes and causal estimands in terms of expected event counts under alternative treatment allocations.
    \item We model the observed post-treatment data as the superposition of latent control and treatment point processes. The events are unlabelled, so even in a treated cell an observed event need not be a treatment event, and vice versa.
    \item Given a candidate labelling, we estimate the model parameters by maximising the complete-data log-likelihood and derive conditions under which the resulting plug-in causal estimates are stable and accurate.
\end{enumerate}

The paper is organised as follows. Sections~\ref{sec:ppFramework}--\ref{sec:identification} develop the framework, estimands, notation, and identification. Sections~\ref{sec:estimation}--\ref{sec:finite-sample-theory} present estimation, a practical algorithm for estimation, and finite-sample theory deriving sufficient conditions for estimation. Sections~\ref{sec:simulation_study} and \ref{sec:application} report the simulation study and an application to seismic activity in Oklahoma, and \cref{sec:conclusion} discusses implications and extensions.

\section{Framework and causal estimands}\label{sec:ppFramework}

We observe a spatiotemporal point pattern \(\Nproc\) on \(D_{\rm obs}=[0,T]\times\mathcal S\subset [0,\infty)\times\mathbb R^d\) (typically \(d=2\)). A realisation is \(\phi=\{(t_i,x_i)\}_{i=1}^{\Nproc(D_{\rm obs})}\). We write the base product measure as \(\dd\tau=\dd t\otimes\dd\mu(x)\) and set \(|B|:=\int_B \dd\tau\). Treatment is applied at time \(t^\ast\in(0,T]\), and the post-treatment analysis window is \(D=(t^\ast,T]\times\mathcal S\). On \(D\), the observed process is represented as the superposition of two simple latent component processes, a control component \(\Nzero\) and a treatment component \(\None\), so that \(\Nproc=\Nzero+\None\) and \(\Nproc(B)=\Nzero(B)+\None(B)\) for every Borel set \(B\subseteq D\). When interaction between components is allowed, we view \((\Nzero,\None)\) as a single bivariate point process.

Point process data often occurs in continuous spacetime, and therefore, however defined, there are infinite potential outcomes without some level of aggregation. To remedy this problem, we partition the post-treatment subset of observation window $(t^\ast,T]\times \mathcal{S}$ into $J$ cells $\mathcal{I}_j$ such that

\begin{assumption}[Exogenous finite tessellation]\label{ass:tessellation}
The tessellation $\mathcal T=\{\mathcal I_j\}_{j=1}^J$ is a finite, time-invariant partition of $(t^\ast,T]\times\mathcal S$: there exists a finite partition $\{\mathcal S_j\}_{j=1}^J$ of $\mathcal S$ such that
\[
\mathcal I_j := (t^\ast,T]\times \mathcal S_j,\qquad j=1,\ldots,J,
\]
with $|\mathcal I_j|>0$ for all $j$. Moreover, there exists a $\sigma$-field $\Tsig$ such that $\mathcal T$ is $\Tsig$-measurable and $\Tsig$ is independent of $\sigma(\Nproc)$. If $\mathcal T$ is deterministic, the exogeneity condition is vacuous.
\end{assumption}

$\mathcal{T}$ is assumed to be fixed \textit{a priori}; all potential outcomes derived are implicitly conditional upon $\mathcal{T}$. In many practical applications, $\mathcal{T}$ may naturally be dependent on the secondary characteristics of some exogenous point process. For instance $\mathcal{T}$ could be a Voronoi tessellation of an independent process realisation. Generally, any tessellation satisfying Assumption~\ref{ass:tessellation} can be selected by the practitioner, and the selection of tessellation scheme may depend on the causal questions of interest. 

Naive aggregation wherein points falling in continuous spacetime are binned and counted is insufficient for valid causal estimation as modelling the spatiotemporal structure within cells is necessary for addressing spillover. In contrast, partitioning the post-treatment subset of the observation window allows for the spatiotemporal structure within cells to be retained while simultaneously achieving finite outcomes by Assumption~\ref{ass:tessellation}.

A potential concern when defining outcomes based on spatial partitions is the modifiable areal unit problem. This problem arises because the results of analysing spatially aggregated data can depend on the specific scale (granularity) and geometry of the spatial units used \citep{fotheringham1991modifiable}. 
In our setting, we aim to estimate the underlying continuous conditional intensity function, which bypasses the information loss typical of purely aggregate analyses. However, the specific potential outcomes and the derived causal estimands, like the individual treatment effect (defined below), are explicitly defined with respect to the chosen tessellation $\mathcal{T}$ and its cells $\{\mathcal{I}_j\}$. This is discussed in detail in \cref{sec:scale-of-effect}. 

Treatment is assigned (or observed) at the cell level. Let $Z$ denote the binary treatment vector, and $Z_j$ the treatment status of $\mathcal{I}_j$. In the simplest experimental setting, treatment could be assigned randomly to cells where $Z_j\overset{\text{iid}}{\sim}\text{Bernoulli}(\zeta)$ where $\zeta$ is the probability of treatment. Throughout, \(Z\in\{0,1\}^J\) denotes the realised random treatment vector, while \(z\in\{0,1\}^J\) denotes a fixed allocation value used in an intervention. Potential processes and causal estimands are indexed by fixed allocations \(z\); uppercase \(Z\) is reserved for observed-data statements such as exchangeability, positivity, and design descriptions.

\subsection{Causal Framework}\label{sec:causal-framework-dag}
To motivate the study of the particular form of spillover we are addressing, we utilise directed acyclic graphs (DAGs) \citep{pearl2009causal} to formally comment on the conditional dependence structure of the random process of interest. The DAG below (\cref{fig:dag_2}) is a discrete-time analogue of the continuous-time point-process model used later. Their role is to explain why, once temporal dependence and cross-cell outcome spillover are present, potential outcomes must be indexed by the full treatment allocation. Formal estimands and causal identification are stated in \cref{sec:causalEstimands,sec:identification}, and \cref{sec:finite-sample-theory} concerns one observed spatiotemporal realisation under a growing-window regime rather than repeated i.i.d.\ copies of these discrete-time nodes.

Fix a tessellation \(\mathcal T=\{\mathcal I_j\}_{j=1}^J\) satisfying Assumption~\ref{ass:tessellation}, with
\(\mathcal I_j=(t^\ast,T]\times \mathcal S_j\). For intuition only, choose grid points \(t^\ast=s_0<s_1<\cdots<s_M=T\), and define the period counts
\(
Y_j^m:=\Nproc((s_{m-1},s_m]\times \mathcal S_j),
\qquad j=1,\ldots,J,\ \ m=1,\ldots,M.
\)
Let \(Z_j\) denote the treatment assigned to cell \(j\). Using Pearl's intervention notation, $\mathrm{do}(Z=z)$ denotes the hypothetical intervention that sets the treatment allocation to $z$, rather than conditioning on the event $Z=z$ in the observed data-generating process \citep{pearl2009causal}. To isolate outcome spillover, we suppress unmeasured confounding and treat any covariate process \(\CovX\) as observed, predictable, and exogenous to the intervention path, so its trajectory is unchanged under \(\mathrm{do}(Z=z)\); see \citet{papadogeorgou2023spatial} for settings with spatial confounding. 
Thus the DAG discussion here isolates outcome spillover rather than the separate problem of treatment confounding.

Figure~\ref{fig:dag_2} shows the essential difficulty in the two-cell case; the general \(J\)-cell case is the obvious extension. Past outcomes in one cell may affect future outcomes in another, so even without confounding the effect of changing \(Z_1\) on a later outcome such as \(Y_1^m\) depends on the entire treatment allocation through the intermediate outcome history. This is the discrete-time analogue of spatiotemporal outcome spillover. Consequently, the relevant potential outcomes are indexed by the full allocation \(z\in\{0,1\}^J\), not just by a single coordinate \(z_j\). Without cross-cell arrows, or with only one post-treatment transition, simpler cell-wise contrasts may sometimes be available; our interest is the generic continuously evolving regime in which that simplification fails.

\begin{figure}[!t]
\centering
\begin{adjustbox}{max width=.75\linewidth,max totalheight=0.32\textheight,keepaspectratio}
\begin{tikzpicture}[
    scale=0.92,
    transform shape,
    ->, >=Stealth,
    node distance=2.1cm,
    every node/.style={font=\small},
    lagged/.style={rectangle, draw, rounded corners, inner sep=2pt}
]

\node (Y1_t1) {Y$_1^{1}$};
\node[lagged] (Y1_t2) [right of=Y1_t1] {Y$_1^{2}$};
\node (Y1_t3) [right of=Y1_t2] {Y$_1^{3}$};
\node (Y1_T)  [right=3cm of Y1_t3] {Y$_1^{T}$};

\node (Y2_t1) [below of=Y1_t1] {Y$_2^{1}$};
\node[lagged] (Y2_t2) [right of=Y2_t1] {Y$_2^{2}$};
\node (Y2_t3) [right of=Y2_t2] {Y$_2^{3}$};
\node (Y2_T)  [right=3cm of Y2_t3] {Y$_2^{T}$};

\node (Z1) [left of=Y1_t1] {$Z_1$};
\node (Z2) [left of=Y2_t1] {$Z_2$};

\draw (Z1) -- (Y1_t1);
\draw (Z1) to [out=25,in=165,looseness=1.5] (Y1_t2);
\draw (Z1) to [out=35,in=155,looseness=1.5] (Y1_t3);

\draw (Z2) -- (Y2_t1);
\draw (Z2) to [out=335,in=205,looseness=1.5] (Y2_t2);
\draw (Z2) to [out=325,in=215,looseness=1.5] (Y2_t3);

\draw (Y1_t1) -- (Y1_t2);
\draw (Y1_t1) -- (Y2_t2);
\draw (Y1_t2) -- (Y1_t3);
\draw (Y1_t2) -- (Y2_t3);

\draw (Y2_t1) -- (Y2_t2);
\draw (Y2_t1) -- (Y1_t2);
\draw (Y2_t2) -- (Y2_t3);
\draw (Y2_t2) -- (Y1_t3);

\draw[dotted] (Y1_t3) -- (Y1_T);
\draw[dotted] (Y2_t3) -- (Y2_T);
\draw[dotted] (Y1_t3) -- (Y2_T);
\draw[dotted] (Y2_t3) -- (Y1_T);

\end{tikzpicture}
\end{adjustbox}
\caption{Discrete-time motivational DAG with temporal dependence and cross-cell outcome spillover. The cross-cell arrows show why later outcomes depend on the full treatment allocation. The boxed lagged outcomes \(Y_1^2\) and \(Y_2^2\) are the post-treatment variables one might be tempted to condition on; doing so blocks some directed spillover paths to later outcomes but can also open noncausal paths through collider structures such as \(Z_1\to Y_1^2\leftarrow Y_2^1\) and \(Z_2\to Y_2^2\leftarrow Y_1^1\).}
\label{fig:dag_2}
\end{figure}
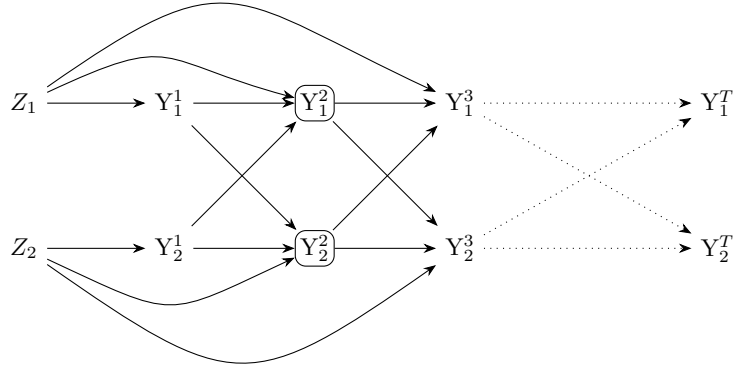

A natural temptation is to condition on lagged outcomes and study quantities such as
\(
P\!\left(Y_1^m\mid \mathrm{do}(Z_1=1),Y_1^{m-1},Y_2^{m-1}\right),
\)
but this generally does not identify the total causal effect. The lagged outcomes are post-treatment descendants of \(Z\), so conditioning on them changes the estimand. In Figure~\ref{fig:dag_2}, the boxed nodes also sit as colliders on paths such as \(Z_1\to Y_1^2\leftarrow Y_2^1\) and \(Z_2\to Y_2^2\leftarrow Y_1^1\). Thus conditioning on lagged outcomes can block some directed spillover paths while simultaneously opening noncausal paths, inducing dependence that in turn obstructs identifiability through exchangeability assumptions (see \cref{sec:identification}) and produces post-treatment bias \citep{elwert2014endogenous}. 
 The same warning applies to a data-adaptive tessellation: if \(\mathcal T\) were constructed from \(\Nproc\) (violating the exogeneity part of Assumption~\ref{ass:tessellation}), then conditioning on \(\mathcal T\) would itself induce selection.

The point-process framework used below replaces these discrete-time nodes by a continuous-history dependence structure. The arrows \(Y_i^{m-1}\to Y_j^m\) in Figure~\ref{fig:dag_2} become dependence of the conditional intensity on the full past, and causal effects are defined through the potential process \(\Nproc^z\) under the full intervention \(\mathrm{do}(Z=z)\). This is why the remainder of the paper works with full-allocation potential processes and likelihood-based estimation of interacting component intensities rather than conditioning on lagged outcomes.

\subsection{Definitions of Causal Estimands}\label{sec:causalEstimands}

For a fixed allocation \(z\in\{0,1\}^J\), let \(\Nproc^z\) denote the post-treatment potential process on \(D\) under the intervention \(\mathrm{do}(Z=z)\), and define the cell-level and total post-treatment counts
\[
Y_j(z):=\Nproc^z(\mathcal I_j),\qquad
Y_\bullet(z):=\Nproc^z(D)=\sum_{j=1}^J Y_j(z).
\]
Because outcome spillover is allowed, \(Y_j(z)\) may depend on the entire allocation vector \(z\), not only on \(z_j\) (and therefore,  implicitly, on $\mathcal{T}$). For a reference allocation \(z\), let \(z^{j=1}\) and \(z^{j=0}\) denote the allocations obtained by setting the \(j\)th coordinate of \(z\) to \(1\) and \(0\), respectively, while leaving all other coordinates unchanged.

\begin{definition}[Individual Treatment Effect (ITE)]
For cell \(\mathcal I_j\), define
\[
\psi_j(z)
:=
\bbE\!\left[Y_j\!\left(z^{j=1}\right)\right]
-
\bbE\!\left[Y_j\!\left(z^{j=0}\right)\right]
=
\bbE\!\left[\Nproc^{z^{j=1}}(\mathcal I_j)\right]
-
\bbE\!\left[\Nproc^{z^{j=0}}(\mathcal I_j)\right].
\]
\end{definition}

\begin{definition}[Average Individual Treatment Effect (AITE)]
\[
\psi(z):=\frac{1}{J}\sum_{j=1}^J \psi_j(z).
\]
\end{definition}
\begin{definition}[Difference under fixed allocations (DAITE)]
For \(z_a,z_b\in\{0,1\}^J\), define
\[
\psi(z_a,z_b)
:=
\bbE\!\left[Y_\bullet(z_a)\right]
-
\bbE\!\left[Y_\bullet(z_b)\right]
=
\bbE\!\left[\Nproc^{z_a}(D)\right]
-
\bbE\!\left[\Nproc^{z_b}(D)\right].
\]
\end{definition}

\begin{definition}[Treatment AITE (TAITE) and DTAITE]
Let \(\mathcal Z\subseteq\{0,1\}^J\) be a finite set of allocations. Define
\[
\psi(\mathcal Z):=\frac{1}{|\mathcal Z|}\sum_{z\in\mathcal Z}\psi(z).
\]
For two finite sets \(\mathcal Z_a,\mathcal Z_b\subseteq\{0,1\}^J\), define
\(
\psi(\mathcal Z_a,\mathcal Z_b):=\psi(\mathcal Z_a)-\psi(\mathcal Z_b).
\)
Equivalently, if \(\pi_{\mathcal Z}\) denotes the uniform distribution on \(\mathcal Z\), then
\[
\psi(\mathcal Z)=\sum_{z}\pi_{\mathcal Z}(z)\psi(z),
\qquad
\psi(\mathcal Z_a,\mathcal Z_b)
=
\sum_z \pi_{\mathcal Z_a}(z)\psi(z)
-
\sum_z \pi_{\mathcal Z_b}(z)\psi(z).
\]
\end{definition}

We retain the symbol \(\psi\) for several related causal estimands and distinguish them by their arguments: \(\psi_j(z)\) is the cell-level ITE, \(\psi(z)\) is the AITE at allocation \(z\), \(\psi(\mathcal Z)\) is the TAITE over the allocation set \(\mathcal Z\), \(\psi(z_a,z_b)\) is the fixed-allocation total-count contrast (DAITE), and \(\psi(\mathcal Z_a,\mathcal Z_b)\) is the corresponding difference in TAITEs (DTAITE).

\subsection{Point process preliminaries and notation}\label{sec:notation}

As the time step in \cref{fig:dag_2} tends to zero, the data-generating process becomes continuous in time. The resulting dependence structure under an arbitrary tessellation \(\mathcal T\) is therefore naturally modelled by a spatiotemporal point process. We now introduce the basic notation.

We continue to write \(D=(t^\ast,T]\times\mathcal S\) for the post-treatment window and \(\dd\tau=\dd t\otimes\dd\mu(x)\) for the base product measure. We model the observed process \(\Nproc\) as a simple counting measure on \(D\), adapted to the observed filtration
\((\Fobs_t)_{t\in[t^\ast,T]}\). Let \(\mathcal B(D)\) denote the Borel \(\sigma\)-algebra on \(D\). For \(A\in\mathcal B(D)\), \(\Nproc(A)\) denotes the number of observed events in \(A\), and \(|A|:=\int_A \dd\tau\). Whenever the \(\Fobs\)-conditional intensity exists, it is the predictable process
\[
\lambda(\tau;\theta\mid\Fobs_{t-})
=\lim_{(h,\delta)\downarrow 0}
\frac{\bbE\big[\Nproc([t,t+h)\times B_\delta(x))\mid\Fobs_{t-}\big]}
{h\,\mu(B_\delta(x))},
\qquad \tau=(t,x)\in D,
\]
where \(B_r(x)\) denotes the Euclidean ball of radius \(r\) centred at \(x\). All intensities are evaluated at pre-jump times (left limits), so \(\lambda(\gamma_i)=\lambda(t_i-,x_i)\). When we later work with the latent component processes \(\Nzero\) and \(\None\), we write \(\lambda_0\) and \(\lambda_1\) for the corresponding component intensities under the filtration then in force, with \(\lambda=\lambda_0+\lambda_1\).

A point process is {\sl simple} if, with probability one, all points are distinct. Since the conditional intensity \(\lambda\) uniquely determines the finite-dimensional distributions of a simple point process (Proposition 7.2.IV of \citet{daley2003introduction}), one typically specifies a simple spatiotemporal point process through a model for \(\lambda\). A point process is {\sl stationary} if its law is invariant under shifts in time and/or space (e.g.\ Definition 12.1.II of \citet{daley2007introduction}).

A canonical history dependent specification is the Hawkes process \citep{Hawkes1,hawkes1974cluster}, the conditional intensity of which takes the form

\begin{equation}\label{eq:canonicalHawkes}
\lambda_k(\tau\mid\Fobs_{t-};\theta)
= \mu_k(\tau;\theta)+\sum_{\ell\in\{0,1\}}\int_{(t-h,t)\times B_R(x)}
g_{k\ell}(t-u,x-y;\theta) \Nproc_\ell(\dd u,\dd y)
\end{equation}
where \(\mu_k(\tau;\theta)\ge 0\) are background intensities, \(g_{k\ell}\ge 0\) are triggering kernels with finite support
\((0,h]\times B_R(0)\), and \(\Nproc_\ell\) denotes the component counting measure.  The \(L^1\) masses of the kernels determine the excitation strength; the standard subcriticality condition is that the associated branching matrix has spectral radius less than \(1\). Hawkes processes exhibit self-excitation and clustering, capturing phenomena common in fields such as seismology, epidemiology, and finance \citep{bacry2015hawkes,reinhart2018review}.

We allow the treatment and control intensities to depend on an observed predictable covariate field \(\CovX(\tau)\in\mathbb R^p\) (e.g., spatial inhomogeneity, seasonality, weather, or population density), writing
\(
\lambda_k(\tau\mid\Fobs_{t-};\theta)
=
\lambda_k(\tau\mid\Fobs_{t-},\CovX;\theta).
\)
Unless stated otherwise, we treat \(\CovX\) as exogenous with respect to treatment, so that under \(\mathrm{do}(Z=z)\) the covariate path is unchanged. Consequently, all likelihoods and causal estimands are understood as conditional on \(\CovX\) (and may be marginalised over \(\CovX\) if desired). Covariates that are themselves affected by treatment would generally change the causal target (e.g., to a controlled direct effect) \citep{robins1992identifiability}, so we do not use post-treatment mediators as adjustment variables. For example, in the model corresponding to \cref{eq:canonicalHawkes}, one may specify
\(
\mu_k(\tau;\theta)=\exp\{\eta_k+\beta_k^\top \CovX(\tau)\}. 
\)
A multiplicative specification of the form \(\mu_k(\tau;\theta)=\lambda_{0k}(\tau)\exp\{\beta_k^\top \CovX(\tau)\}\) is the point-process analogue of a Cox proportional hazards model, with \(\lambda_{0k}\) playing the role of a baseline hazard/intensity and the exponential term encoding proportional covariate effects. More generally, one may instead use an additive model, spline basis, or another flexible specification.
\section{Identification}\label{sec:identification}

We now state the assumptions under which the causal estimands in \cref{sec:causalEstimands} are identified from the observed law. Because outcome spillover is allowed, potential outcomes must be indexed by the entire treatment vector rather than by a single cell-wise treatment status. The assumptions below are the standard consistency/exchangeability/positivity ingredients adapted to the present interference setting, together with a regime-stable point-process model for extrapolation to off-support interventions \citep{pearl2009causal,robins1992identifiability}.

Let \(\mathcal Z_{\mathrm{tar}}\subseteq \{0,1\}^J\) denote the set of treatment allocations relevant to the scientific question (or, more generally, the support of the policies under comparison). For each \(z\in\mathcal Z_{\mathrm{tar}}\), let \(\Nproc^z\) denote the post-treatment potential point process that would be observed on \((t^\ast,T]\times\mathcal S\) under the intervention \(\mathrm{do}(Z=z)\), and define the corresponding cell counts \(Y_j(z):=\Nproc^z(\mathcal I_j)\), \(j=1,\ldots,J\).

The intervention fixes the cell-wise allocation \(z\); it does not fix or reveal the latent component labels of post-treatment events. Under \(\mathrm{do}(Z=z)\), the post-treatment process remains a superposition
\(
\Nproc^z=\Nzero^z+\None^z,
\)
and the allocation field \(z\) may affect both component intensities through the post-treatment dynamics. In particular, treated cells need not contain only \(\None^z\)-events, and control cells need not contain only \(\Nzero^z\)-events. Spatial dispersion and excitation may place \(\Nzero^z\)-events in treated cells and \(\None^z\)-events in untreated cells. Thus the intervention acts on the regime, not on a cell-wise selector that turns one component off.

By construction, \(\bbE[\Nproc^z(B)] = \bbE[\Nproc(B)\mid \mathrm{do}(Z=z)]\) for every Borel \(B\subseteq (t^\ast,T]\times\mathcal S\). No no-interference restriction is imposed: \(\Nproc^z\) may depend on the full allocation vector \(z\) and on the entire post-treatment history. Throughout this section, define the pre-intervention information
\(
\mathcal B_0:=\sigma(\mathcal T,\CovX)\vee \Fobs_{t^\ast-},
\)
that is, the \(\sigma\)-field generated by the tessellation, the exogenous covariate path, and the observed history up to \(t^\ast\).

\begin{assumption}[Well-defined interventions and consistency]\label{ass:id-consistency}
For every \(z\in\mathcal Z_{\mathrm{tar}}\), the intervention \(\mathrm{do}(Z=z)\) defines a unique post-treatment potential process \(\Nproc^z\). Moreover, if the realised treatment vector equals \(z\), then
\[
\Nproc(B)=\Nproc^z(B)\qquad\text{a.s. for every Borel }B\subseteq (t^\ast,T]\times\mathcal S.
\]
The pre-treatment history \(\Fobs_{t^\ast-}\), the tessellation \(\mathcal T\), and the exogenous covariate path \(\CovX\) are common across \(\{\Nproc^z:z\in\mathcal Z_{\mathrm{tar}}\}\).
\end{assumption}

\begin{assumption}[Conditional exchangeability]\label{ass:id-exchange}
\[
\{\Nproc^z:z\in\mathcal Z_{\mathrm{tar}}\}\ \perp\!\!\!\perp\ Z\ \big|\ \mathcal B_0.
\]
That is, after conditioning on \(\mathcal B_0\), the realised treatment allocation carries no additional information about the collection of potential post-treatment processes.
\end{assumption}

\begin{assumption}[Positivity on the design support]\label{ass:id-positivity}
There exists a subset \(\mathcal Z_{\mathrm{supp}}\subseteq \mathcal Z_{\mathrm{tar}}\) such that, for every \(z\in\mathcal Z_{\mathrm{supp}}\),
\(
\bbP(Z=z\mid \mathcal B_0)>0\qquad\text{a.s.}
\)
If the target of interest is a randomised policy \(\pi\) over allocations, it suffices that \(\mathrm{supp}(\pi)\subseteq \mathcal Z_{\mathrm{supp}}\).
\end{assumption}

Assumptions~\ref{ass:id-consistency}--\ref{ass:id-positivity} identify causal quantities only for allocations in the conditional support of the treatment mechanism. In particular, we do not require full support over all \(2^J\) allocations. The next proposition is the point-process analogue of the usual g-formula.

\begin{proposition}[Identification on the observed design support]\label{prop:id-onsupport}
Suppose Assumptions~\ref{ass:id-consistency}--\ref{ass:id-positivity} hold. Then for every \(z\in\mathcal Z_{\mathrm{supp}}\) and every Borel set \(B\subseteq (t^\ast,T]\times\mathcal S\),
\(
\bbE\!\left[\Nproc^z(B)\right]
=
\bbE\!\left[\bbE\!\left\{\Nproc(B)\mid Z=z,\mathcal B_0\right\}\right].
\)
Equivalently, for each cell \(\mathcal I_j\),
\(
\bbE\!\left[Y_j(z)\right]
=
\bbE\!\left[\bbE\!\left\{\Nproc(\mathcal I_j)\mid Z=z,\mathcal B_0\right\}\right].
\)
\end{proposition}

For a point process, the natural mean object under allocation \(z\) is the map
\[
B \mapsto \bbE[\Nproc^z(B)],
\]
which gives the expected number of post-treatment events falling in any Borel region \(B \subseteq (t^\ast,T]\times\mathcal S\). We refer to this expected-count map as the interventional mean measure. Proposition~\ref{prop:id-onsupport} therefore identifies this interventional mean measure, equivalently the family of expected cell counts \(\{\bbE[Y_j(z)]\}_{j=1}^J\), rather than the full interventional law of \(\Nproc^z\). Proposition~\ref{prop:id-onsupport} is nonparametric, but it applies only to allocations that appear with positive probability under the observed treatment mechanism. Many scientifically relevant contrasts, however, involve off-support allocations, such as the all-treated versus all-untreated regimes considered later. Those contrasts are not nonparametrically identified from observational data without further structure. Our approach is therefore explicitly model-based for off-support regimes.

To formalise this, let
\[
\mathcal F_t^z
:=
\sigma\!\Big(
\mathcal B_0,\,
\Nproc^z\big((t^\ast,s]\times A\big):s\le t,\ A\in\mathcal B(\mathcal S)
\Big),
\qquad t\in[t^\ast,T],
\]
be the filtration generated by \(\mathcal B_0\) and the post-treatment potential history under allocation \(z\).

\begin{assumption}[Common structural model across regimes]\label{ass:id-structural}
There exists \(\theta^\star\in\Theta\) such that, for every \(z\in\mathcal Z_{\mathrm{tar}}\), the potential process \(\Nproc^z\) admits latent component processes \((\Nzero^z,\None^z)\) with
\(
\Nproc^z=\Nzero^z+\None^z,
\)
and predictable component intensities
\(
\lambda_{k,\theta^\star}^z(\tau\mid \mathcal F_{t-}^z,\CovX)\) for \( k\in\{0,1\},
\)
so that the total intensity is
\[
\lambda_{\theta^\star}^z(\tau\mid \mathcal F_{t-}^z,\CovX)
=
\lambda_{0,\theta^\star}^z(\tau\mid \mathcal F_{t-}^z,\CovX)
+
\lambda_{1,\theta^\star}^z(\tau\mid \mathcal F_{t-}^z,\CovX).
\]
Equivalently, the family \(\{\Nproc^z:z\in\mathcal Z_{\mathrm{tar}}\}\) is generated by a common parametric model \(\{\lambda_\theta^z:\theta\in\Theta\}\), with the same true parameter value \(\theta^\star\) across all regimes \(z\). Under the observed allocation \(Z\), the observed post-treatment process is generated by the corresponding regime \(z=Z\).
\end{assumption}

\begin{assumption}[Identifiability of the structural parameter]\label{ass:id-theta}
If \(\theta,\theta'\in\Theta\) satisfy
\[
\mathcal L_\theta\!\left(\Nproc\mid Z=z,\mathcal B_0\right)
=
\mathcal L_{\theta'}\!\left(\Nproc\mid Z=z,\mathcal B_0\right)
\qquad\text{a.s. for every } z\in\mathcal Z_{\mathrm{supp}},
\]
then \(\theta=\theta'\).
\end{assumption}

Under Assumption~\ref{ass:id-structural}, each regime \(z\) induces an interventional mean measure
\begin{equation}\label{eq:id-mean-measure}
\Lambda_\theta^z(B\mid \mathcal B_0)
:=
\bbE_\theta\!\left[
\int_B \lambda_\theta^z(\tau\mid \mathcal F_{t-}^z,\CovX)\,\dd\tau
\,\Big|\, \mathcal B_0
\right],
\qquad B\subseteq (t^\ast,T]\times\mathcal S.
\end{equation}
By the compensator identity, \(\Lambda_\theta^z(B\mid \mathcal B_0)=\bbE_\theta[\Nproc^z(B)\mid \mathcal B_0]\).

\begin{proposition}[Model-based identification of off-support regimes]\label{prop:id-offsupport}
Suppose Assumptions~\ref{ass:id-consistency}, \ref{ass:id-exchange}, \ref{ass:id-structural}, and \ref{ass:id-theta} hold, and that \(\theta^\star\) is identified from the observed law on \(\mathcal Z_{\mathrm{supp}}\). Then for every \(z\in\mathcal Z_{\mathrm{tar}}\) and every Borel set \(B\subseteq (t^\ast,T]\times\mathcal S\),
\[
\bbE\!\left[\Nproc^z(B)\mid \mathcal B_0\right]
=
\Lambda_{\theta^\star}^z(B\mid \mathcal B_0),
\qquad
\bbE\!\left[\Nproc^z(B)\right]
=
\bbE\!\left[\Lambda_{\theta^\star}^z(B\mid \mathcal B_0)\right].
\]
Hence every causal estimand in \cref{sec:causalEstimands} that can be written as a finite linear combination of \(\{\bbE[\Nproc^z(B)]: z\in\mathcal Z_{\mathrm{tar}},\ B\subseteq (t^\ast,T]\times\mathcal S\}\) is identified by \(\theta^\star\).
\end{proposition}

The distinction between Propositions~\ref{prop:id-onsupport} and \ref{prop:id-offsupport} is central. For allocations \(z\in\mathcal Z_{\mathrm{supp}}\), interventional mean measures are identified nonparametrically under Assumptions~\ref{ass:id-consistency}--\ref{ass:id-positivity}. For allocations outside the design support, identification is available only relative to the regime-stable structural model in Assumption~\ref{ass:id-structural}, so such targets should be interpreted as model-based causal extrapolations. As discussed in \cref{sec:notation}, for a simple spatiotemporal point process the conditional intensity determines the law of the process, which means these off-support effects can be operationalised by evaluating the regime-indexed intensity family under the target allocation.

In randomised experiments, Assumptions~\ref{ass:id-exchange} and \ref{ass:id-positivity} are properties of the design. In observational settings, Assumption~\ref{ass:id-exchange} requires that \(\mathcal B_0\) include all common causes of \(Z\) and \(\{\Nproc^z\}\); otherwise causal identification fails. The identification results above are superpopulation statements for \((\mathcal B_0,Z,\{\Nproc^z\}_{z\in\mathcal Z_{\mathrm{tar}}})\). 
The estimation theory below instead conditions on the realised allocation and studies recovery of the common structural parameter \(\theta^\star\) from one expanding realisation of the observed regime.  This distinction matters because model-based identification of off-support interventional mean measures does not, by itself, guarantee that \(\theta^\star\) is estimable from a single realisation; for history dependent models, that requires the additional compatibility conditions developed in \cref{sec:finite-sample-theory}.  Once \(\theta^\star\) is recovered, the mean measures \(\Lambda_{\theta^\star}^z\) and the causal estimands built from them are obtained by plug-in evaluation.

\section{Estimation}\label{sec:estimation}

The causal estimands in \cref{sec:causalEstimands} are defined in terms of the underlying point process intensities under different treatment assignments. In practice, these intensities depend on unknown parameters $\theta$ that must be learned from the observed, unlabelled point pattern. The main challenge is that after time $t^\ast$ we do not know, for each event, whether it arose from the treatment process $N_1$ or the control process $N_0$. 

We emphasise that cell treatment and event labels are distinct objects. The treatment vector \(Z\) (or fixed allocation \(z\)) specifies the intervention regime for cells, whereas the latent variable \(r_i\in\{0,1\}\) records which component process generated event \(\gamma_i\). Because spillover is allowed, an event in a treated cell need not have label \(1\), and an event in an untreated cell need not have label \(0\).

Our strategy is to treat the control and treatment point labels as latent variables and to work in a parametric family $\{\lambda_k(\cdot\mid\mathcal H_{t-};\theta):\theta\in\Theta\}$ for $k\in\{0,1\}$. If we knew the labels, standard point process likelihood theory would let us estimate $\theta$ efficiently by maximising the complete-data log-likelihood \citep{ogata1978asymptotic}. We therefore employ a type of EM algorithm that alternates between (i) updating the labels to better match the current fitted intensities, and (ii) updating the parameters given the current labels. This section explains how the complete-data likelihood is constructed and how, in principle, the EM updates operate; computational refinements are deferred to \cref{sec:stochEM}. Throughout Sections~\ref{sec:estimation}--\ref{sec:finite-sample-theory}, we condition on the realised post-treatment allocation and suppress that dependence in the notation, writing \(\lambda_k(\cdot\mid\mathcal H_{t-};\theta)\) for the observed-regime component intensities. The full regime-indexed family \(\{\lambda_{k,\theta}^z\}\) reappears only when we evaluate counterfactual mean measures \(\Lambda_\theta^z\).

Throughout, we maintain a parametric specification for the conditional intensities, primarily for feasibility of estimation and theory. This is in line with recent causal work in complex spatial and networked settings, such as \citet{papadogeorgou2022causal}. Extending our framework to fully nonparametric or Bayesian nonparametric intensity models, for example along the lines of \citet{rousseau2025estimation}, is an important avenue for future research. 
In the parametric examples below, the observed-regime intensities \(\lambda_k(\cdot\mid \mathcal H_{t-};\theta)\) should be understood as the specialization of a regime-indexed family \(\lambda_{k,\theta}^z\), where the cellwise allocation \(z\) enters through the post-treatment treatment field and hence through the component intensities and their induced histories.

Assuming a parametric framework, estimation of the causal estimands in \cref{sec:causalEstimands} relies on consistently estimating model parameters. Our approach tackles this by attempting to identify spillover across cells of different treatment assignment. Post-treatment, the observed process \(\Nproc=\Nzero+\None\) is unlabelled; let \(r=(r_i)_{i=1}^{\Nproc(D)}\in\{0,1\}^{\Nproc(D)}\) be the latent labels of events \(\gamma_i\). Here we take \(D=(t^\ast,T]\times\mathcal S\) as the post-treatment analysis window and interpret likelihoods conditional on the observed history up to \(t^\ast\); see \cref{sec:setup} for the predictability conventions. Given a candidate labelling \(r\), let \(\widehat\lambda_k^r(\tau;\theta)\) denote the component intensities obtained by feeding only past labels \(r_{1:\Nproc(t-)}\) into the predictable integrands. The complete-data maximiser
\begin{equation}\label{eq:cdll-main}
\tilde\theta_r:=\arg\max_{\theta\in\Theta}\ell_r(\theta)
\end{equation}
is well defined under standard regularity \citep{ogata1978asymptotic,RATHBUN199655}.

The complete-data log-likelihood \eqref{eq:cdll-main} induces the marginal
(observed-data) likelihood
\begin{equation}\label{eq:obsLike}
L(\theta)\equiv p(\Nproc\mid \theta)
=\sum_{r\in\{0,1\}^{\Nproc(D)}} p(\Nproc,r\mid \theta)
=\sum_{r}\exp\{\ell_r(\theta)\}.
\end{equation}
The corresponding maximum likelihood estimator (MLE) is the asymptotically optimal benchmark under standard regularity for parametric point process likelihoods \citep{ogata1978asymptotic}. The difficulty is computational: \eqref{eq:obsLike} sums over \(2^{\Nproc(D)}\) labellings. In Poisson-like (independent-increments) regimes, \(\ell_r(\theta)\) factorises over events and \(f(r\mid \Nproc,\theta)\) becomes a product of independent Bernoulli terms, recovering the familiar mixture-model structure \citep{mclachlan2019finite}. Under history dependence, however, a single post-\(t^\ast\) flip perturbs future intensities, so neither \(\ell_r(\theta)\) nor \(f(r\mid \Nproc,\theta)\) factorises, and exact evaluation/optimization of \(\log L(\theta)\) is intractable at finite computing budgets.

Before \(t^\ast\), only \(\Nzero\) contributes, so the control parameters of the kernel are estimable directly from the observed pre-treatment history if \(\int_0^{t^\ast} \mathrm dN_0(t)\) is sufficiently large. Spillover arises only after \(t^\ast\): events from \(\Nzero\) can occur inside treated cells and vice-versa. Our strategy is to search over plausible labellings \(r\) of post-\(t^\ast\) events and pick those that maximise a principled objective such as a likelihood criterion, using the fitted \(\widehat\lambda_k^r\).

The computational dichotomy above also has a simple interpretation in the DAG picture of \cref{fig:dag_2}.  In an independent-increments (Poisson) specification for the component processes, the ``spillover'' arrows that cross between rows in \cref{fig:dag_2} (e.g., $Y_1^{t-1}\to Y_2^t$ and $Y_2^{t-1}\to Y_1^t$) are absent: conditional on the treatment vector $Z$ (and any exogenous covariates driving inhomogeneity), counts in disjoint spatiotemporal regions do not causally propagate through the process history.  In that case, the post-$t^\ast$ latent-label problem is best viewed as classification in a superposition (mixture) of Poisson components, and many targets of interest depend only on recovering the component intensities in expectation (i.e., getting the compensators / expected counts right), rather than on perfectly resolving each individual event label.  In this sense, spillover in the Poisson case is primarily descriptive (how intensity mass is apportioned), not a generative event-chain mechanism.

By contrast, under genuine history dependence (as in Hawkes-type models) the cross cell arrows in \cref{fig:dag_2} represent real causal propagation: a single early label flip changes future conditional intensities and hence couples the entire label configuration. This is precisely the regime in which replacing integration over ``unobservables'' by hard maximisation can result in inconsistent estimates (cf.\ \citet{meng2010s}). 

Because \eqref{eq:obsLike} is intractable under history dependence, we work within the expectation--maximisation (EM) framework to connect estimation to the MLE target. Exact EM requires expectations with respect to the label posterior \(f(r\mid \Nproc,\theta)\), which is itself intractable when a single label flip perturbs future intensities. We therefore use a stochastic approximation to EM (SEM), described in \cref{sec:stochEM}, which estimates the E-step using Monte Carlo samples of labellings. For our finite-sample analysis, we study an idealised penalised classification (hard-EM) update as a proof device; this deterministic operator makes explicit the conditions under which label updates contract and is analysed in \cref{sec:finite-sample-theory}.

\subsection{Stochastic Expectation maximisation}\label{sec:stochEM}

In theory, we could traverse the space of labellings in an informed manner using a Metropolis--Hastings style algorithm \citep{hastings1970monte}. However, in practice, as each likelihood must be recalculated at each step and the initial labelling is typically far from high-likelihood labellings, we have found this to be impractical. We therefore turn to SEM as a computationally tractable alternative.

A standard EM algorithm to compute maximum likelihood estimates in this context iterates between expectation and maximisation steps as follows: the expectation step given current parameter vector $\theta'$ can be computed as
	\begin{align}
	Q(\theta,\theta') 
	= \bbE_{r \mid N, \theta'}\!\left[ \log p(N,r \mid \theta)\right]
	= \sum_{r} \log p(N,r \mid \theta)\, p(r \mid N, \theta')  \label{eq:Q}
	\end{align}
	where the expectation is with respect to the conditional distribution of the labels $r$ given the observed data $N$ under parameter value $\theta'$.
	Next the maximisation step is defined via the updated estimate
	\(
	\hat\theta = \arg\max_\theta Q(\theta,\theta').
	\)
    
Under mild conditions, the EM algorithm converges monotonically to the MLE of $\theta$. Approximate standard errors for \(\hat\theta\) can be obtained from the observed information, estimated via Louis' identity \citep{louis1982finding} using the same importance-sampling label draws as in the E-step. This yields an estimated covariance matrix for \(\hat\theta\) that may be used for Wald-type uncertainty summaries and for propagating uncertainty to causal estimands below. Formal asymptotic justification for these approximations in the present growing-window, history dependent setting is left for future work.

However, $Q$ cannot be computed from \eqref{eq:Q} directly, as (i) we do not have $p(r \mid N, \theta')$ in closed form; and (ii) the sum over $r$ runs over an enormous space of labellings. Conceptually, if exact draws from \(p(r\mid N,\theta')\) were available, we would approximate the E-step by
\begin{equation}
Q^*(\theta,\theta') = \frac{1}{M}\sum_{i=1}^M \log p(N,r^{(i)}\mid \theta), \label{eq:Qstar}
\end{equation}
where $\{r^{(i)}\colon i =1,\ldots,M\}$ is a random sample from $p(r \mid N, \theta')$.

In our setting, we cannot sample directly from \(p(r\mid N,\theta')\), so we introduce a proposal distribution \(q(r\mid N)\) from which sampling is easy. Draw \(r^{(1)},\ldots,r^{(M)}\) independently from \(q(r\mid N)\). We then approximate the EM objective by
\begin{equation}
Q^\dagger(\theta,\theta') = \sum_{i=1}^M w_i \log p(N,r^{(i)}\mid \theta), \label{eq:Qdagger}
\end{equation}
where
\begin{align}
w_i \propto \frac{p(r^{(i)}\mid N,\theta')}{q(r^{(i)}\mid N)} 
    \propto \frac{p(N,r^{(i)}\mid \theta')}{q(r^{(i)}\mid N)}, \label{eq:w}
\end{align}
and the constant of proportionality is chosen so that \(\sum_{i=1}^M w_i=1\). The importance weights \eqref{eq:w} scale like
\(
w_i \propto \exp\{\ell_{r^{(i)}}(\theta')\}/q(r^{(i)}\mid \Nproc).
\)
Thus \eqref{eq:Qdagger} is accurate only if the proposal \(q(r\mid \Nproc)\) places non-negligible mass on the high-posterior region of \(p(r\mid \Nproc,\theta')\). In independent-increments (factorizable) settings this overlap requirement is mild, because posterior mass decomposes across events. Under strong history dependence design of a suitable proposal distribution $q$ becomes the central algorithmic constraint. Typically only a tiny fraction of the possible \(2^{\Nproc(D)}\) label configurations will carry non-negligible probability mass (for $\theta$ in the vicinity of the maximum likelihood estimate), so that a naive choice like a uniform proposal distribution will fail entirely in practice. Exact use of \eqref{eq:Qdagger} requires the proposal mass \(q(r\mid N)\) to be fully specified and evaluable. The discrepancy-guided proposal used below is best viewed as a stochastic search proposal; when its exact mass is not computed, the resulting algorithm is an approximate SEM / stochastic hill-climb scheme rather than exact self-normalised importance sampling.

Conceptually, one can view SEM as requiring a proposal \(q(r\mid N)\) that concentrates on high-posterior label configurations under the current iterate \(\theta'\). This can be achieved by refining \(q\) adaptively during a preliminary stage of the algorithm. The simplest limiting case is to concentrate on a single labelling at each iteration; i.e.\ set $M=1$ in (\ref{eq:Qstar}). This leads to a hard (or classification) EM algorithm \citep{celeux1992classification}, updates for which are analysed in \cref{sec:finite-sample-theory}. In practice, the discrepancy-guided proposal in \cref{sec:pracImpSEM} is designed to mimic a greedy hard-EM relabelling step by focusing Monte Carlo effort on regions where the fitted components disagree most.

\section{Finite-sample theory}\label{sec:finite-sample-theory}

The identification results in \cref{sec:identification} are population-level statements about interventional mean measures indexed by a common structural parameter \(\theta^\star\). The question here is narrower: from one growing realisation of the observed process, when can the latent-label likelihood procedure recover that parameter well enough for plug-in causal estimation? Much of the material in this section is technical, and many of the supporting results and underlying assumptions are described in full detail in the Supplementary Material. Our aim here is to present the critical results and comment on their practical significance.

Our practical estimator is stochastic EM. For theory, however, we analyse the predictable blockwise hard-EM surrogate developed in \cref{sec:perflip}. \cref{sec:setup}--\cref{sec:functional-transfer} develops the proof machinery; here we state the resulting recursion and interpret its ingredients. At a high level, the assumptions ensure predictable blockwise updates, control the effect of local label mistakes and martingale fluctuation, require enough signal away from ambiguous regions, and give the M-step enough local curvature to turn smaller label error into smaller parameter error.

We begin by introducing the blockwise notation entering \cref{thm:sec-contraction}.  Fix the deterministic update schedule
\(
t^\ast=u_0<u_1<\cdots<u_{\Tmax}=T\),
\(
D_m:=(u_{m-1},u_m]\times\mathcal S\),
\(
D^{(m)}:=(t^\ast,u_m]\times\mathcal S.
\)
After the first \(m\) blocks have been processed, the hard-EM proof device has produced a label vector \(r^{(m)}\) for the events in \(D^{(m)}\) and a parameter iterate \(\theta^{(m)}\). Let \(r^\star\) denote the true latent post-treatment labelling. The restricted Hamming distance
\[
d_H^{[u_m]}(r^{(m)},r^\star)
=
\sum_{i:\,t_i\le u_m}\1\{r_i^{(m)}\neq r_i^\star\}
\]
counts label errors among events observed up to time \(u_m\), and the corresponding prefix mislabelling rate is
\[
e_m:=\frac{d_H^{[u_m]}(r^{(m)},r^\star)}{|D^{(m)}|},
\qquad e_0:=0.
\]
We also write
\[
\omega_m:=\frac{|D^{(m)}|}{|D^{(m+1)}|},
\]
so \(\omega_m\) is the fraction of the next prefix window coming from already-processed blocks.

Define the oracle log-likelihood ratio
\[
s^{\rm or}_{\theta^\star}(\tau)=\log\frac{\lambda_1^\star(\tau)}{\lambda_0^\star(\tau)},
\]
the decisive sets
\(
S_{\theta^\star}^+(b):=\{\tau:\ s^{\rm or}_{\theta^\star}(\tau)\ge b\}\),
\(
S_{\theta^\star}^-(b):=\{\tau:\ s^{\rm or}_{\theta^\star}(\tau)\le -b\},
\)
and the ambiguous band
\(
A_{\theta^\star}(2b):=\{\tau:\ |s^{\rm or}_{\theta^\star}(\tau)|\le 2b\}.
\)
On \(S_{\theta^\star}^+(b)\) the oracle favours label \(1\), and on \(S_{\theta^\star}^-(b)\) it favours label \(0\); a \emph{minority label} means the opposite choice. More generally, \(S_\theta^\pm(b)\) denotes the same construction at a generic parameter \(\theta\). Thus the hard E-step condition used below says that, on the newly processed block, the greedy rule chooses label \(1\) on \(S_{\theta^{(m)}}^+(b)\) and label \(0\) on \(S_{\theta^{(m)}}^-(b)\).

Two constants appear repeatedly. The first is the score-alignment slack \(\Delta_s\) from Assumption~\ref{ass:G4}, which controls the gap between the oracle log-ratio and the label-induced log-ratio. The second is the local log likelihood ratio (LLR) Lipschitz constant \(L_s\) from Assumption~\ref{ass:G5}, which controls how quickly decisive sets move as \(\theta\) changes. The penalty level \(\alpha\) used in the hard E-step is chosen so that sufficient condition \eqref{eq:EsufficientB} holds.

The blockwise floor term is
\[
\varepsilon_{b,m}
:=
\eta_+\bigl(|D_{m+1}|;b\bigr)
+
\eta_-\bigl(|D_{m+1}|;b\bigr)
+
\frac{2}{|D_{m+1}|}\,
N\!\big(A_{\theta^\star}(2b)\cap D_{m+1}\big),
\]
where \(\eta_\pm(\cdot;b)\) are the small oracle-minority rates from Assumption~\ref{ass:G8}. The E-step and M-step fluctuation terms are
\[
\delta_m
:=
C_{\rm H,1}\sqrt{\frac{K_{\rm win}z_{\rm Fr}}{|D_{m+1}|}},
\qquad
\xi_m
:=
\frac{C_{\rm sc}}{m_{\rm sc}}
\left(
\sqrt{\frac{K_{\rm win}z_{\rm Fr}}{|D^{(m)}|}}
+
\frac{z_{\rm Fr}}{|D^{(m)}|}
\right),
\]
where \(K_{\rm win}\) is the window-envelope constant from Assumption~\ref{ass:G2prime}, \(z_{\rm Fr}\) is the Freedman inequality/union-bound level from the Supplement, and \(m_{\rm sc}\) is the local strong-concavity modulus from Assumption~\ref{ass:G6}. In the Hawkes process verification \(K_{\rm win}=O(\log|D|)\) and \(z_{\rm Fr}=\Theta(\log|D|)\), so
\[
\delta_m\asymp \sqrt{\frac{K_{\rm win}\log |D|}{|D_{m+1}|}},
\qquad
\xi_m\asymp \sqrt{\frac{K_{\rm win}\log |D|}{|D^{(m)}|}}+\frac{\log|D|}{|D^{(m)}|}.
\]

\begin{theorem}[Contraction to a statistical floor]\label{thm:sec-contraction}
Assume that Assumptions~\ref{ass:G0},~\ref{ass:G0pred},~\ref{ass:G1}--\ref{ass:G9}, and~\ref{ass:warmstart} hold, and choose \(b\) and \(\alpha\) so that the sufficient condition \eqref{eq:EsufficientB} holds. In particular, on each newly processed block the greedy E-step assigns no minority labels on decisive regions relative to the current iterate \(\theta^{(m)}\). Let
\[
A:=\frac{L_s}{b},
\]
where \(L_s\) is the LLR Lipschitz constant from Assumption~\ref{ass:G5}. Then, on the same high-probability event as \cref{thm:contract}, there exists a constant \(B_{\rm main}>0\) (the M-step sensitivity constant derived in the Supplement) such that, for every block \(m\),
\[
e_{m+1}
\le
\omega_m e_m
+
(1-\omega_m)\Big(A\|\theta^{(m)}-\theta^\star\|+\varepsilon_{b,m}+\delta_m\Big),
\]
and
\(
\|\theta^{(m+1)}-\theta^\star\|
\le
B_{\rm main}\,e_{m+1}+\xi_{m+1}.
\)
If \(\omega:=\max_m\omega_m<1\) and
\(
\rho:=\omega+(1-\omega)AB_{\rm main}<1,
\)
then the label error decays geometrically up to a floor of order
\[
\frac{1}{1-\rho}\max_m\big(\varepsilon_{b,m}+\delta_m+A\xi_m\big),
\]
and the parameter error decays geometrically up to the induced floor
\[
B_{\rm main}\frac{1}{1-\rho}\max_m\big(\varepsilon_{b,m}+\delta_m+A\xi_m\big)+\max_m\xi_m.
\]
\end{theorem}

Because past blocks are frozen, the first inequality is an averaging step: old prefix error is carried forward with weight \(\omega_m\), and the new block contributes fresh error of size
\(
A\|\theta^{(m)}-\theta^\star\|+\varepsilon_{b,m}+\delta_m.
\)
Here \(\varepsilon_{b,m}\) is the intrinsic statistical obstruction: small oracle-minority mass plus the mass of the ambiguous band on the new block. The term \(\delta_m\) is the E-step fluctuation. The second inequality is the M-step: once labels are mostly right, local strong concavity pulls the likelihood maximiser toward the truth, up to the fluctuation term \(\xi_{m+1}\). Thus the residual finite-sample floor is not a generic optimisation pathology; it is the combination of genuinely low-signal mass and the fluctuation terms that remain even when the iterate is already close to \(\theta^\star\). The role of \(A=L_s/b\) is especially transparent. Larger \(b\) makes the decisive regions more stable, hence improves the contraction factor, but it also widens the region treated as ambiguous through \(A_{\theta^\star}(2b)\). 

Although the data are observed offline, this section works with a predictable blockwise operator so that the martingale arguments are valid. This should be read as a proof device: once the schedule and tie-breaking rules are fixed, the same iterates can be replayed pathwise from a stored event list. The point is to capture the mechanism that a successful SEM implementation must approximate, without trying to analyse the Monte Carlo part of the practical algorithm itself.

\subsection{When does the floor vanish?}\label{subsec:regimes}

Finite-sample contraction does not by itself imply consistency.  All consistency statements in this subsection are understood along a triangular-array increasing-domain regime \(D_n\), with the threshold \(b_n\) and, in the Hawkes process case, the associated alignment, curvature, and band-mass quantities allowed to vary with \(n\). The point is that consistency is an array-level consequence of vanishing floor terms, not a generic fixed-parameter corollary of \cref{thm:sec-contraction}. To pass from \cref{thm:sec-contraction} to consistency, consider a triangular array of growing windows \(D_n\) with \(|D_n|\to\infty\), and let all quantities above carry an \(n\)-index. In particular, \(b_n\) is the decisiveness threshold, \(\Delta_{s,n}\) is the score-alignment slack from Assumption~\ref{ass:G4}, \(L_{s,n}\) is the LLR Lipschitz constant from Assumption~\ref{ass:G5}, \(m_{{\rm sc},n}\) is the local strong-concavity modulus from Assumption~\ref{ass:G6}, and \(B_{{\rm main},n}\) is the corresponding M-step sensitivity constant.

The same tradeoff described in the Supplement remains: \(b_n\) should shrink so that the ambiguous band becomes negligible, but not so fast that score alignment and contraction break down.

\begin{proposition}[Compatibility conditions for consistency]\label{prop:compatibility}
Let the quantities in \cref{thm:sec-contraction} carry an \(n\)-index, with \(|D_n|\to\infty\). If
\(
b_n\downarrow 0\), \(
\Delta_{s,n}=o(b_n)\), \(
L_{s,n}B_{{\rm main},n}=o(b_n),
\)
and
\(
\max_m \varepsilon_{b_n,m}\xrightarrow{\mathbb P}0\), \(
\max_m \delta_{m,n}\xrightarrow{\mathbb P}0\), \(
\max_m \xi_{m,n}\xrightarrow{\mathbb P}0\),
with \(\sup_n\max_m \omega_{m,n}<1\), then
\(
\|\theta_n^{(\Tmax_n)}-\theta_n^\star\|\xrightarrow{\mathbb P} 0.
\)
For Hawkes models, a sufficient route is that
\[
\varepsilon_{{\rm col},n}\log |D_n|=o(b_n),
\qquad
\frac{L_{s,n}\log|D_n|}{m_{{\rm sc},n}}=o(b_n),
\]
where \(\varepsilon_{{\rm col},n}\) is the column-similarity error from the Hawkes process alignment verification, together with an array-level vanishing-band condition derived from \cref{lem:regA,lem:regB}. Without such array-level scaling, the Hawkes process theory generally gives contraction to a non-zero floor rather than automatic consistency. Note that along the triangular array we allow the true data-generating parameter to vary with \(n\), writing \(\theta_n^\star\). In the fixed-parameter case, \(\theta_n^\star\equiv\theta^\star\).
\end{proposition}

Thus \(b_n\) has two jobs at once: it must be small enough to make the ambiguous band negligible, but large enough relative to the alignment error \(\Delta_{s,n}\) and the contraction scale \(L_{s,n}B_{{\rm main},n}\) to preserve decisive-set stability.

Two cases are worth separating. In Poisson models there is no feedback, so a label mistake does not alter future intensities. In that regime the hard-EM floor is largely a classification issue, and it should not be read as a necessity result for SEM: because mislabelling does not propagate, averaging over labels can still recover the right parameter even when hard assignments plateau. In Hawkes models the situation is sharper: one wrong early label perturbs future intensities, so small ambiguous-band mass is not a technical afterthought but part of what makes consistent recovery plausible from a single realisation. 

The Supplement gives two transparent regimes under which the ambiguous-band contribution is small. In a weak-excitation regime, the oracle log-odds remain close to the baseline log-odds, so low-signal points occupy little \(\lambda^\star\)-mass. In a dominant-column-bias regime, provided the baseline is aligned with that bias, feedback systematically favours one component strongly enough that the log-odds inherit a stable sign except near the baseline decision boundary. These regimes control the ambiguous-band part of the floor. In Hawkes models, the full floor vanishes only when that control is combined with the compatibility conditions in \cref{prop:compatibility}.

\subsection{Consequences for causal estimands}\label{subsec:estimands}

Under Assumption~\ref{ass:Lambda-Lip}, every estimand in \cref{sec:causalEstimands} is a finite linear functional of the interventional mean measures \(\Lambda_\theta^z\), so parameter error transfers directly to plug-in causal error.

\begin{proposition}[Transfer to plug-in causal estimands]\label{thm:sec-estimands}
For any finite linear functional
\[
F(\theta)=\sum_{\ell=1}^M c_\ell \Lambda_\theta^{z_\ell}(B_\ell),
\]
on the high-probability event of \cref{thm:sec-contraction},
\[
|F(\theta^{(m)})-F(\theta^\star)|
\le
L_\lambda\Big(\sum_{\ell=1}^M |c_\ell|\,|B_\ell|\Big)\,\|\theta^{(m)}-\theta^\star\|.
\]
\end{proposition}

\begin{corollary}[Consistency of plug-in causal estimands]\label{cor:estimands-compat}
Under \cref{prop:compatibility}, plug-in estimators are consistent for fixed-volume causal functionals, including cell-wise contrasts with uniformly bounded cell size and volume-normalised global contrasts. For unnormalised growing-window totals, consistency additionally requires
\[
\Big(\sum_{\ell=1}^{M_n}|c_{\ell,n}|\,|B_{\ell,n}|\Big)\,
\|\theta_n^{(\Tmax_n)}-\theta_n^\star\|
\xrightarrow{\mathbb P}0.
\]
\end{corollary}

In particular, fixed-scale contrasts track parameter consistency directly, whereas unnormalised growing-window totals require an additional scale-rate condition.

\section{Simulation Study}\label{sec:simulation_study}

We illustrate finite-sample performance on simulated spatiotemporal Hawkes  process data with outcome spillover. The goal is to assess recovery of treatment and control parameters and the associated causal estimands from \cref{sec:causalEstimands}. We consider two spatiotemporal Hawkes processes, one for control and one for treatment, with exponential temporal decay and Gaussian spatial kernel. For $k\in\{c,t\}$ the conditional intensity is
\begin{equation*}
\lambda_k(t,x,y \mid \mathcal{H}_{t-}) = \mu_k + K_k \sum_{t_i < t} \exp\!\left\{-\beta_k (t-t_i)
- \alpha_k  \big\|(x,y) - (x_i,y_i)\big\|_{2}^{2} \right\}. 
\end{equation*}

We simulate 100 realisations on \([0,100]\times[0,100]\), with treatment at \(t^\star=10\), end time \(T=110\), a \(10\times10\) spatial grid tessellation, and 50\% of cells treated (see \cref{fig:pp_realiz}). Control and treated processes share \((\mu,\alpha,\beta)=(8,0.01,10)\) and differ only in branching ratio, with \(K_c=0.80\) and \(K_t=0.20\). Additional design details and supplementary simulation results are reported in \cref{app:additional_simulation}.

Figure~\ref{fig:all_nothing} reports the all-or-nothing DAITE (per unit time) estimates under oracle, naive, and SEM labelling. The naive method understates the treatment effect, consistent with spillover from high-productivity control regions inflating the fitted treated process. SEM substantially reduces this bias and tracks the oracle benchmark much more closely. When the control component is estimated using only post-treatment data, performance degrades markedly. A robustness check is reported in \cref{app:additional_simulation} and underscores the value of pre-treatment information.

\begin{figure}
\centering
\includegraphics[width=.5\textwidth]{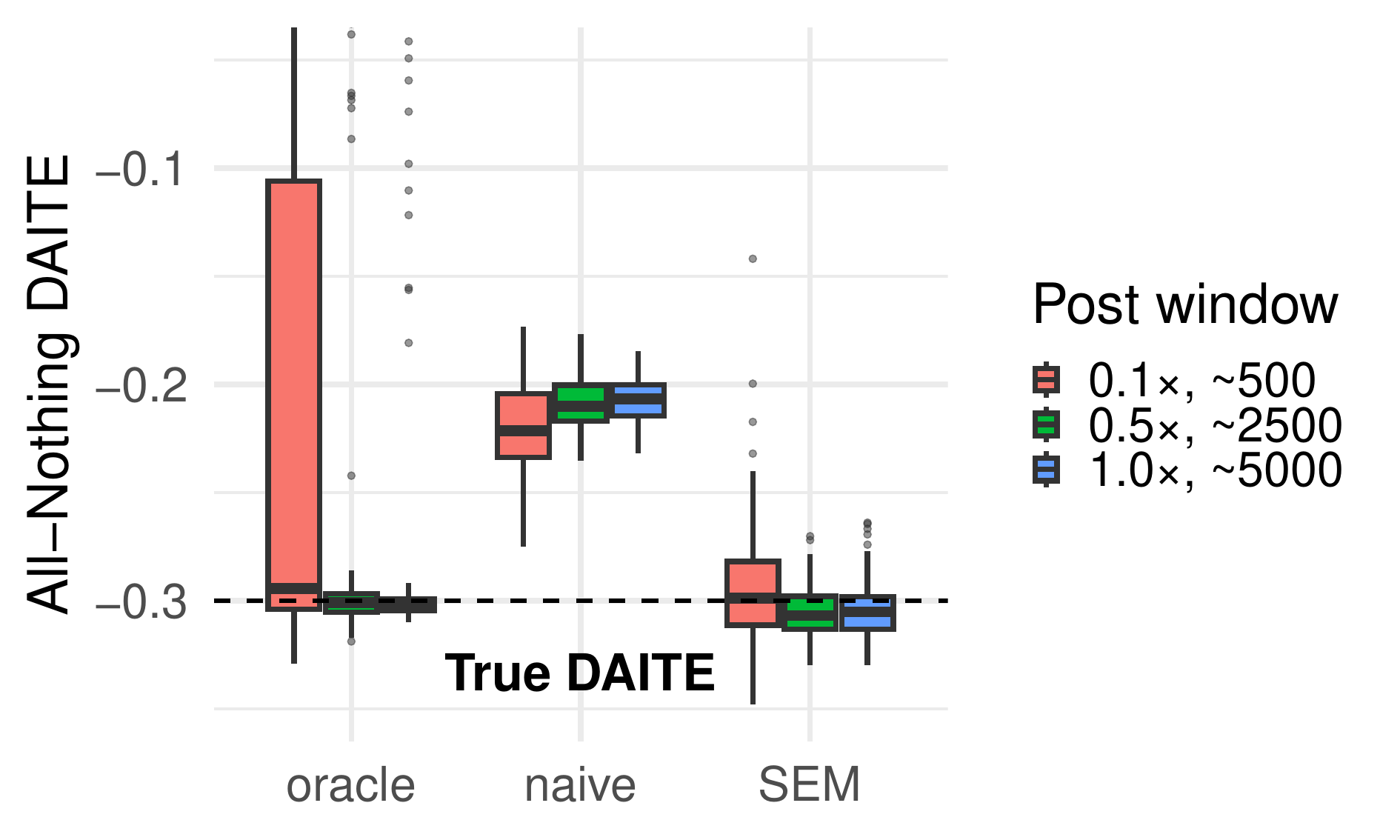}
\caption{Estimated observed-regime vs no-directive DAITE under oracle and naive labelling and the full SEM procedure. The dashed line marks the true value.}
\label{fig:all_nothing}
\end{figure}

\section{Seismic Activity in Oklahoma}\label{sec:application}

On 18 March 2015 the Oklahoma Corporation Commission (OCC) issued a directive requiring disposal-well operators within a designated Area of Interest (AOI) in central Oklahoma to reduce wastewater disposal volumes. We study the causal effect of this directive on Oklahoma seismicity. Although seismicity declined after the directive, causal attribution is difficult because treatment assignment was not random and earthquake occurrence exhibits substantial temporal carryover and spatial spillover. Our target is a DAITE comparing the observed regime with the counterfactual of no directive anywhere in Oklahoma, interpreted as the expected number of earthquakes prevented by the policy. Additional scientific background is given in Supplement~\ref{sec:ok_background}.

We restrict to earthquakes of magnitude 2.5 or greater and consider bivariate epidemic-type aftershock sequence (ETAS) models \citep{ogata} for latent control and treated components with a shared triggering kernel and KDE-driven inhomogeneous background \citep{nandan2017objective}. We use the first half of the pre-treatment data to estimate the background rate and proceed with our SEM approach comparing it to the naive approach setting labels equal to the location treatment status.

Our estimand of interest is the 100-day DAITE, comparing the observed regime with the counterfactual of no directive. \Cref{fig:oklahoma-combined}a  shows that SEM materially changes the estimated policy contrast relative to naive cell-wise labelling: over the same 100-day horizon the naive fit implies 341 prevented events, whereas SEM gives \(-118\). This reversal is consistent with \cref{fig:oklahoma-combined}b, which shows no immediate post-directive decline in cumulative earthquake counts. This illustrates how failure to account for outcome spillover can substantially distort short-run policy contrasts when carryover and spillover are long-range.

\begin{figure}
\centering

\begin{minipage}[t]{0.48\linewidth}
    \centering
    \includegraphics[width=\linewidth]{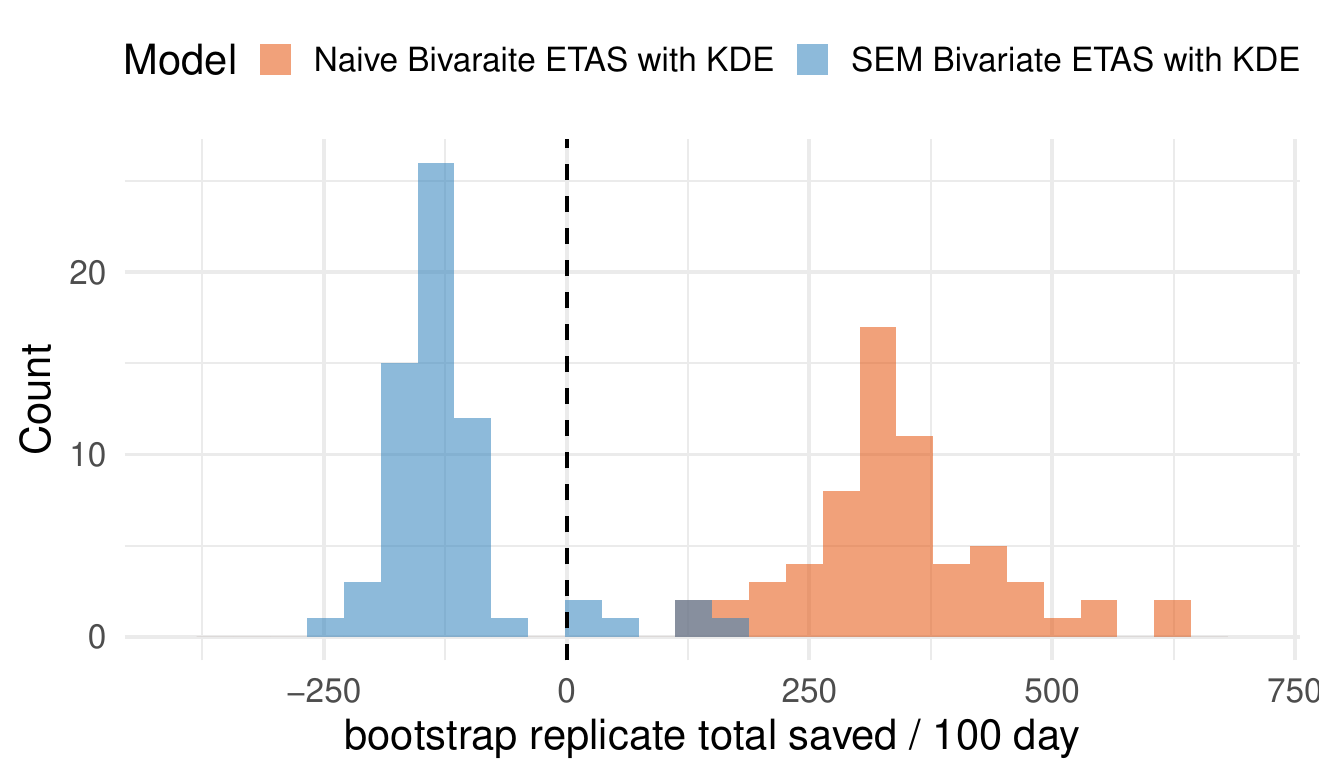}
    \small
    (a) Bootstrap replicate distribution of all-or-nothing total saved for naive vs.\ SEM. The dashed vertical line marks zero savings.
\end{minipage}
\hfill
\begin{minipage}[t]{0.48\linewidth}
    \centering
    \includegraphics[width=\linewidth]{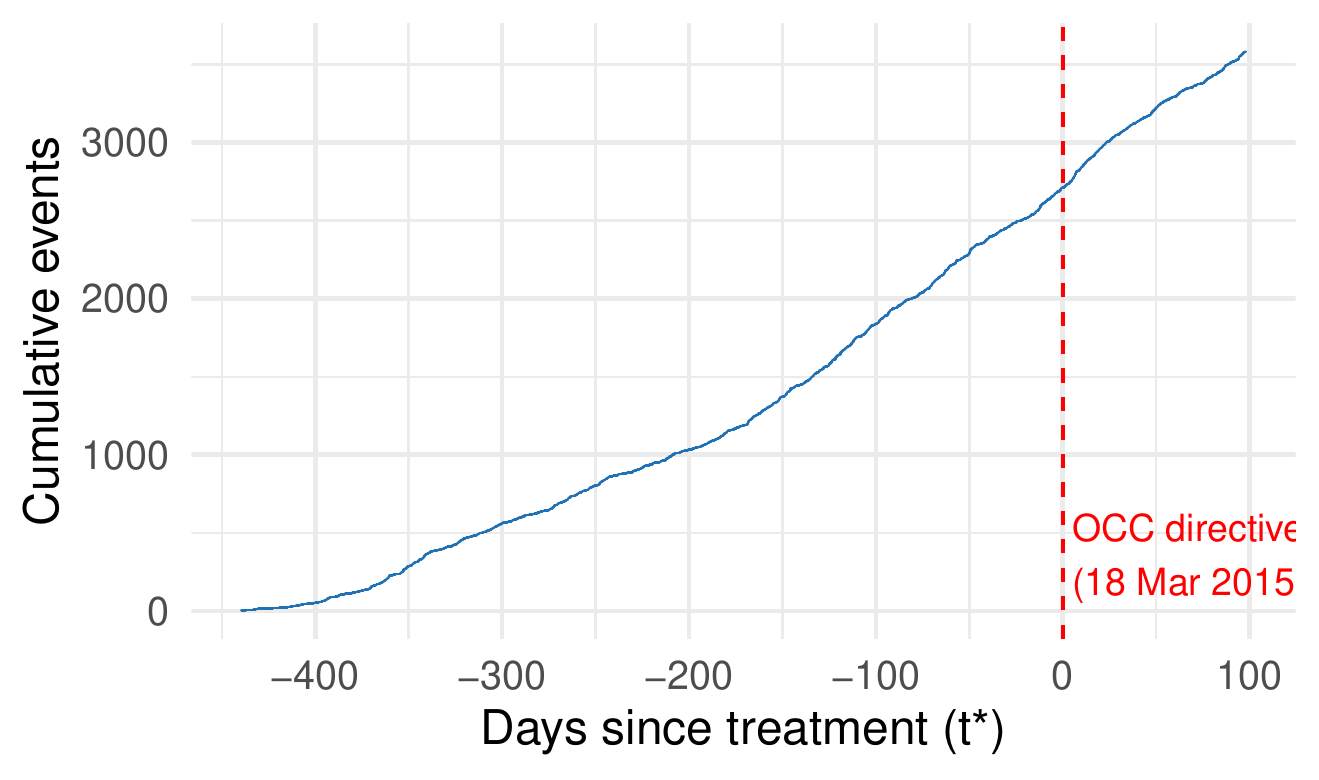}
    
    \small
    (b) Cumulative counts of earthquakes in Oklahoma greater than 2.5 in magnitude.
\end{minipage}

\caption{Oklahoma results.}
\label{fig:oklahoma-combined}
\end{figure}

This above-described result initially appears contradictory to existing literature that claims disposal reductions eventually decrease seismicity (cf. \citet{LangenbruchZoback2016}). Our analysis suggests caution about drawing that conclusion too early: when carryover and spillover are long ranged, pre-treatment dynamics and untreated-area activity cannot be ignored in short-run causal comparisons. The fully self-exciting naive fit is flexible enough to absorb that misspecification and still imply substantial earthquake savings, whereas SEM, after plausibly relabelling points to their latent process, suggests little or no short-run saving.

\section{Discussion}\label{sec:conclusion}

We have developed a causal framework for continuous spatiotemporal point-process outcomes under cell-level interventions with spillover and carryover. By modelling the post-treatment process as an unlabelled superposition of latent control and treatment components, the framework retains exact event times and locations while yielding finite-dimensional causal estimands through an exogenous tessellation.

Important extensions include multivalued or time-varying treatments, richer marked-process formulations, weaker modelling assumptions, and finite-sample guarantees for the practical SEM algorithm itself. More broadly, the paper shows that causal inference for continuous spatiotemporal event data is possible without collapsing outcomes to coarse counts, provided that the treatment design, model structure, and estimand scale are aligned.

\section*{Acknowledgements}
This work was supported by the Marsden Fund under grant 25-UOO-086. During the preparation of this work, the authors used GPT-5.2 Pro and Opus 4.6 to assist with code development, calculations, and identification of relevant references. All AI-assisted output was subsequently reviewed, verified, and edited by the authors, who take full responsibility for the content of the publication.

\section{Data and Code Availability}

The methods described are implemented in the R package \texttt{PPDisentangle} \cite{PPDisentangle}. Data and replication scripts will be archived in a Zenodo repository upon acceptance.  Earthquake records for Oklahoma were sourced from the U.S. Geological Survey (USGS) Earthquake Hazards Program (\url{https://earthquake.usgs.gov/}). Regulatory Area-of-Interest (AOI) geometries were obtained from the Oklahoma Corporation Commission GIS service (\url{https://gis.occ.ok.gov/}). County boundaries were retrieved from the U.S. Census Bureau via the \texttt{tigris} R package.

\bibliographystyle{rss}      
\bibliography{bibliography} 
\newpage

\setcounter{page}{1}
\renewcommand{\thepage}{S\arabic{page}} 
\section*{Supplement}
\setcounter{section}{0}
\renewcommand{\thesection}{S\arabic{section}}

\section{Guidance for practitioners}

\subsection{Spatiotemporal scale of effect}\label{sec:scale-of-effect}
Although \(\mathcal T\) partitions only space, the practical importance of spillover is inherently spatiotemporal. In many applications, the tessellation is not chosen to suit the analysis, but is fixed \textit{a priori} by substantive, administrative, or design considerations, or generated by an exogenous mechanism independent of the realised point pattern. Examples include neighbourhood boundaries in a vaccine roll-out, school catchments, administrative regions, or regulatory zones. The relevant question is therefore not how to choose \(\mathcal T\) so that spillover appears, but whether spillover operates at a scale for which the resulting cell-level estimands are scientifically meaningful and practically estimable.

In space, let \(L_{\text{process}}\) denote the characteristic spatial interaction or decay length of the process and \(L_{\text{cell}}\) a characteristic cell size, for example the square root of the average cell area. In time, let \(H_{\text{process}}\) denote the characteristic memory or temporal decay scale of the process, to be compared with the post-treatment analysis horizon \(T-t^\ast\). For Hawkes models, these scales are determined by the spatial and temporal decay of the triggering kernel; for dispersal models, they are given by the corresponding dispersion scales.

Our framework is most useful when the spatial spillover scale \(L_{\text{process}}\) is of the same order as \(L_{\text{cell}}\), or at least not much smaller. In that regime, spillover is large enough to matter for the cell-level causal question, but not so diffuse that cell-level contrasts are washed out. The temporal analogue enters through the analysis horizon \(T-t^\ast\): the temporal decay should not be so slow, relative to \(T-t^\ast\), that the induced spillover cannot be meaningfully observed or identified within the available post-treatment window.

If influence is very local in space relative to \(L_{\text{cell}}\), then spillover is largely confined to neighbourhoods of spatial cell boundaries, making the problem closer to classical edge correction \citep{cronie2011some}.
Conversely, if the spillover range is much larger than the cell size, or if the temporal memory is long relative to \(T-t^\ast\), each cell is affected by many other cells and by a long history, making separation of local treatment effects and cumulative spillover substantially harder and, in some settings, not practically estimable without stronger structure. In the language of \cref{sec:finite-sample-theory}, this is the regime in which the ambiguous band can carry substantial mass and the resulting statistical floor becomes less favourable; see also \cref{thm:sec-contraction,prop:compatibility}.

Accordingly, the choice of \(\mathcal T\) should be viewed first as part of the substantive definition of the intervention and outcome, and only then as a statistical design choice. Before defining cell-level effects, one should ask whether the chosen units are natural for the policy question, whether spillover operates over a comparable spatial range, and whether the post-treatment horizon is long enough, relative to the temporal decay of the process, for those effects to be observed.

\subsection{Practical implementation of stochastic EM}\label{sec:pracImpSEM}

Choosing a proposal distribution \(q(r\mid N)\) is the central design choice in stochastic EM. We propose a simple discrepancy-driven proposal that focuses sampling on regions where the current control fit under-explains or over-explains the observed counts.

Let \(C_j^{\text{obs}}:=N(\mathcal I_j)\) denote the observed post-treatment count in cell \(\mathcal I_j\). Using the current iterate \(\theta'\), compute the component-wise predicted post-treatment counts
\[
C_{j,0}:=\int_{\mathcal I_j}\widehat\lambda_0(\tau;\theta')\,\dd\tau,
\qquad
C_{j,1}:=\int_{\mathcal I_j}\widehat\lambda_1(\tau;\theta')\,\dd\tau,
\]
or their simulation-based analogues when these integrals are not available in closed form. Using the current estimates \(\lambda_0(\cdot\mid\theta'_0)\) and \(\lambda_1(\cdot\mid\theta'_1)\), compute predicted post-treatment counts \(C_{j,0}\) and \(C_{j,1}\) (via compensator integrals or fast simulation), and define
\[
D_j^{+}=\bigl(C_j^{\text{obs}}-C_{j,0}\bigr)_+,
\qquad
D_j^{-}=\bigl(C_{j,0}-C_j^{\text{obs}}\bigr)_+.
\]
Large \(D_j^{+}\) flags excess mass the control model cannot explain and suggests flipping control labels to treatment in \(\mathcal I_j\); \(D_j^{-}\) suggests the reverse. A single proposal is generated by (i) initializing from a baseline labelling \(r_{\text{base}}\) (for example, \(r_i=Z_j\) for each event \(\gamma_i\in\mathcal I_j\)); (ii) drawing flip budgets \(P^{+}\sim\mathrm{Poisson}\!\bigl(\sum_j D_j^{+}\bigr)\) and \(P^{-}\sim\mathrm{Poisson}\!\bigl(\sum_j D_j^{-}\bigr)\); (iii) forming cell weights \(W_j^{+}\propto D_j^{+}\) and \(W_j^{-}\propto D_j^{-}\); (iv) sampling \((L_1^{+},\dots,L_J^{+})\sim\mathrm{Multinomial}(P^{+};W^{+})\) and \((L_1^{-},\dots,L_J^{-})\sim\mathrm{Multinomial}(P^{-};W^{-})\); and (v) within each cell \(j\), uniformly relabelling \(L_j^{+}\) currently control-labelled points to treatment and \(L_j^{-}\) currently treatment-labelled points to control. Within each selected cell, one may optionally prioritise points using the current label-induced score \(\tilde s_{\theta'}(\gamma_i;r_{\text{base}})\), for example by ranking \(0\to1\) candidates by larger values and \(1\to0\) candidates by smaller values.

The discrepancy-guided proposal above is motivated by, but not itself covered by, the analysis in \cref{sec:finite-sample-theory} wherein we study an idealised predictable blockwise hard-EM surrogate rather than the adaptive SEM implementation.  The causal estimands defined in \cref{sec:causalEstimands} are smooth functionals $\Psi(\theta)$ of the interventional mean measure, so uncertainty in $\hat\Psi=\Psi(\hat\theta)$ may be propagated from $\hat\theta$ using a first-order delta-method approximation. When derivatives of $\Psi$ are difficult to evaluate or the linear approximation is unreliable, a parametric bootstrap under the fitted model provides a convenient alternative \citep{davis2025multivariate}. 

\section{Additional simulation results}\label{app:additional_simulation}
\cref{tab:sim-true-values} shows the model parameters. Figure~\ref{fig:pp_realiz} shows a single realisation of the superposed Hawkes processes. The treatment reduces productivity by lowering the branching ratio $K$.  Tables \ref{tab:treated_params_summary} and \ref{tab:control_params_summary} summarise the treated and control parameter estimates across simulations respectively.

\begin{table}
\caption{True parameter values used in the Hawkes process simulation study.}
\label{tab:sim-true-values}
\centering
\begin{tabular}{lcc}
\toprule
Parameter & Control process & Treated process \\
\midrule
$\mu$     & 8.0  & 8.0  \\
$\alpha$  & 0.01 & 0.01 \\
$\beta$   & 10.0 & 10.0 \\
$K$       & 0.80 & 0.20 \\
\midrule
\textbf{Interpretation} & high triggering & low triggering \\
\bottomrule
\end{tabular}
\end{table}

\begin{figure}[!ht]
\centering
\includegraphics[width=\textwidth]{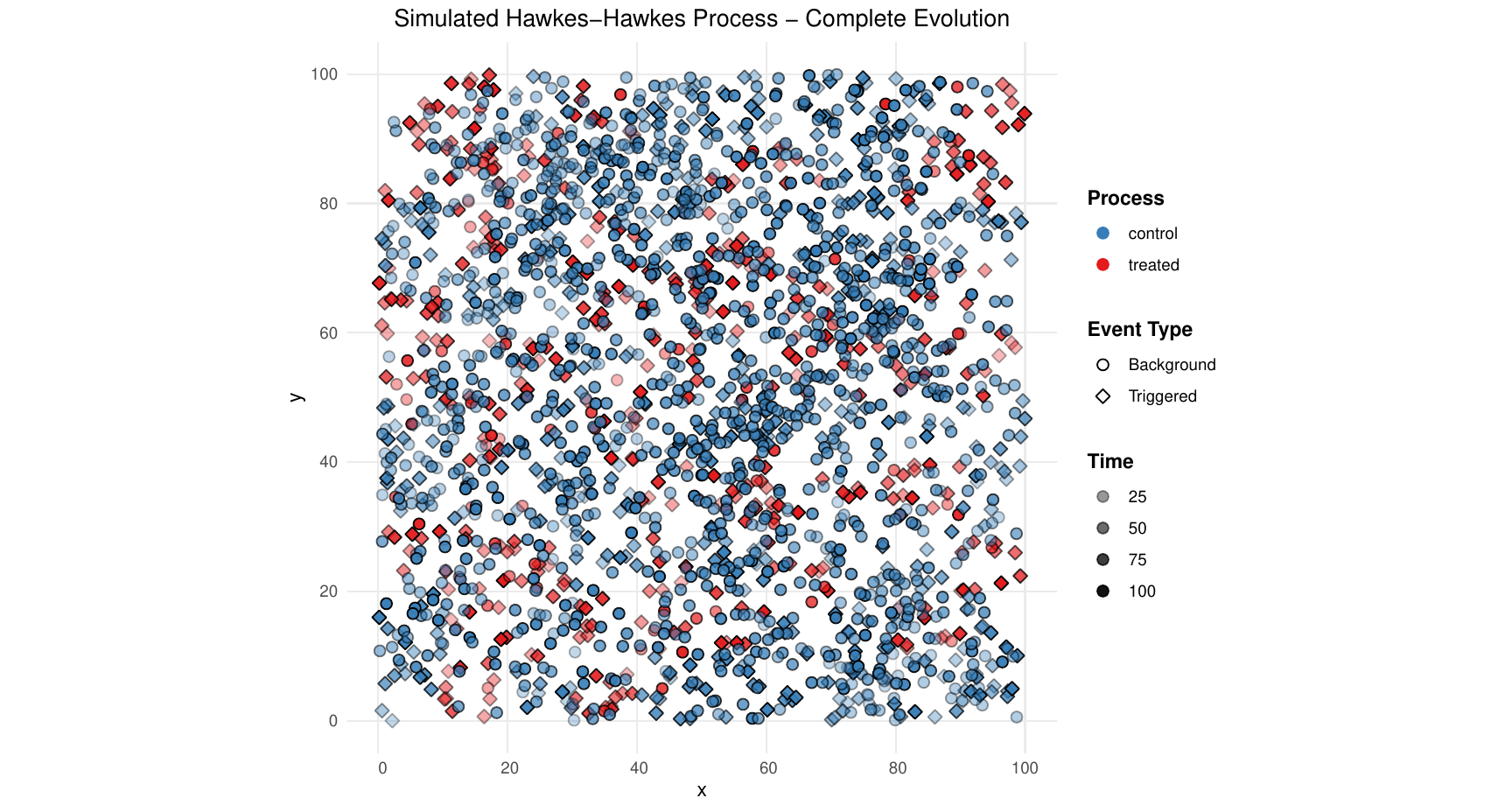}
\caption{One realisation of control and treatment process points superposed}
\label{fig:pp_realiz}
\end{figure}

Figures~\ref{fig:all_nothing} shows the estimated all-or-nothing DAITE, for 2 labelling methods, the oracle (true labelling known), the naive (labelling equals locations' treatment) and the full SEM estimation method. Due to the presence of pre-treatment data, the control parameters are used at their true values in the calculation of the DAITE. The naive labelling underestimates the effect of treating all regions because spillover from high-productivity control regions inflates the estimated $K$ for the treated process. In contrast, the SEM method yields DAITE estimates that are both close to the oracle labelling's estimate and the true DAITE.

In addition, the oracle labelling produces significantly more outliers than the SEM and for a low number of points has much more variability than the SEM. As the number of points increases the oracle improves, with the SEM improving less so, this reflects that adding more points reduces estimation error under a known labelling, but makes the labelling task harder, so adding more data does not necessarily always improve the feasibility of the SEM method.

\begin{table}
\caption{Treated parameter mean estimates across simulations for largest post-treatment time window, with standard errors in parentheses.}
\label{tab:treated_params_summary}
\centering
\begin{minipage}{\textwidth}
\centering
\begin{tabular}{lcccc}
\toprule
Method & $\mu$ mean (SE) & $\alpha$ mean (SE) & $\beta$ mean (SE) & $K$ mean (SE) \\
\midrule
Oracle & 8.073 (0.037) & 0.013 (0.000) & 9.202 (0.472) & 0.236 (0.019) \\
Naive  & 15.030 (0.061) & 0.037 (0.000) & 9.763 (0.100) & 0.218 (0.003) \\
SEM    & 8.326 (0.092) & 0.051 (0.001) & 14.514 (0.352) & 0.127 (0.004) \\
\bottomrule
\end{tabular}
\vspace{0.5em}
\end{minipage}
\end{table}

\begin{table}
\caption{Control parameter mean estimates across simulations for largest post-treatment time window, with standard errors in parentheses.}
\label{tab:control_params_summary}
\centering
\begin{minipage}{\textwidth}
\centering
\begin{tabular}{lcccc}
\toprule
Method & $\mu$ mean (SE) & $\alpha$ mean (SE) & $\beta$ mean (SE) & $K$ mean (SE) \\
\midrule
Oracle & 8.378 (0.037) & 0.011 (0.000) & 9.967 (0.034) & 0.730 (0.002) \\
Naive  & 13.577 (0.082) & 0.029 (0.000) & 9.842 (0.098) & 0.359 (0.003) \\
SEM    & 8.949 (0.053) & 0.011 (0.000) & 10.120 (0.041) & 0.715 (0.001) \\
\bottomrule
\end{tabular}
\vspace{0.5em}
\end{minipage}
\end{table}

Figure \ref{fig:all_nothing_bad} shows the same DAITE estimates, but uses control parameters estimated from the post-treatment data. This is clearly worse, with both approaches struggling to recover the true DAITE. This further underscores the need for pre-treatment data.

\begin{figure}[!ht]
\centering
\includegraphics[width=\textwidth]{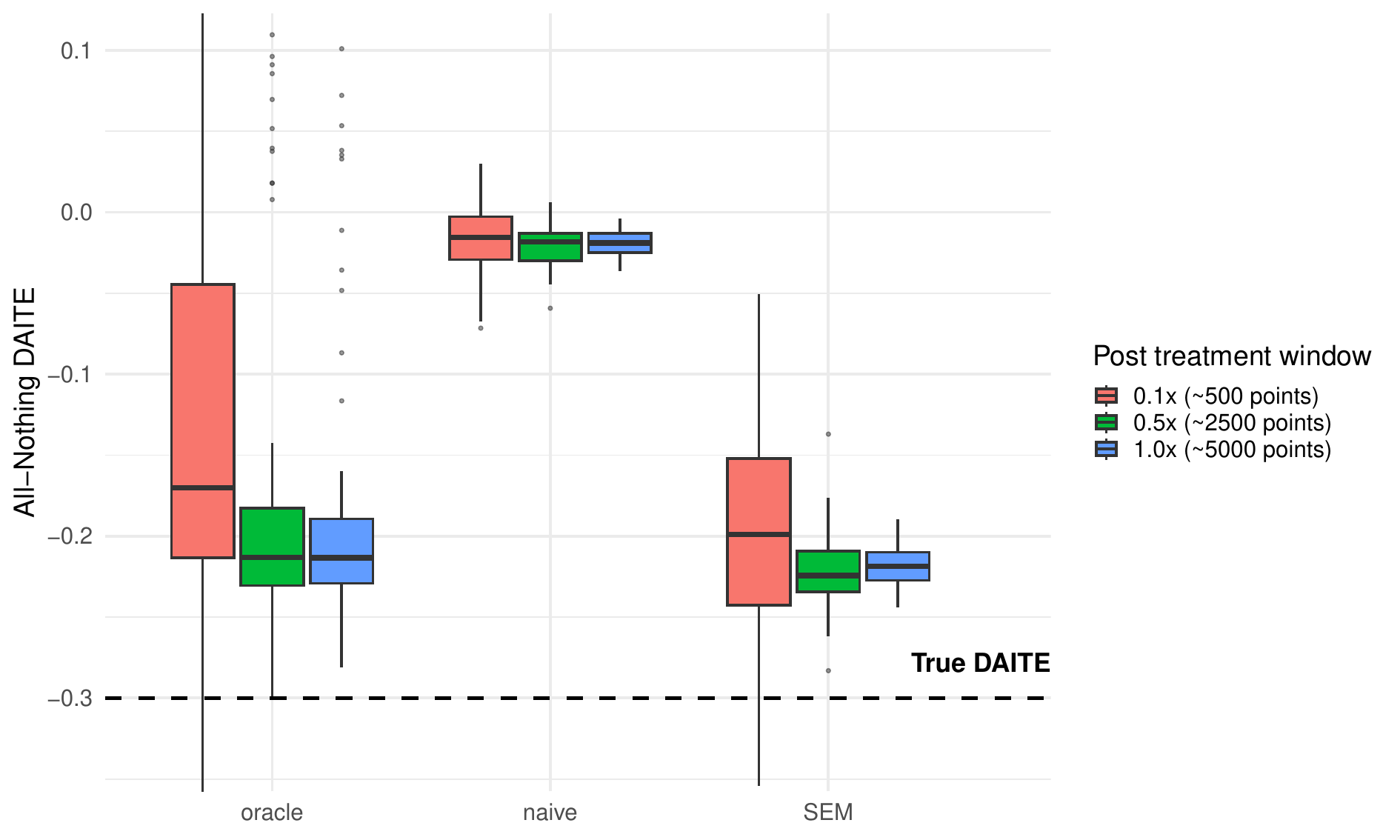}
\caption{Estimated all-or-nothing DAITE under oracle and naive labelling and the full SEM procedure. The dashed line marks the true value.}
\label{fig:all_nothing_bad}
\end{figure}

\section{Additional results for Oklahoma study}\label{sec:additional_application}

\subsection{Scientific background and policy context}\label{sec:ok_background}

A sharp rise in seismicity was observed in Oklahoma after 2008. Existing research attributes this increase primarily to high-volume wastewater (produced-water) disposal, especially disposal into the deep Arbuckle formation near basement-connected faults, rather than to hydraulic fracturing \citep{WeingartenEtAl2015}. 
Earthquake activity clustered near large-volume disposal wells and closely tracked rapid increases in disposal volumes in central Oklahoma \citep{WalshZoback2015}. Later work emphasized injection depth relative to crystalline basement and delayed pressure diffusion, which help explain both the regional spread of seismicity and the lagged decline in earthquake rates after disposal volumes were reduced \citep{Hincks2018}.
Although Oklahoma has also experienced earthquakes directly associated with hydraulic-fracturing stages, that literature generally treats them as a distinct process and a smaller contributor to the statewide surge in the 2010s \citep{RubinsteinMahani2015,SkoumalEtAl2018}.

On 18 March 2015 the Oklahoma Corporation Commission (OCC) issued a directive requiring disposal-well operators within a designated Area of Interest (AOI) in central Oklahoma to reduce wastewater disposal volumes . The AOI is seismically active, and the Arbuckle formation is well suited to wastewater disposal.

\subsection{Study Design}\label{sec:ok_design}

Counties whose centroids lie inside the published AOI polygon are treated; the remainder are control. Only events with magnitude at least $2.5$ are retained. Pre-$t^\ast$ events are split by time order into two halves: the first half is reserved to estimate the spatial background (KDE); the second half is used as pre-treatment data (pre-treatment points are all labelled control).  \cref{fig:AOI} shows the AOI and \cref{fig:OK_points} show the points pre and post-treatment.

\begin{figure}[!ht]
\centering
\includegraphics[width=\textwidth]{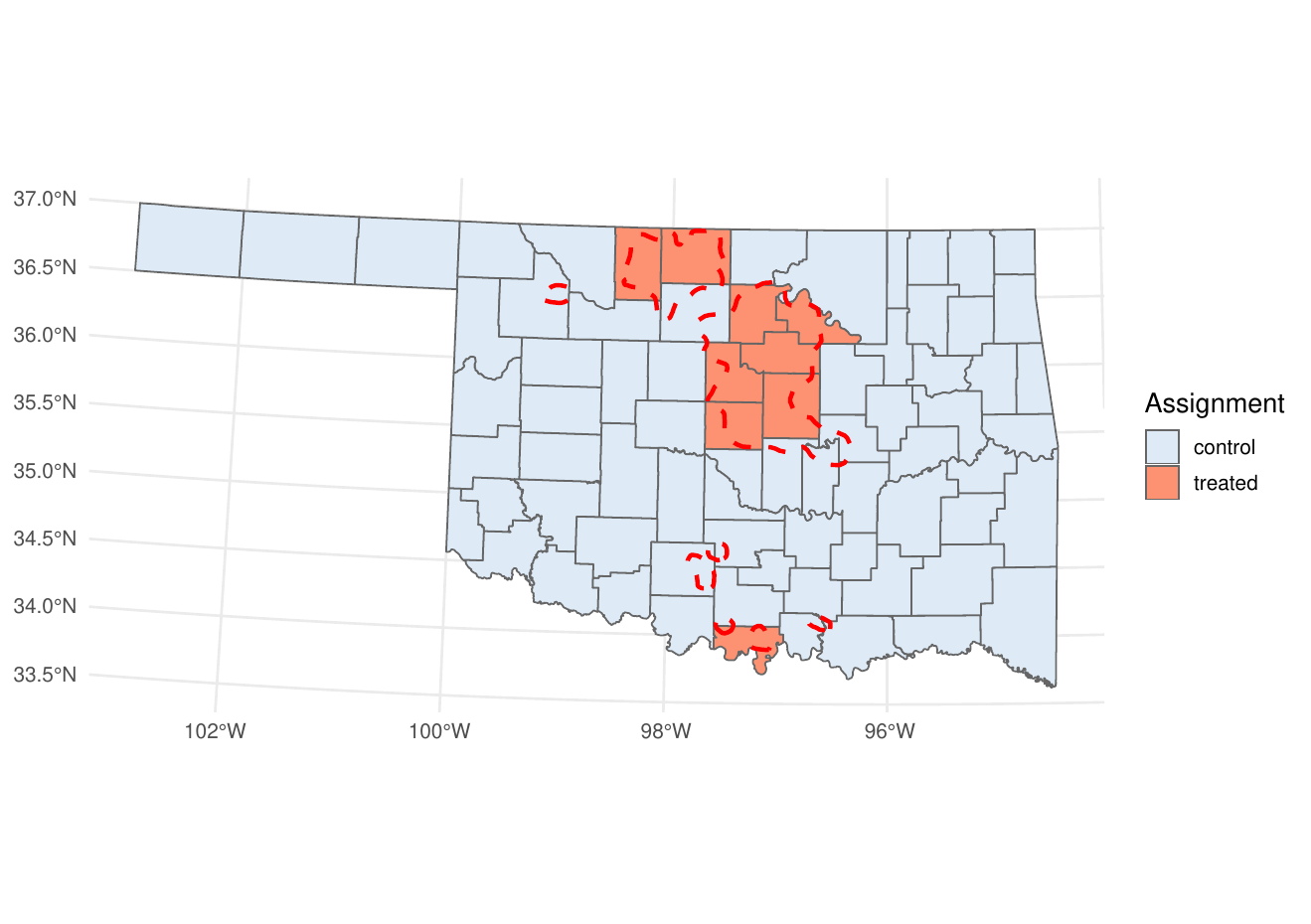}
\caption{Oklahoma and the AOI marked with dashed lines, counties are considered treated if their centroid is within the AOI}
\label{fig:AOI}
\end{figure}

\begin{figure}[!ht]
\centering
\includegraphics[width=\textwidth]{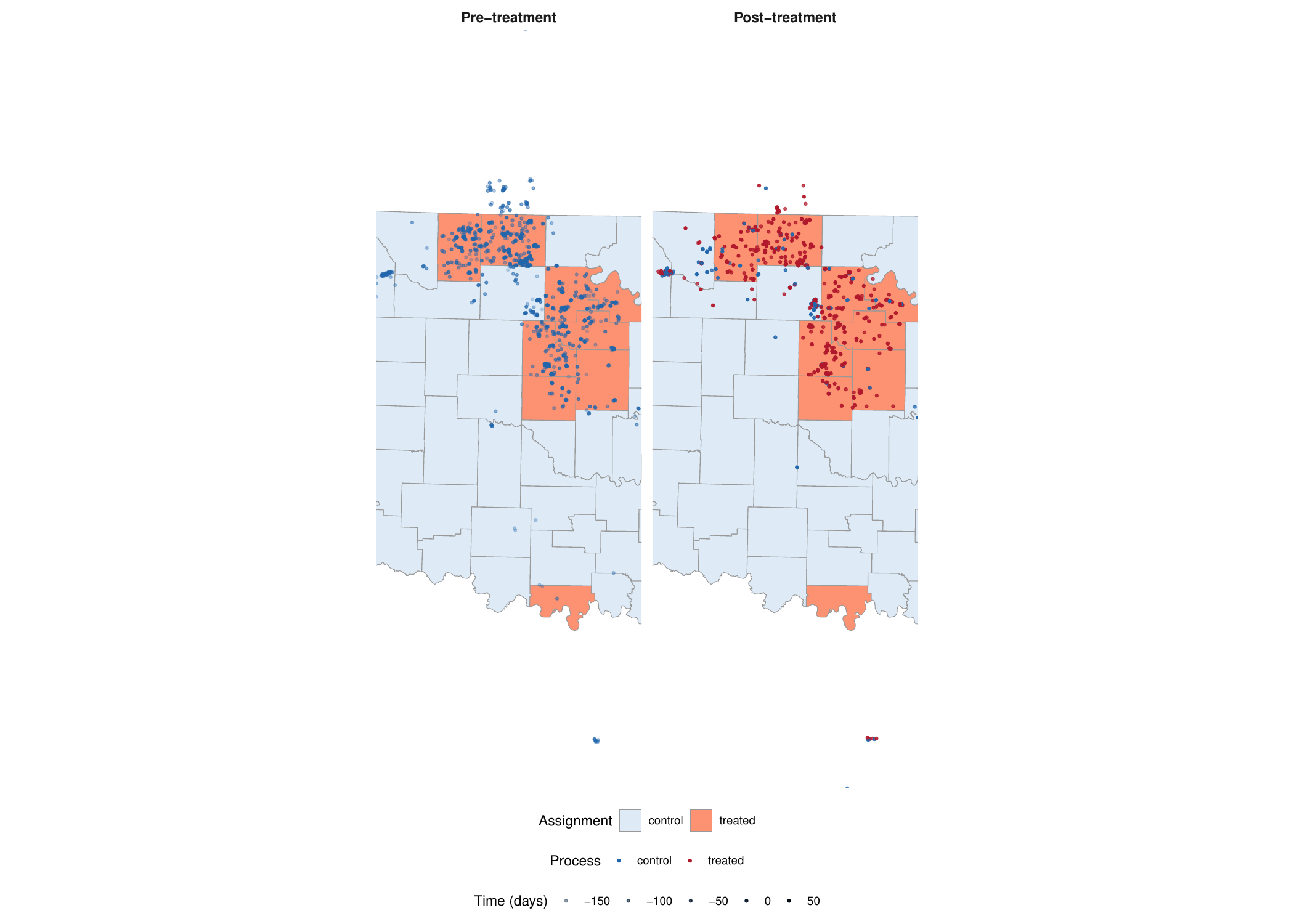}
\caption{Pre and post-treatment point pattern with colour denoting location inferred process}
\label{fig:OK_points}
\end{figure}

\subsection{Primary specifications}\label{sec:ok_primary_models}

Section~\ref{sec:application} reports contrasts from two fitted bivariate ETAS specifications that share the same nonparametric spatial background and the same county tessellation; for brevity we refer to them as the naive KDE fit (joint maximum likelihood with event labels fixed to the cell's treatment status) and the SEM KDE fit (stochastic EM with data-driven relabelling).

Both specifications are bivariate marked ETAS models on Oklahoma counties (US Census 2022 boundaries, EPSG:5070), with separate control and treated components, Omori-Utsu temporal triggering and an isotropic power-law spatial kernel with magnitude-dependent scale, and cross-excitation between components. The Naive KDE fit assigns each earthquake to control or treated according to the county's regulatory status (centroid inside the OCC AOI) and maximizes the resulting bivariate likelihood. The SEM KDE fit applied our algorithm with stochastic relabelling proposals with re-estimation under the same intensity family; configuration details are summarized below.

Let $z\in\{0,1\}$ index latent process. The conditional intensity for component $k\in\{0,1\}$ takes the form
\[
\lambda_k(t,x,y)
  = \frac{\mu_k}{|S_k|}\, W_k(x,y)
  + \sum_{\ell \in \{0,1\}} \sum_{\substack{j : t_j < t \\ r_j = \ell}}
    A_{k\ell}\, e^{\alpha_{m,k\ell}(m_j - m_0)}\,
    g(t - t_j)\, f(x-x_j, y-y_j \mid m_j),
\]
with Omori-Utsu $g(\Delta t)=\frac{p-1}{c}(1+\Delta t/c)^{-p}$ on $\Delta t>0$ and power-law spatial factor $f \propto (1+r^2/d(m_j))^{-q}$ with $d(m_j)=D\exp\{\gamma(m_j-m_0)\}$. Likelihood evaluation applies a finite temporal history cutoff $t_{\mathrm{trunc}}$ (days): parents farther back than $t_{\mathrm{trunc}}$ do not contribute to the triggering sum, and the temporal normalization of the Omori kernel is adjusted accordingly. The number of days truncated was determined using the same data reserved for the KDE estimation, and was chosen so that the triggering density had decayed $95\%$ from its initial value.

\subsection{Nonparametric background}\label{sec:ok_kde}
On the county partition, $W_k$ is obtained by smoothing the spatial locations of held-out pre-$t^\ast$ events with a Gaussian kernel: bandwidth follows Diggle's cross-validation (\verb|bw.diggle| in \texttt{spatstat}), implemented as $\sigma = 2\,\hat\sigma_{\mathrm{diggle}}$ in the code, then density values are normalized cell-wise so that background mass matches each cell's area.

\subsection{Stochastic EM configuration and Diagnostics}\label{sec:ok_sem_config}

The SEM KDE fit requires tuning and uses the hyper parameters described in \cref{tab:ok_sem_config}, full details can be found in the PPDisentangle \cite{PPDisentangle} \texttt{R} package. 

\begin{table}
\caption{\label{tab:ok_sem_config}Stochastic SEM configuration for the Oklahoma \emph{SEM KDE fit}}
\centering
\begin{minipage}{\textwidth}
\centering
\small
\begin{tabular}{@{}l@{\quad}l@{}}
\toprule
Setting & Value\\
\midrule
SEM outer rounds ($N_{\mathrm{iter}}$) & 100\\
Inner iterations (per outer round) & 2000\\
Inner proposals (per inner iteration) & 20\\
Retained labellings (per outer round) & 20\\
Outer optim cap (adaptive SEM, uni.~path) & 220\\
Outer optim cap (bivariate ETAS path) & 1000\\
Warm-start fixed full-parameter step & no\\
Relabelling selection method & sample\_weighted\\
Relabelling selection temperature & 0.08\\
Proposal change factor & 0.01\\
Change-factor min.~multiplier & 0.2\\
Change-factor max.~multiplier & 2\\
Max.~relabel step (fraction of points) & 0.05\\
Forced param.~update flip fraction & 0.05\\
Temporal relabel weight & 0\\
Temporal relabel scale (days) & 15\\
\bottomrule
\end{tabular}
\end{minipage}
\end{table}

It is crucial to track the convergence of the SEM algorithm. \cref{fig:sem_convergence} shows the progression of the likelihood together with the number of proposed and accepted label flips. The flattening of the likelihood together with the stabilisation of the number of flips proposed and accepted suggest that the labelling sampler has reached an area of high likelihood of labellings.

\begin{figure}
\begin{minipage}[t]{0.48\linewidth}
    \centering
    \includegraphics[width=\linewidth]{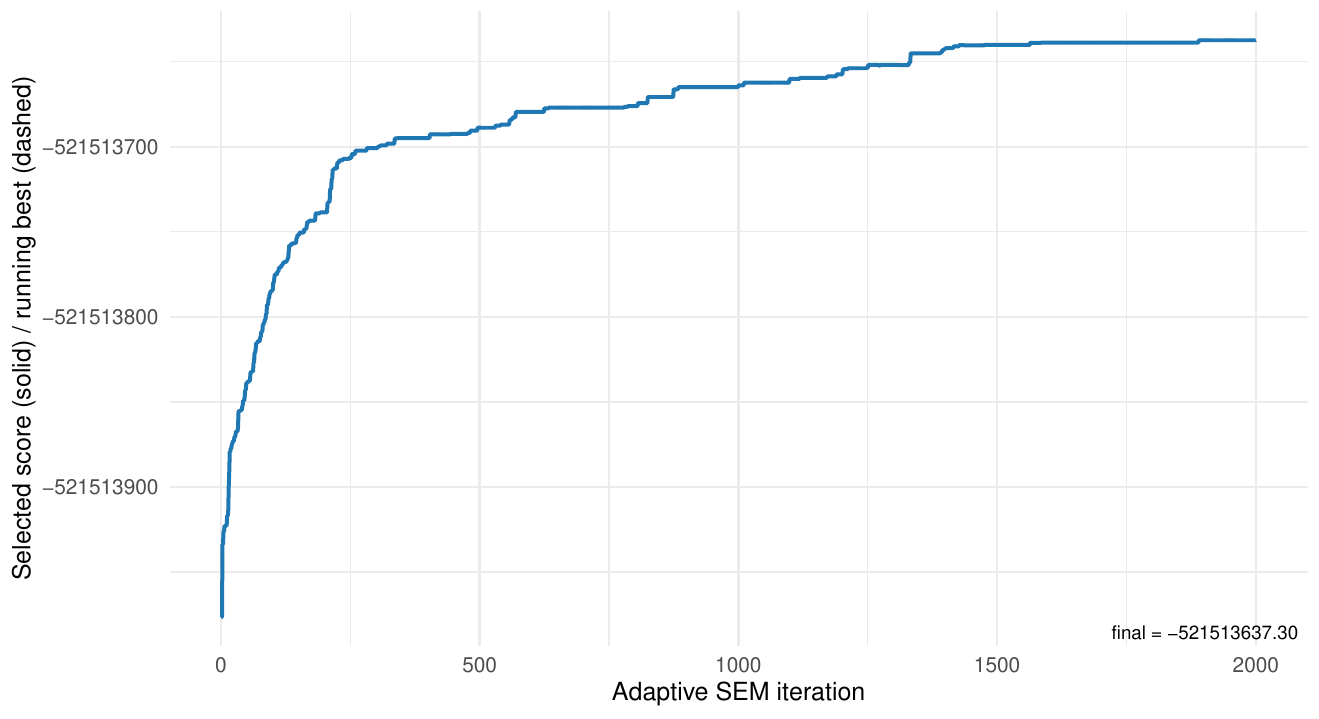}
    \small
    (a) Likelihood trace during adaptive step of relabelled points.
\end{minipage}
\hfill
\begin{minipage}[t]{0.48\linewidth}
    \centering
    \includegraphics[width=\linewidth]{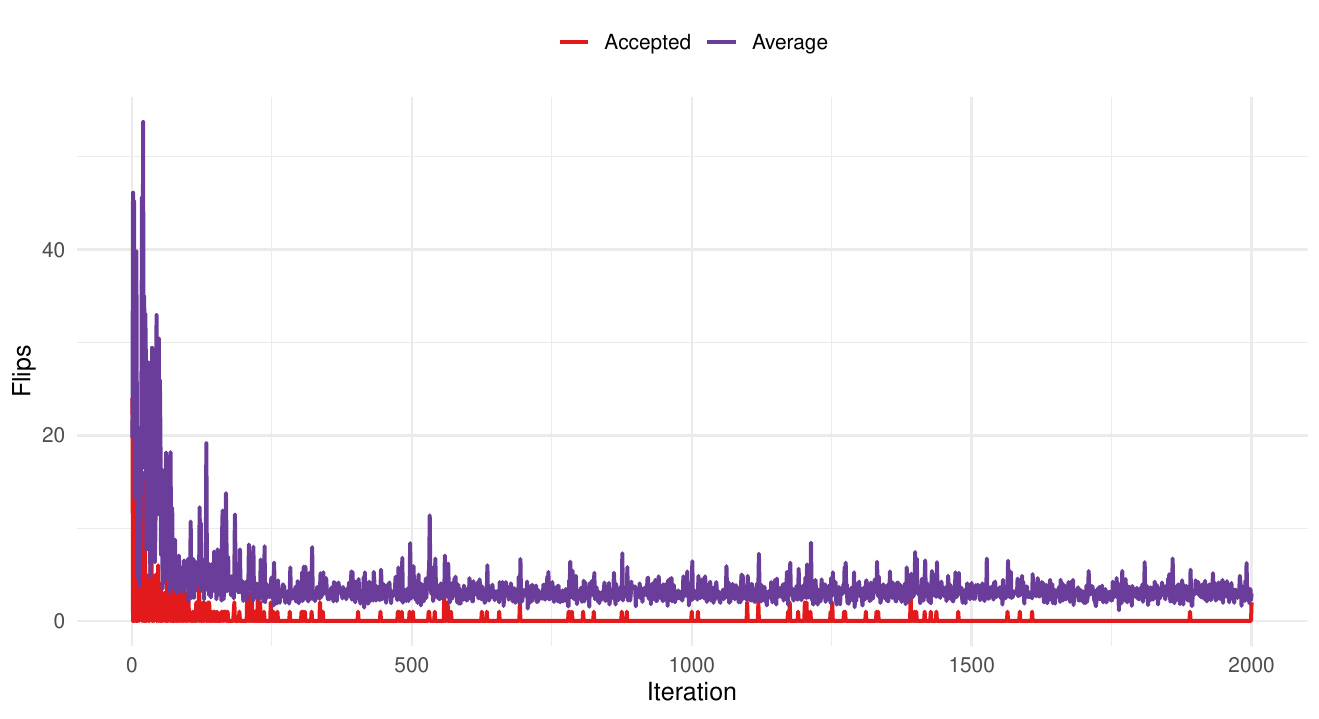}
    \small
    (b) Trace of number of proposed and accepted flips in the adaptive step.
\end{minipage}
\caption{SEM Diagnostics for Oklahoma KDE SEM fit}
\label{fig:sem_convergence}
\end{figure}

\subsection{Results}

\cref{tab:fitted_params} shows the fitted parameters for the Naive KDE and SEM KDE fits. 
We show simulated estimates in \cref{tab:ok_counts}. Distributions of parametric bootstrapped results are shown in \cref{fig:oklahoma-combined}a.

\begin{table}
\caption{Full bivariate ETAS parameter estimates: naive KDE vs SEM KDE fit}
\label{tab:fitted_params}
\centering
\begin{minipage}{\textwidth}
\centering
\begin{tabular}{@{}lcc@{}}
\toprule
Parameter & Naive KDE & SEM KDE\\
\midrule
$\mu_0$ & 0.55464 & 0.37216\\
$\mu_1$ & 0.66415 & 0.82775\\
$A_{00}$ & 0.64514 & 0.70944\\
$\alpha_{m00}$ & 0.8105 & 0.12455\\
$A_{11}$ & 0.52016 & 0.48404\\
$\alpha_{m11}$ & 0.96478 & 0.86313\\
$A_{01}$ & 0.02944 & 0.12425\\
$\alpha_{m01}$ & 1.22685 & 0.65214\\
$A_{10}$ & 0.01362 & 0.00055\\
$\alpha_{m10}$ & 1.07839 & 2.30358\\
c & 4.39229 & 4.81292\\
p & 2.06049 & 2.00183\\
D & 0.18758 & 0.31547\\
$\gamma$ & 1.39907 & 0.49213\\
q & 1.50194 & 1.50113\\
\bottomrule
\end{tabular}
\vspace{0.5em}
\end{minipage}
\end{table}

\begin{table}
\caption{Expected post-treatment earthquake counts under the no-directive regime $\mathbf 0$ and the observed regime $z_{\mathrm{obs}}$, for naive and SEM fits.}
\label{tab:ok_counts}
\centering
\begin{minipage}{\textwidth}
\centering
\begin{tabular}{lrrr}
\toprule
Method 
& $\bbE[\Nproc^{\mathbf 0}(D)]$ 
& $\bbE[\Nproc^{z_{\mathrm{obs}}}(D)]$ 
& $\Delta = \bbE[\Nproc^{\mathbf 0}(D)] - \bbE[\Nproc^{z_{\mathrm{obs}}}(D)]$ \\
\midrule
Naive & 946 & 605 & 341 \\
SEM   & 334 & 453 & $-118$ \\
\bottomrule
\end{tabular}
\vspace{0.5em}
\end{minipage}
\end{table}

\subsection{Sensitivity: alternative spatial partitions}\label{sec:ok_partition_sens}
To assess dependence on the administrative tessellation, the analysis refits the \emph{same two KDE specifications}, the naive KDE fit and the SEM KDE fit, on alternative partitions. We tested square-grid tessellations whose cell diameters are $1$, $2$, and $5$ times an estimated spatial triggering range from the KDE training pattern, and a two-cell ``AOI region'' partition (AOI interior treated, remainder of the state control). Details are found in the PPDisentangle \cite{PPDisentangle} \texttt{R} package. We did not find qualitative differences between the fits.

\section{Proofs}
\subsection{Setup and complete-data likelihood}\label{sec:setup}

We first collect the basic notation, conventions, and the complete-data likelihood used throughout.

All vector norms $\|\cdot\|$ are Euclidean. Throughout we work with predictable versions of the component intensities $\tau\mapsto \lambda_k(\tau\mid \mathcal F_{t-};\theta)$ that are left continuous in $t$ for $\mu$-a.e.\ $x\in\mathcal S$, jointly measurable in $(t,x)$, and continuous in $\theta$. Left continuity is used whenever pre\-jump values are evaluated at event times; when we write $\lambda(\gamma_i)$ we mean $\lambda(t_i-,x_i)$.

We fix $D=(t^\ast,T]\times\mathcal S$ with product base measure $\dd\tau=\dd t\otimes\dd\mu(x)$ and volume
$|D|=(T-t^\ast) \mu(\mathcal S)$.  Throughout, the deterministic observation and schedule windows are open on the left and closed on the right (e.g.\ $D=(t^\ast,T]\times\mathcal S$). When defining local history neighbourhoods intended to enter an intensity at time \(t\), we take them to be open at the right endpoint (i.e.\ they exclude \(t\)), e.g.\ \((t-L,t)\times B_R(x)\), to preserve predictability.

We assume $\mu(\mathcal S)<\infty$, hence $|D|<\infty$.  All statements that refer to ``large $|D|$'' are in a deterministic triangular-array sense: we consider a sequence of windows $D_n=(t^\ast_n,T_n]\times\mathcal S_n$ with $|D_n|\to\infty$ and (unless explicitly stated) all constants in assumptions are uniform along the sequence. For readability we drop the subscript $n$. 

When discussing high-signal regimes (for example when verifying the concentration Assumption~\ref{ass:G8}), we also allow the data-generating parameter to vary along the triangular array:
$\theta^\star=\theta^\star_n$ (hence $P_{\theta^\star}=P_{\theta^\star_n}$), with $\theta^\star_n\in\Theta_\circ$ for all $n$, and with the constants in our assumptions chosen uniformly in $n$. We continue to suppress the subscript unless the dependence matters (e.g.\ for stationary laws, we write $\pi=\pi_{|D|}$). The parameter dimension $p$ is fixed (does not grow with $|D|$). Although the point pattern on $D=(t^\ast,T]\times\mathcal S$ is observed offline, the finite-sample analysis will work with a time-respecting update schedule to ensure predictability of the integrands used in the martingale bounds.

Fix a deterministic partition
\begin{align*}
t^\ast=&u_0<u_1<\cdots<u_{\Tmax}=T,\\
D_m:=&(u_{m-1},u_m]\times\mathcal S,\\
D^{[u]}:=(&t^\ast,u]\times\mathcal S,\\
D^{(m)}:=&D^{[u_m]}.
\end{align*}
Write $N^{[u]}:=N(D^{[u]})$ and $N^{(m)}:=N(D^{(m)})$.
We continue to denote the events on $D$ by $\gamma_i=(t_i,x_i)$,
$i=1,\dots,N(D)$, ordered by increasing $t_i$; then $N^{[u]}=\max\{i:\ t_i\le u\}$. With the convention $\max\emptyset=0$, we have $N^{[t^\ast]}=0$ and $D^{(0)}=D^{[u_0]}=\emptyset$.
Moreover, for $u\in[t^\ast,T]$,
\[
|D^{[u]}|=(u-t^\ast)\,\mu(\mathcal S),
\qquad
|D_m|=(u_m-u_{m-1})\,\mu(\mathcal S),
\qquad
|D^{(m)}|=|D^{[u_m]}|.
\]
Throughout the remainder of this section we work on the inference window $D=(t^\ast,T]\times\mathcal S$ and use the shifted ground counting process
\[
N^D(t):=N\big((t^\ast,t]\times\mathcal S\big),\qquad t\in[t^\ast,T],
\]
with the convention $N^D(t^\ast)=0$ (hence also $N^D(t^\ast-)=0$). For notational simplicity we suppress the superscript $D$ and write $N(t)$ and $N(t-)$ in place of $N^D(t)$ and $N^D(t-)$. Note that for $u\in[t^\ast,T]$,
\[
N^{[u]}=N(D^{[u]})=N\big((t^\ast,u]\times\mathcal S\big)=N^D(u).
\]

We view $(N_0,N_1)$ (and hence $N=N_0+N_1$) as defined on $(-\infty,T]\times\mathcal S$.
In addition to the inference window $D=(t^\ast,T]\times\mathcal S$, we assume that the restriction of $N$ to a history interval of length $L$ immediately preceding $t^\ast$ is observed and treated as fixed conditioning information. Concretely, throughout we assume that $N$ is observed on $(t^\ast-L,t^\ast]\times\mathcal S$ and we condition on this observed history. (If one prefers to work with an observation start time $0$, one may shift time so that $t^\ast-L\ge 0$, equivalently $t^\ast\ge L$.) All likelihoods and intensities on $D$ are understood as conditional on this observed history. In particular, whenever a neighbourhood $(t-L,t)\times B_R(x)$ reaches below $t^\ast$, the count is evaluated using the observed pre-$t^\ast$ history.

\begin{remark}[Pre-\(t^\ast\) history in component-dependent models]\label{rem:prehistory-marked}
The preceding conditioning convention is sufficient for the model-agnostic setup, but for specializations in which the component intensities depend on component-specific pre-\(t^\ast\) history (notably Hawkes models), conditioning only on the superposed history \(N\) before \(t^\ast\) is not enough to identify the oracle intensities \(\lambda_k(\cdot\mid\mathcal G_{t-};\theta)\). In those specializations we therefore work either under empty pre-history or conditional on the marked pre-\(t^\ast\) history \((N_0,N_1)\) over the memory window; see \S\ref{app:hawkes}.
\end{remark}

We observe an unlabelled simple point process $N=N_0+N_1$ on $D$ with simple components $N_0,N_1$. The observed filtration $\{\mathcal H_t\}$ is generated by $N$. The oracle filtration $\{\mathcal G_t\}$ augments $\{\mathcal H_t\}$ with the true component labels $r^\star=(r^\star_i)_{i=1}^{N(D)}$, so that $N_0,N_1$ are adapted to $\{\mathcal G_t\}$. We work with predictable versions with respect to the filtration in force and we always assume
\[
\mathcal H_t\ \subset\ \mathcal G_t\qquad\text{for every }t,
\]
so $\mathcal H$-predictable processes are also $\mathcal G$-predictable. We work under the probability law $P_{\theta^\star}$ on which $N$ has $\mathcal G$-predictable conditional intensity $\lambda^\star$.

\begin{remark}[Model vs.\ observed vs.\ algorithmic objects]\label{rem:objects}
Before proceeding, we note three related objects. (i) Model: the marked process $(N_0,N_1)$, the oracle filtration $(\mathcal G_t)$, and the oracle intensities
$\lambda_k(\cdot\mid \mathcal G_{t-};\theta)$. (ii) Observed process: the superposed process $N=N_0+N_1$ and the observed filtration $(\mathcal H_t)$. (iii) Algorithm/analysis: label sequences $r$, induced counts $N_k^r$, and label-induced intensities
$\widehat\lambda_k^r(\cdot;\theta)$, which are constructed from $(\mathcal H_t)$ and past labels only
(Assumption~\ref{ass:G0pred}).
\end{remark}

Write $\lambda_k^\star(\tau):=\lambda_k(\tau\mid \mathcal G_{t-};\theta^\star)$ and $\lambda^\star:=\lambda_0^\star+\lambda_1^\star$. For any predictable set $S\subset D$, define the component compensated martingales
\[
M_k^\star(S)\ :=\ N_k(S) - \int_S \lambda_k^\star(\tau) \dd\tau,\qquad k\in\{0,1\},
\]
and the superposed compensated martingale
\[
M^\star(S)\ :=\ N(S) - \int_S \lambda^\star(\tau) \dd\tau\ =\ M_0^\star(S)+M_1^\star(S).
\]

Throughout the analysis, we index complete-data likelihoods and intensities by a labelling
\[
r=(r_i)_{i\ge 1}\in\{0,1\}^{\bbN}.
\]
Only the first $N(D)$ entries affect likelihood terms. For the per-flip algebra we may treat $r$ as an arbitrary fixed sequence. For the algorithmic update path $(r^{(m)})_{m\le\Tmax}$ we allow $r^{(m)}$ to be random but require it to satisfy the blockwise non-anticipation condition in Assumption~\ref{ass:G0pred} below, which ensures that all label-induced intensities used inside martingale integrals on the relevant cumulative windows $D^{(m)}$ are predictable (see the definition of $\widehat\lambda_k^r$ below).

For the analysed update path $(r^{(m)})_{m=0}^{\Tmax}$ we denote the (random) set of visited labellings by
\[
\mathcal R_{\rm path}
:=\big\{r^{(m)}:\ m=0,1,\dots,\Tmax\big\}.
\]
When single-event perturbations of an iterate labelling are needed (e.g.\ in per-flip arguments),
we write the (optional) single-flip closure
\[
\mathcal R_{\rm path}^{\rm flip}
:=\mathcal R_{\rm path}\ \cup\
\big\{\big(r^{(m)}\big)^{(\gamma_j)}:\ m=0,1,\dots,\Tmax,\ j\le N^{(m)}\big\},
\]
where $(r^{(m)})^{(\gamma_j)}$ denotes the labelling obtained by flipping only the $j$th coordinate of $r^{(m)}$
(as in Lemma~\ref{lem:singleflip-nonant}).

Throughout we assume the parameter space $\Theta\subset\bbR^p$ is nonempty, and we fix a nonempty, compact and convex subset $\Theta_\circ\subset\Theta$. Unless explicitly stated otherwise, maximisations are over $\Theta_\circ$.

For the M-step it is convenient to introduce an explicit maximiser correspondence.
For any nonempty compact convex set $K\subseteq\Theta_\circ$ and any labelling $r$, define
\[
\operatorname{ArgMax}_K(r)\ :=\ \arg\max_{\theta\in K}\ell_r(\theta),
\qquad
\hat\theta_K(r)\ \in\ \operatorname{ArgMax}_K(r),
\]
where $\hat\theta_K(r)$ denotes an arbitrary (fixed) measurable selection from the nonempty set
$\operatorname{ArgMax}_K(r)$.

For $u\in[t^\ast,T]$ we also use the restricted correspondence
\[
\operatorname{ArgMax}_K^{[u]}(r)\ :=\ \arg\max_{\theta\in K}\ell_r^{[u]}(\theta),
\qquad
\hat\theta_K^{[u]}(r)\ \in\ \operatorname{ArgMax}_K^{[u]}(r),
\]
where $\hat\theta_K^{[u]}(r)$ denotes an arbitrary (fixed) measurable selection.
When the feasible set is $K=\Theta_{\rm sc}$ we write $\hat\theta_r^{[u]}:=\hat\theta_{\Theta_{\rm sc}}^{[u]}(r)$.

In the local-curvature regime of Assumption~\ref{ass:G6} we will also use the localised neighbourhood
\[
\Theta_{\rm sc}:=\Theta_\circ\cap B(\theta^\star,r_{\rm sc}).
\]
When we write $\hat\theta_r$ without specifying $K$, we mean the localised choice
$\hat\theta_r:=\hat\theta_{\Theta_{\rm sc}}(r)$. This localization is imposed only for the analysis (to remain inside the curvature region); see Lemma~\ref{lem:Mstep-coincide} for conditions under which maximising over $\Theta_\circ$ and over $\Theta_{\rm sc}$ coincide.

We assume throughout that the true parameter lies in the constrained set, i.e.\ $\theta^\star\in\Theta_\circ$. When we invoke an unconstrained first-order condition at an M-step maximiser, we will do so by requiring the maximiser $\hat\theta_r$ to lie in the interior of the relevant feasible set $K\in\{\Theta_\circ,\Theta_{\rm sc}\}$ (cf.\ Assumption~\ref{ass:G9}); otherwise we use the variational inequality form.

For $k\in\{0,1\}$ and for a given history $\mathcal F_{t-}\in\{\mathcal H_{t-},\mathcal G_{t-}\}$ possibly augmented by past labels (defined below), denote the predictable component intensity by
\[
\lambda_k(\tau\mid \mathcal F_{t-};\theta)\qquad(\tau=(t,x)\in D),
\]
and set $\lambda=\lambda_0+\lambda_1$. We assume $\theta\mapsto \lambda_k(\tau\mid \mathcal F_{t-};\theta)$ is continuous for each $\tau$.

A labelling is a sequence $r=(r_i)_{i\ge 1}\in\{0,1\}^{\bbN}$. Define the selected-component process on $[t^\ast,T)$ by
\begin{equation}\label{eq:sel-def}
\mathrm{sel}_r(t):=r_{N(t-)+1}\qquad (t\in[t^\ast,T)).
\end{equation}
This agrees with the pre-event label at each event time. For two labellings $r,r'\in\{0,1\}^{\bbN}$, define the (window-restricted) Hamming distance
\[
d_H(r,r'):=\sum_{i=1}^{N(D)} \mathbf 1\{r_i\neq r_i'\}.
\]
Only the first $N(D)$ entries matter throughout. For $u\in[t^\ast,T]$ we also define the
$u$-restricted Hamming distance
\[
d_H^{[u]}(r,r'):=\sum_{i:\,t_i\le u} \mathbf 1\{r_i\neq r_i'\},
\qquad\text{so that}\qquad
d_H(r,r')=d_H^{[T]}(r,r').
\]
In what follows, all predictable integrands are constructed using only past labels $r_{1:N(t-)}$ on the relevant (cumulative) window; see the definition of $\widehat\lambda_k^r$ below.

\begin{lemma}[Predictability of the selected component]\label{lem:sel-predictable}
Fix a deterministic labelling $r$. Then $t\mapsto \mathrm{sel}_r(t)$ is $\{\mathcal H_{t-}\}$- and $\{\mathcal G_{t-}\}$-predictable on $[t^\ast,T)$ and, at event times $t_i<T$, $\mathrm{sel}_r(t_i-)=r_i$.
\end{lemma}

\begin{proof}
Since $N$ is adapted, c\`adl\`ag, and simple, $t\mapsto N(t-)$ is predictable and takes integer values with unit jumps. For each $m\in\bbN\cup\{0\}$, the set $\{(t,\omega):N(t-,\omega)=m\}$ is predictable, hence the indicator $t\mapsto \1\{N(t-)=m\}$ is predictable. On $[t^\ast,T)$ we have the pointwise identity with pairwise disjoint predictable supports
\[
\mathrm{sel}_r(t)=\sum_{m=0}^\infty r_{m+1} \1\{N(t-)=m\}.
\]
Because the supports are disjoint and for each $(\omega,t)$ exactly one indicator equals $1$, this is a countable pointwise sum of predictable indicators; hence it is predictable. Simplicity implies $N(t_i-)=i-1$, whence $\mathrm{sel}_r(t_i-)=r_i$.
\end{proof}

\begin{remark}[On predictability of $\mathrm{sel}_r$ for random iterate labels]\label{rem:sel-random}
Lemma~\ref{lem:sel-predictable} treats deterministic labellings $r$. For random iterate labellings $r^{(m)}$, the process $t\mapsto \mathrm{sel}_{r^{(m)}}(t)=r^{(m)}_{N(t-)+1}$ need not be $\mathcal G$-predictable. In what follows, $\mathrm{sel}_r$ is only used inside pathwise algebraic identities, and any resulting selection-remainder term is dominated by stochastic integrals with $\mathcal G$-predictable integrands built from past labels $r_{1:N(t-)}$.
\end{remark}

For $\tau=(t,x)$ and a labelling $r$, define the label-induced intensities using only past labels by
\[
\widehat\lambda_k^r(\tau;\theta):=\lambda_k\!\big(\tau\mid \mathcal H_{t-}\vee\sigma(r_{1:N(t-)});\theta\big),\qquad
\widehat\lambda^r:=\widehat\lambda_0^r+\widehat\lambda_1^r.
\]
If $r$ is deterministic then $\sigma(r_{1:N(t-)})$ is contained in $\sigma(N(t-))\subset \mathcal H_{t-}$ (it is not generally trivial because the random length $N(t-)$ is $\mathcal H_{t-}$-measurable); hence adjoining $\sigma(r_{1:N(t-)})$ does not change measurability, and $\widehat\lambda_k^r(\cdot;\theta)$ are $\mathcal H_{t-}$- and $\mathcal G_{t-}$-predictable. More generally, along the update path we assume the blockwise non-anticipation condition of Assumption~\ref{ass:G0pred}; in that case for each update $m$ and each $t\in[t^\ast,u_m)$ the vector of past labels $(r^{(m)}_1,\ldots,r^{(m)}_{N(t-)})$ is $\mathcal G_{t-}$-measurable and $\mathcal H_{t-}\vee\sigma(r^{(m)}_{1:N(t-)})\subseteq\mathcal G_{t-}$, so $\widehat\lambda_k^{r^{(m)}}(\cdot;\theta)$ are $\mathcal G_{t-}$-predictable on $D^{(m)}$ (and, more generally, on any $D^{[u]}$ with $u\le u_m$).
All intensities are left-continuous at event times, and every stochastic integral below is evaluated at pre-jump times.

For a labelling $r$ and measurable $A\subset \mathcal S$,
\[
\widehat N_k^r\big((t_1,t_2]\times A\big)
:=\sum_{i: t_1<t_i\le t_2, x_i\in A}\mathbf 1\{r_i=k\},\qquad k\in\{0,1\}.
\]
By convention, only past labels $r_{1:N(t-)}$ enter any predictable integrand at time $t$. For notational convenience, in sections where the true component counts $N_k$ are not in play we sometimes write $N_k^r:=\widehat N_k^r$ for the counting measure induced by the labelling $r$; we always reserve $N_k$ (without a
superscript) for the true component process.

For a candidate labelling $r$, the complete-data log-likelihood is
\begin{equation}\label{eq:cdll}
\ell_r(\theta)
=\sum_{i=1}^{N(D)} \log \widehat\lambda_{r_i}^r\!\big(\gamma_i;\theta\big)
-\int_D \sum_{k=0}^1 \widehat\lambda_k^r(\tau;\theta) \dd\tau,\qquad
f_r(\theta):=|D|^{-1}\ell_r(\theta).
\end{equation}

For $u\in[t^\ast,T]$ define the restricted complete-data log-likelihood and its normalised version by
\begin{align}
\ell_r^{[u]}(\theta)
&:=
\sum_{i:\,t_i\le u} \log \widehat\lambda_{r_i}^r\!\big(\gamma_i;\theta\big)
-\int_{D^{[u]}} \sum_{k=0}^1 \widehat\lambda_k^r(\tau;\theta) \dd\tau,
\label{eq:cdll-restricted}\\
f_r^{[u]}(\theta)
&:=|D^{[u]}|^{-1}\,\ell_r^{[u]}(\theta).
\nonumber
\end{align}
When $u=T$ these coincide with \eqref{eq:cdll}. When $u=t^\ast$ we have $D^{[u]}=\emptyset$ and $\ell_r^{[u]}(\theta)=0$ (a.s.); the normalised quantity
$f_r^{[u]}(\theta)=|D^{[u]}|^{-1}\ell_r^{[u]}(\theta)$ is only used for $u>t^\ast$.

For $\theta\in\Theta$ define the oracle and label-induced log-likelihood ratios
\[
s^{\rm or}_\theta(\tau):=\log\frac{\lambda_1(\tau\mid \mathcal G_{t-};\theta)}{\lambda_0(\tau\mid \mathcal G_{t-};\theta)},\qquad
\tilde s_\theta(\tau;r):=\log\frac{\widehat\lambda_1^r(\tau;\theta)}{\widehat\lambda_0^r(\tau;\theta)}.
\]
Given $b>0$, the decisive sets and ambiguous band are
\[
S_\theta^+(b)=\{ \tau:\ s^{\rm or}_\theta(\tau)\ge b \},\quad
S_\theta^-(b)=\{ \tau:\ s^{\rm or}_\theta(\tau)\le -b \},\quad
A_\theta(b)=\{ \tau:\ |s^{\rm or}_\theta(\tau)|\le b \}.
\]
Since $s_\theta^{\rm or}$ is $\mathcal G$-predictable, the indicator processes $(\omega,\tau)\mapsto \1_{S_\theta^\pm(b)}(\omega,\tau)$ and $(\omega,\tau)\mapsto \1_{A_\theta(b)}(\omega,\tau)$ are $\mathcal G$-predictable.

\begin{definition}[Minority label on decisive sets]\label{def:minority}
Fix $\theta$ and $b>0$. For an event $\gamma_i$:
\[
\text{minority}(\gamma_i;\theta,b)=
\begin{cases}
0,& \gamma_i\in S_\theta^+(b),\\
1,& \gamma_i\in S_\theta^-(b),\\
\text{undefined},& \gamma_i\notin S_\theta^+(b)\cup S_\theta^-(b).
\end{cases}
\]
\end{definition}

\begin{remark}[Oracle-aided viewpoint]\label{rem:oracle}
The penalty below uses the oracle decisive sets $S_\theta^\pm(b)$; it is therefore not implementable from observed data. We analyse the resulting oracle-aided hard--EM purely to derive sufficient structural conditions for contraction and consistency.
\end{remark}

\subsection{Model-agnostic structural assumptions}\label{sec:assumptions}

For each window $D$, fix a radius $\varepsilon_{\rm net}\le |D|^{-1}$ and an $\varepsilon_{\rm net}$-net $\mathcal N_{\rm net}\subset\Theta_\circ$; let $N_{\rm net}:=|\mathcal N_{\rm net}|$.
Since the parameter dimension $p$ is fixed and $\Theta_\circ$ is compact, $N_{\rm net}$ grows at most polynomially in $|D|$; in particular $\log N_{\rm net}=O(\log|D|)$. Throughout we write
\begin{equation}\label{eq:z-def}
z_{\rm Fr}:=\log\!\Big(c_0 (1+\eta_{\rm act}) C_{\rm act} |D| \Tmax (p+2) N_{\rm net}\Big) + \eta_x\log|D|,
\end{equation}
and for concreteness we take $c_0=10$ and $\eta_x=2$. This choice absorbs two-sided Freedman inequality factors, the covering number $N_{\rm net}$, and small constant multipliers in the master union bound over updates, components, coordinates, and net points, as well as similar auxiliary Freedman inequality calls used in the appendices, yielding probabilities at least $1-|D|^{-\eta''}$ for some $\eta''>0$. Since $p$ is fixed, $z_{\rm Fr}=\Theta(\log|D|)$.

On $E_{\rm act}$ (defined in \cref{ass:G0}), the blockwise operator of Algorithm~\ref{alg:blockwise-hardem} assigns labels progressively and never revises past events. Hence the total number of accepted eventwise label assignments over the full update path is at most $N(D)$, and our union bounds and flip-count accounting refer to those progressive assignments. In the martingale analysis we apply Freedman's inequality to a finite collection of scalar martingales (e.g.\ per-event terms and gradient-coordinate terms at M-step anchors) indexed by update/block indices, event indices, components, coordinates, and net points. The factor $\Tmax$ in \eqref{eq:z-def} is a conservative envelope that absorbs finite unions over update indices and ancillary Freedman calls (e.g.\ those used in appendices), in addition to the $O(|D|)$ eventwise decisions controlled on $E_{\rm act}$.

To make the Freedman inequality bounds uniform over $\theta\in\Theta_\circ$, we first prove them at each net point $\theta\in\mathcal N_{\rm net}$ and then extend to all of $\Theta_\circ$ using deterministic Lipschitz bounds in $\theta$ (coming from the bounded gradients in \cref{ass:G3} and the uniform lower margins in \cref{ass:G1}). Since $\varepsilon_{\rm net}\le |D|^{-1}$, the resulting deterministic discretization error is $O(|D|^{-1})$ at the $f_r$ scale and can be absorbed into the residual terms.

Thus the total number of Freedman inequality tail events we union bound over is at most a constant multiple of
\[
(p+2)\,N_{\rm net}\,\sum_{m=1}^{\Tmax} N^{(m)}.
\]
Since $N^{(m)}\le N(D)$ for every $m$, we have $\sum_{m=1}^{\Tmax}N^{(m)}\le \Tmax\,N(D)$, and on $E_{\rm act}$ this is
bounded by $(1+\eta_{\rm act})C_{\rm act}|D|\,\Tmax$. Consequently, choosing $z_{\rm Fr}$ as in \eqref{eq:z-def} (which already
contains the factor $|D|\Tmax$) makes the overall union bound failure probability polynomially small in $|D|$.

\begin{assumption}[Activity and tails]\label{ass:G0}
There exists $C_{\rm act}<\infty$ and $\eta_{\rm act}\in(0,1)$ such that $\bbE N(D)\le C_{\rm act}|D|$ and
\[
\bbP\!\big(N(D)\ge (1+\eta_{\rm act})C_{\rm act} |D|\big)\le e^{-c|D|}
\]
for some $c>0$ and all sufficiently large $|D|$. Denote
\[
E_{\rm act}:=\big\{N(D)\le (1+\eta_{\rm act})C_{\rm act} |D|\big\}.
\]
\end{assumption}

\begin{assumption}[Predictable blockwise updates]\label{ass:G0pred}
Fix a deterministic partition $t^\ast=u_0<u_1<\cdots<u_{\Tmax}=T$ as in \S\ref{sec:setup}.
Along the update path $\{(r^{(m)},\theta^{(m)}): m=0,\dots,\Tmax\}$ the following hold:

\begin{enumerate}[nosep,label=(\alph*)]
\item \textbf{(Parameter timing)} For each $m\in\{0,\dots,\Tmax\}$, the update parameter
$\theta^{(m)}$ is $\mathcal G_{u_m}$-measurable.

\item \textbf{(Progressive, non-revisiting labels)} For each $m\in\{1,\dots,\Tmax\}$ and each event index $i$ with
$t_i\le u_m$, the label $r_i^{(m)}$ is $\mathcal G_{t_i}$-measurable, and labels are never revised:
for all $\ell\ge m$, $r_i^{(\ell)}=r_i^{(m)}$ for every $i$ with $t_i\le u_m$.
In particular, for each $m$ the vector $(r_1^{(m)},\dots,r_{N^{(m)}}^{(m)})$ is $\mathcal G_{u_m}$-measurable.
\end{enumerate}

In particular, for every $m$ and every $t\in[t^\ast,u_m)$ the vector of past labels
$(r^{(m)}_1,\ldots,r^{(m)}_{N(t-)})$ is $\mathcal G_{t-}$-measurable and hence
\[
\mathcal H_{t-}\vee\sigma(r^{(m)}_{1:N(t-)})\ \subseteq\ \mathcal G_{t-}.
\]
Consequently, for every $m$, every $u\le u_m$, and every $\theta\in\Theta_\circ$,
the label-induced intensities $\widehat\lambda_k^{r^{(m)}}(\cdot;\theta)$ are $\mathcal G_{t-}$-predictable on $D^{[u]}$.
\end{assumption}

\begin{assumption}[Uniform margins over labellings]\label{ass:G1}
There exist $\underline\mu_\Sigma>0$ and $\underline\mu_k>0$ ($k=0,1$) such that, for a.e.\ $\tau\in D$, all $\theta\in\Theta_\circ$, and all deterministic labellings $r$,
\[
\widehat\lambda_0^r(\tau;\theta)+\widehat\lambda_1^r(\tau;\theta)\ \ge\ \underline\mu_\Sigma,\qquad
\widehat\lambda_k^r(\tau;\theta)\ \ge\ \underline\mu_k.
\]
Write $\underline\mu_{\min}:=\min\{\underline\mu_0,\underline\mu_1\}$.
\end{assumption}

\begin{assumption}[Oracle positivity]\label{ass:oracle-positivity}
There exist $\underline\lambda_k>0$ such that $\lambda_k(\tau\mid\mathcal G_{t-};\theta)\ge \underline\lambda_k$ for all $\tau\in D$, $k\in\{0,1\}$, and $\theta\in\Theta_\circ$. Hence $s_\theta^{\rm or}=\log(\lambda_1/\lambda_0)$ is well-defined.
\end{assumption}

\begin{assumption}[Window envelope along the iterate path]\label{ass:G2prime}
There exist deterministic model constants $L>0$ (temporal memory), $R\ge 0$ (spatial radius; $R=0$ if not applicable), and $C_{\rm base},C_{\rm loc}<\infty$ (depending only on the model class and $\Theta_\circ$). For $n\in\bbN$, define the window-count event
\[
\Omega_{n}(L,R):=\left\{\ \sup_{\tau=(t,x)\in D}\ N\Big( (t-L,t)\times B_R(x) \Big)\ \le\ n\ \right\}.
\]
There exist constants $c_1,c_2>0$ and an exponent $\eta_{\rm win}>0$ such that, with $n_*(|D|):=\lceil c_1\log|D|\rceil$,
\begin{equation}\label{eqn:etaWin}
    \bbP\big(\Omega_{n_*}(L,R)^c\big)\ \le\ c_2 |D|^{-\eta_{\rm win}}. 
\end{equation}

On the event $\Omega_{n_*}(L,R)$ we have, simultaneously for all $\theta\in\Theta_\circ$, for the oracle intensities $\lambda^\star$ and for all labellings $r\in \mathcal R_{\rm path}^{\rm flip}$,
\[
\sup_{\tau\in D}\ \big(\lambda_0^\star+\lambda_1^\star\big)(\tau)\ \le\ K_{\rm win},\qquad
\sup_{\tau\in D}\ \big(\widehat\lambda_0^r+\widehat\lambda_1^r\big)(\tau;\theta)\ \le\ K_{\rm win},
\]
where
\[
K_{\rm win}:=C_{\rm base}+C_{\rm loc} n_*(|D|).
\]
In particular, on $\Omega_{n_*}(L,R)$ we have $K_{\rm win}=O(\log|D|)$.
\end{assumption}

\begin{remark}
In Assumption~\ref{ass:G2prime} we only require the existence of $\eta_{\rm win}>0$ such that \cref{eqn:etaWin} holds, thereby providing a polynomially small tail for the window event. This can then be combined with the martingale errors that are controlled  via Freedman's inequality. In particular, the exponent $\eta_x>1$ that appears in the definition of the Freedman inequality level $z_{\rm Fr}$ is not constrained by $\eta_{\rm win}$.
\end{remark}

\begin{assumption}[Single-flip locality and influence]\label{ass:G3}
For any single flip at an event location $\zeta$, with $\Delta\widehat\lambda^r=\widehat\lambda^{r'}-\widehat\lambda^{r}$ and $\Delta\widehat\lambda_{k}^r=\widehat\lambda_{k}^{r'}-\widehat\lambda_{k}^{r}$, there exist finite constants
\[
B_\infty,\ B_1,\ B_2,\qquad B_\infty^{\rm comp},\ B_1^{\rm comp},\ B_2^{\rm comp},\qquad B_1^\Sigma,\ B_2^\Sigma
\]
such that
\[
\|\Delta\widehat\lambda^r\|_\infty\le B_\infty,\quad \int_D |\Delta\widehat\lambda^r| \dd\tau\le B_1,\quad \int_D \frac{(\Delta\widehat\lambda^r)^2}{\underline\mu_\Sigma} \dd\tau\le B_2,
\]
and, for $k\in\{0,1\}$,
\[
\|\Delta\widehat\lambda_k^r\|_\infty\le B_\infty^{\rm comp},\quad \int_D |\Delta\widehat\lambda_k^r| \dd\tau\le B_1^{\rm comp},\quad \int_D \frac{(\Delta\widehat\lambda_k^r)^2}{\underline\mu_\Sigma} \dd\tau\le B_2^{\rm comp},
\]
as well as $\sum_k \int_D |\Delta\widehat\lambda_k^r| \dd\tau \le B_1^\Sigma$ and $\sum_k \int_D \frac{(\Delta\widehat\lambda_k^r)^2}{\underline\mu_\Sigma} \dd\tau \le B_2^\Sigma$.

In addition (and used when $\nabla\log\widehat\lambda_k$ denominators appear), there exist finite constants $\widetilde B_2^{\rm comp}<\infty$ and $\widetilde B_2^\Sigma<\infty$ such that, for each $k$,
\[
\int_D \frac{(\Delta\widehat\lambda_k^r)^2}{\underline\mu_k}\,\dd\tau \ \le\ \widetilde B_2^{\rm comp},
\qquad
\sum_{k=0}^1 \int_D \frac{(\Delta\widehat\lambda_k^r)^2}{\underline\mu_k}\,\dd\tau \ \le\ \widetilde B_2^\Sigma.
\]

In terms of the gradient-level analogues, there exist finite constants
\[
\overline B_\infty^{\rm comp},\ \overline B_1^{\rm comp},\ \overline B_2^{\rm comp},\qquad \overline B_1^\Sigma
\]
such that, for each component $k$,
\[
\|\Delta(\nabla_\theta \widehat\lambda_k^r)\|_\infty\le \overline B_\infty^{\rm comp},\quad
\int_D \|\Delta(\nabla_\theta \widehat\lambda_k^r)\| \dd\tau\le \overline B_1^{\rm comp},\quad
\int_D \frac{\|\Delta(\nabla_\theta \widehat\lambda_k^r)\|^2}{\underline\mu_\Sigma} \dd\tau\le \overline B_2^{\rm comp},
\]
and $\sum_k\int_D \|\Delta(\nabla_\theta \widehat\lambda_k^r)\| \dd\tau\le \overline B_1^\Sigma$. 

Define the total-gradient difference $\Delta(\nabla_\theta \widehat\lambda^{r}):=\sum_{k=0}^1 \Delta(\nabla_\theta \widehat\lambda_k^{r})$. Then, for notational convenience, introduce finite constants $\overline B_\infty^\Sigma$ and $\overline B_2^\Sigma$ such that
\[
\|\Delta(\nabla_\theta \widehat\lambda^{r})\|_\infty\le \overline B_\infty^\Sigma,
\qquad
\int_D \frac{\|\Delta(\nabla_\theta \widehat\lambda^{r})\|^2}{\underline\mu_\Sigma}\,\dd\tau\le \overline B_2^\Sigma.
\]
For instance, one may take $\overline B_\infty^\Sigma:=2\overline B_\infty^{\rm comp}$ and $\overline B_2^\Sigma:=4\overline B_2^{\rm comp}$.

Furthermore, there exists $S_\infty<\infty$ and $H_\infty < \infty$ such that
\[
\operatorname*{ess sup}_{\tau,k,\theta,r}\ \|\nabla_\theta \widehat\lambda_k^r(\tau;\theta)\|\ \le\ S_\infty,
\qquad
\operatorname*{ess sup}_{\tau,k,\theta,r}\ \|\nabla^2_{\theta\theta} \widehat\lambda_k^r(\tau;\theta)\|\ \le\ H_\infty.
\]
The same bounds hold when $\widehat\lambda_k^r(\tau;\theta)$ is replaced by the oracle intensity
$\lambda_k(\tau\mid\mathcal G_{t-};\theta)$.

Finally, there exist finite constants $\widetilde{\overline B}_2^{\rm comp}<\infty$ such that, for each $k$,
\[
\int_D \frac{\|\Delta(\nabla_\theta \widehat\lambda_k^r)\|^2}{\underline\mu_k}\,\dd\tau \ \le\ \widetilde{\overline B}_2^{\rm comp}.
\]
\end{assumption}

\begin{assumption}[Uniform score alignment]\label{ass:G4}
There exist a deterministic sequence \(\Delta_s(|D|)\ge 0\), an exponent \(\eta>0\), and events
\(E^{\rm align}_{|D|}\) with \(\mathbb P(E^{\rm align}_{|D|})\ge 1-|D|^{-\eta}\) such that, on \(E^{\rm align}_{|D|}\), uniformly over \(\theta\in\Theta_\circ\), all deterministic labellings \(r\), and a.e.\ \(\tau\in D\),
\[
\big|\tilde s_\theta(\tau;r)-s^{\rm or}_\theta(\tau)\big|\ \le\ \Delta_s(|D|).
\]
For brevity, on \(E^{\rm align}_{|D|}\) we write \(\Delta_s:=\Delta_s(|D|)\).
\end{assumption}

\begin{remark}[Consequence on decisive sets]
On the event \(E^{\rm align}_{|D|}\), Assumption~\ref{ass:G4} implies that, for any \(b>0\),
\[
\tau\in S_\theta^+(b)\Rightarrow \tilde s_\theta(\tau;r)\ge b-\Delta_s,
\qquad
\tau\in S_\theta^-(b)\Rightarrow \tilde s_\theta(\tau;r)\le -b+\Delta_s.
\]
In particular, if \(b\ge \Delta_s\) then \(\tilde s_\theta(\cdot;r)\) and \(s_\theta^{\rm or}(\cdot)\) have the same sign
on \(S_\theta^+(b)\cup S_\theta^-(b)\).
\end{remark}

\begin{assumption}[LLR Lipschitz in $\theta$ (average-case)]\label{ass:G5}
There exist a deterministic sequence $L_s(|D|)<\infty$, an $\eta>0$, and events $E^{\rm Lip}_{|D|}$ with $\mathbb{P}(E^{\rm Lip}_{|D|})\ge 1-|D|^{-\eta}$ such that, for all $\theta,\theta'\in\Theta_{\rm sc}$,
\[
\int_D |s^{\rm or}_\theta(\tau)-s^{\rm or}_{\theta'}(\tau)|\,\lambda^\star(\tau)\,d\tau
\le L_s(|D|)\,|D|\,\|\theta-\theta'\|.
\]

Moreover, on the same event $E^{\rm Lip}_{|D|}$ the bound holds uniformly on the deterministic schedule subwindows:
for every $m\in\{1,\dots,\Tmax\}$ and all $\theta,\theta'\in\Theta_{\rm sc}$,
\begin{align*}
\int_{D^{(m)}} \big| s^{\rm or}_\theta(\tau)-s^{\rm or}_{\theta'}(\tau) \big| \lambda^\star(\tau) \dd\tau
\ \le&\ L_s(|D|)\, |D^{(m)}| \|\theta-\theta'\|,\\
\int_{D_m} \big| s^{\rm or}_\theta(\tau)-s^{\rm or}_{\theta'}(\tau) \big| \lambda^\star(\tau) \dd\tau
\ \le&\ L_s(|D|)\, |D_m| \|\theta-\theta'\|.
\end{align*}

\end{assumption}

On the event $E^{\rm Lip}_{|D|}$ in Assumption~\ref{ass:G5}, we write $L_s:=L_s(|D|)$ for brevity.

\begin{assumption}[Local strong concavity \& near-zero score at $r^\star$]\label{ass:G6}
There exist constants $r_{\rm sc}>0$, $m_{\rm sc}>0$, $C_{\rm sc}<\infty$, and $\eta>0$ such that, for all sufficiently large $|D|$, there is an event $E^{\rm sc}_{|D|}$ with $\bbP(E^{\rm sc}_{|D|})\ge 1-|D|^{-\eta}$ on which:

\begin{enumerate}[nosep,label=(\alph*)]
\item \textbf{(Uniform local curvature)}
$f_{r^\star}$ is differentiable on
\[
\Theta_{\rm sc}:=\Theta_\circ\cap B(\theta^\star,r_{\rm sc}),
\qquad
B(\theta^\star,r_{\rm sc}):=\{\theta\in\bbR^p:\ \|\theta-\theta^\star\|\le r_{\rm sc}\},
\]
and satisfies the gradient monotonicity inequality
\begin{equation}\label{eq:RSC-grad}
\big(\nabla f_{r^\star}(\theta')-\nabla f_{r^\star}(\theta)\big)^\top(\theta'-\theta)
\ \le\ -m_{\rm sc}\|\theta'-\theta\|^2
\qquad\forall\theta,\theta'\in\Theta_{\rm sc}.
\end{equation}
If $f_{r^\star}$ is twice differentiable, a sufficient condition for \eqref{eq:RSC-grad} is the uniform Hessian bound
$-\nabla^2 f_{r^\star}(\theta)\succeq m_{\rm sc}I_p$ for all $\theta\in\Theta_{\rm sc}$.

\item \textbf{(Score at the truth)} 
\[
\|\nabla f_{r^\star}(\theta^\star)\|\ \le\ C_{\rm sc}\!\left(\sqrt{\frac{K_{\rm win} z_{\rm Fr}}{|D|}} + \frac{z_{\rm Fr}}{|D|}\right)\!,
\]
where $K_{\rm win}$ is as in \cref{ass:G2prime}. In bounded-intensity settings, the score bound in (b) follows from Lemma~\ref{lem:score-nearzero} on $\Omega_{n_*}(L,R)$ by choosing $z_{\rm Fr}\simeq \log|D|$ (and using $\bbP(\Omega_{n_*}(L,R)^c)\le c_2|D|^{-\eta_{\rm win}}$ from Assumption~\ref{ass:G2prime}).

\item \textbf{(Uniformity over deterministic prefixes)} 
For every $m\in\{1,\dots,\Tmax\}$, the conclusions in (a) and (b) also hold with the full-window objective
$f_{r^\star}$ replaced by the prefix objective $f_{r^\star}^{[u_m]}$, with the same constants
$r_{\rm sc},m_{\rm sc},C_{\rm sc}$, and with $|D|$ replaced by $|D^{(m)}|=|D^{[u_m]}|$ in the score bound.

\end{enumerate}

Further, assume $\theta^\star$ lies in a compact subset of $\Theta_\circ$ bounded away from any non-identifiability/collision boundary. Equivalently, there exists $\delta_{\rm sep}>0$ such that $\mathrm{dist}(\theta^\star,\partial\Theta_{\rm id})\ge \delta_{\rm sep}$,
and choose $r_{\rm sc}\le \delta_{\rm sep}/2$.
\end{assumption}

\begin{remark}[On concavity restrictions]\label{rem:concavity-beta}
In Hawkes models, global log-likelihood concavity typically holds for parameters that enter linearly (e.g.\ baselines and excitation amplitudes) but can fail for nonlinear parameters such as the decay rate $\beta$. Accordingly, rather than imposing global concavity on $\Theta_\circ$, we use the local condition in Assumption~\ref{ass:G6}: it suffices that $f_{r^\star}$ is uniformly strongly concave on a deterministic-radius neighborhood $\Theta_{\rm sc}=\Theta_\circ\cap B(\theta^\star,r_{\rm sc})$, and we then keep all iterates inside $\Theta_{\rm sc}$ via an invariant-set argument.
\end{remark}

\begin{assumption}[Warm start]\label{ass:warmstart}
The initialization satisfies $\|\theta^{(0)}-\theta^\star\|\le r_{\rm sc}$ (equivalently $\theta^{(0)}\in\Theta_{\rm sc}$).
\end{assumption}

\begin{remark}[Why Assumption~\ref{ass:warmstart} is reasonable]\label{rem:warmstart}
This warm-start condition motivates the requirement that the process/history up to \(t^\star\) is observed: if one can construct a pilot estimator \(\tilde\theta_0\) from the data on \([0,t^\star]\times\mathcal S\) such that $\tilde\theta_0$ is consistent for the initial parameter $\theta_0^\star$ (and in our setting $\theta_0^\star=\theta^\star$), then for sufficiently large sample size (or amount of data) we have $\|\tilde\theta_0-\theta^\star\|\le r_{\rm sc}$ with high probability. Thus, taking $\theta^{(0)}:=\tilde\theta_0$ makes Assumption~\ref{ass:warmstart} operational.
\end{remark}

\begin{definition}[Canonical telescoping path]\label{def:telescoping-path}
Given two labellings $r,r^\star\in\{0,1\}^{N(D)}$, let
$I:=\{ i\in\{1,\dots,N(D)\}:\ r_i\neq r_i^\star \}$ and $\mathfrak{I}:=|I|$.
Fix a deterministic ordering of $I$ (e.g., increasing index), write
$I=(i_1,\ldots,i_\mathfrak{I})$, and define a sequence $(r^{(j)})_{j=0}^{\mathfrak I}$ by $r^{(0)}:=r$,

\[r^{(j)}:=\text{ the labelling obtained from } r^{(j-1)} \text{ by flipping only coordinate } i_j
\quad (j=1,\ldots,\mathfrak{I}).
\]
We call $(r^{(0)},r^{(1)},\ldots,r^{(\mathfrak{I})})$ the canonical telescoping path
from $r$ to $r^\star$. It has length $\mathfrak{I}=d_H(r,r^\star)$ and ends at $r^{(\mathfrak{I})}=r^\star$. For any functional $F$ on labellings,
\[
F(r)-F(r^\star)=\sum_{j=1}^\mathfrak{I} \big(F(r^{(j-1)})-F(r^{(j)})\big).
\]
\end{definition}

\begin{assumption}[Label-to-score Lipschitz]\label{ass:G7prime}
There exist constants $\bar C_0,\bar C_1,\bar C_2<\infty$ and $\eta>0$ such that, for all large $|D|$, on an event
$E^{\nabla}_{|D|}$ with $\bbP(E^{\nabla}_{|D|})\ge 1-|D|^{-\eta}$ the following holds simultaneously:

for every update index $m\in\{0,\dots,\Tmax-1\}$, letting $u:=u_{m+1}$ and $r:=r^{(m+1)}$, for every choice of
restricted M-step maximiser
\[
\hat\theta_r^{[u]}\ \in\ \arg\max_{\theta\in\Theta_{\rm sc}}\ell_r^{[u]}(\theta),
\]
we have
\[
\left\Vert \nabla f_{r}^{[u]}(\hat\theta_r^{[u]}) - \nabla f_{r^\star}^{[u]}(\hat\theta_r^{[u]})\right\Vert
\ \le\ 
\frac{d_H^{[u]}(r,r^\star)}{|D^{[u]}|}
\Big(\bar C_0 K_{\rm win} + \bar C_1\sqrt{K_{\rm win} z_{\rm Fr}} + \bar C_2 z_{\rm Fr}\Big).
\]
\end{assumption}

\begin{definition}[Decisive-set disagreement set]\label{def:Btheta}
For $b>0$ and any $\theta\in\Theta_\circ$, define
\[
B_\theta\ :=\ \big(S_{\theta^\star}^+(b) \Delta S_{\theta}^+(b)\big)\ \cup\ \big(S_{\theta^\star}^-(b) \Delta S_{\theta}^-(b)\big).
\]
\end{definition}

\begin{assumption}[Mass on decisive sets and oracle-set concentration]\label{ass:G8}
All expectations and probabilities in this assumption are with respect to $P_{\theta^\star}$.
There exist \(\eta>0\), \(c>0\), \(C_{\rm it}>0\), and, for each threshold \(b=b_{|D|}>0\) used in the theorem under consideration, nonnegative sequences
\[
\eta_+(|D|;b),\qquad \eta_-(|D|;b),
\]
such that \(\Tmax=\lceil C_{\rm it}\log|D|\rceil\), and the following hold.
\begin{enumerate}[nosep,label=(\roman*)]

\item \textbf{(Minority mass on oracle decisive sets at threshold \(b\)).}
\[
\bbE N_0\!\big(S_{\theta^\star}^+(b)\big)\ \le\ \eta_+(|D|;b)\,|D|,
\qquad
\bbE N_1\!\big(S_{\theta^\star}^-(b)\big)\ \le\ \eta_-(|D|;b)\,|D|.
\] Since $S_{\theta^\star}^\pm(b)\subseteq S_{\theta^\star}^\pm(b_{\rm th})$
for $b\ge b_{\rm th}$, the bounds for $b_{\rm th}$ imply the same bounds for $b$.

\item \textbf{(Concentration for fixed oracle sets and iterate-dependent disagreement sets).}
Moreover, for all $\delta\in(0,1)$ and for $S$ ranging over the oracle sets
\[
\mathcal F_{\rm fixed}(b):=\big\{S_{\theta^\star}^+(b), S_{\theta^\star}^-(b), A_{\theta^\star}(b), A_{\theta^\star}(2b)\big\},
\]
we have oracle-set concentration simultaneously for total and component counts around their means:
\begin{align*}
    \mathbb{P}\!\left(\left|N(S)-\bbE N(S)\right|\ge \delta|D|\right)\le& \exp\left(-c \frac{\delta^2|D|}{K_{\rm win}}\right),\\
\mathbb{P}\!\left(\left|N_k(S)-\bbE N_k(S)\right|\ge \delta|D|\right)\le& \exp\left(-c \frac{\delta^2|D|}{K_{\rm win}}\right)
\end{align*}

for $k\in\{0,1\}$.

Define the mean-deviation event
\begin{align*}
E^{\rm mean}_{|D|}
&:=
\bigcap_{S\in\mathcal F_{\rm fixed}(b)}
\left\{
\left|N(S)-\mathbb{E} N(S)\right|\le C_{\rm it}\sqrt{K_{\rm win}|D|\log|D|}
\right\} \\
&\quad\cap
\bigcap_{k=0}^1\bigcap_{S\in\mathcal F_{\rm fixed}(b)}
\left\{
\left|N_k(S)-\mathbb{E} N_k(S)\right|\le C_{\rm it}\sqrt{K_{\rm win}|D|\log|D|}
\right\}.
\end{align*}
Assume $\bbP(E^{\rm mean}_{|D|})\ge 1-|D|^{-\eta}$ for some $\eta>0$ (finite union over $\mathcal F_{\rm fixed}(b)$).

We also require concentration around compensators for the same fixed oracle sets and for the iterate-dependent disagreement sets.
Define $E^{\rm comp}_{|D|}$ to be the event on which the following hold simultaneously:
\begin{enumerate}[nosep,label=(\alph*)]
\item for all $S\in\mathcal F_{\rm fixed}(b)$ and all $k\in\{0,1\}$,
\begin{align*}
    \Big|N_k(S)-\int_S \lambda_k^\star(\tau)\dd\tau\Big|
\le& C_{\rm it}\sqrt{K_{\rm win}|D|\log|D|},\\
\Big|N(S)-\int_S\lambda^\star(\tau)\dd\tau\Big|
\le& C_{\rm it}\sqrt{K_{\rm win}|D|\log|D|};
\end{align*}

\item for all iterates $m\in\{0,\dots,\Tmax-1\}$,
\[
\Big|N(B_{\theta^{(m)}}\cap D_{m+1})
-\int_{B_{\theta^{(m)}}\cap D_{m+1}}\lambda^\star(\tau)\dd\tau\Big|
\le C_{\rm it}\sqrt{K_{\rm win}|D_{m+1}|\log|D|}.
\]

\end{enumerate}
Assume $\bbP(E^{\rm comp}_{|D|})\ge 1-|D|^{-\eta}$ for some (possibly different) $\eta>0$. Finally set
\[
E^{\rm mass}_{|D|}\ :=\ E^{\rm mean}_{|D|}\cap E^{\rm comp}_{|D|},
\]
and assume $\bbP(E^{\rm mass}_{|D|})\ge 1-|D|^{-\eta}$ for some $\eta>0$.

\end{enumerate}

All expectation and concentration statements above are also assumed to hold uniformly when the observation
window is replaced by any deterministic prefix window $D^{(m)}$ or deterministic block $D_m$ from the schedule.
Equivalently, we may replace any set $S\subset D$ appearing above by $S\cap W$ and replace $|D|$ by $|W|$ for any
$W\in\{D^{(m)},D_m\}$, with the same constants (and with the corresponding $N(W)$-type counts).
In particular, we include in $E^{\rm comp}_{|D|}$ the compensator deviations for the restricted fixed oracle sets $S\cap D_m$ (for $S\in\mathcal F_{\rm fixed}(b)$) and, for iterate-dependent disagreement sets, the blockwise deviations on $B_{\theta^{(m)}}\cap D_{m+1}$ as in (b). Since the schedule is deterministic and has size $O(\Tmax)$, these additional requirements are absorbed by a finite union.

\end{assumption}

\begin{lemma}[Existence of M-step maximisers]\label{lem:existence-M}
Work on the window-envelope event $\Omega_{n_*}(L,R)$ from Assumption~\ref{ass:G2prime}.
Let $r\in\mathcal R_{\rm path}$ be any labelling on the iterate path and let $u\in[t^\ast,T]$.

\begin{enumerate}[nosep,label=(\alph*)]
\item The map $\theta\mapsto \ell_r^{[u]}(\theta)$ is continuous on $\Theta_\circ$ and hence attains a maximum on $\Theta_\circ$.
Consequently $\operatorname{ArgMax}_{\Theta_\circ}^{[u]}(r)$ is nonempty.

\item Let $\Theta_{\rm sc}:=\Theta_\circ\cap B(\theta^\star,r_{\rm sc})$ as in Assumption~\ref{ass:G6} and
suppose $\Theta_{\rm sc}\neq\emptyset$.
Then $\theta\mapsto \ell_r^{[u]}(\theta)$ also attains a maximum on $\Theta_{\rm sc}$, so
$\operatorname{ArgMax}_{\Theta_{\rm sc}}^{[u]}(r)$ is nonempty.

\item Taking $u=T$ yields the same conclusions for the full-window objective $\ell_r=\ell_r^{[T]}$.
\end{enumerate}
\end{lemma}

\begin{proof}
Fix $u\in[t^\ast,T]$ and $r\in\mathcal R_{\rm path}$.
By Assumption~\ref{ass:G1}, $\widehat\lambda_{r_i}^r(\gamma_i;\theta)\ge \underline\mu_{\min}$ for all $i$ and $\theta$,
so $\theta\mapsto \log\widehat\lambda_{r_i}^r(\gamma_i;\theta)$ is continuous.

On $\Omega_{n_*}(L,R)$, Assumption~\ref{ass:G2prime} provides a deterministic bound $K_{\rm win}$ such that
$\sum_{k=0}^1\widehat\lambda_k^r(\tau;\theta)\le K_{\rm win}$ for all $\tau\in D$, all $\theta\in\Theta_\circ$,
and all $r\in\mathcal R_{\rm path}$. Since $\theta\mapsto \widehat\lambda_k^r(\tau;\theta)$ is continuous for each $\tau$,
dominated convergence (dominated by $K_{\rm win}$ on the finite-measure set $D^{[u]}$) implies continuity of
\[
\theta\ \longmapsto\ \int_{D^{[u]}} \sum_{k=0}^1 \widehat\lambda_k^r(\tau;\theta)\,\dd\tau .
\]
Therefore $\ell_r^{[u]}$ is continuous on $\Theta_\circ$. Since $\Theta_\circ$ is compact,
$\ell_r^{[u]}$ attains a maximum on $\Theta_\circ$.

The same argument applies on the compact set $\Theta_{\rm sc}=\Theta_\circ\cap B(\theta^\star,r_{\rm sc})$.
Finally, taking $u=T$ recovers $\ell_r=\ell_r^{[T]}$.
\end{proof}

\begin{assumption}[M-step optimality condition and differentiability]\label{ass:G9}
For every labelling $r$ on the iterate path $\mathcal R_{\rm path}$ and every $u\in[t^\ast,T]$,
the objective $f_r^{[u]}$ admits a $C^1$ extension to an open neighbourhood of $\Theta_\circ$.

For any convex constraint set $K\in\{\Theta_\circ,\Theta_{\rm sc}\}$ and any maximiser
$\hat\theta_K^{[u]}(r)\in\arg\max_{\theta\in K}\ell_r^{[u]}(\theta)$, the (global) first-order optimality condition holds:
\[
(\theta-\hat\theta_K^{[u]}(r))^\top \nabla f_r^{[u]}(\hat\theta_K^{[u]}(r))\ \le\ 0
\qquad\text{for all }\theta\in K.
\]
In particular, if $\hat\theta_K^{[u]}(r)\in\operatorname{int}(K)$ then
$\nabla f_r^{[u]}(\hat\theta_K^{[u]}(r))=0$.
Taking $u=T$ recovers the same condition for the full-window objective $f_r$.
\end{assumption}

\begin{remark}[Why the variational inequality in Assumption~\ref{ass:G9} is automatic]
The feasible set $\Theta_{\rm sc}:=\Theta_\circ\cap B(\theta^\star,r_{\rm sc})$ is convex since $\Theta_\circ$ is convex. Negative definiteness of the Hessian (equivalently, local strong concavity) is only required on \(\Theta_{\rm sc}\) (Assumption~\ref{ass:G6}); possible degeneracies elsewhere in \(\Theta_\circ\) are irrelevant provided the iterates remain in $\Theta_{\rm sc}$ (Lemma~\ref{lem:inv-sc}). For any convex constraint set $K$ and any maximiser $\hat\theta_K(r)\in\arg\max_{\theta\in K}\ell_r(\theta)$, the standard first-order necessary condition for constrained maximisation gives
\[
(\theta-\hat\theta_K(r))^\top \nabla f_r(\hat\theta_K(r))\le 0\qquad\text{for all }\theta\in K.
\]
In particular, if $\hat\theta_K(r)\in \operatorname{int}(K)$ then $\nabla f_r(\hat\theta_K(r))=0$.
\end{remark}

\begin{remark}[Optional margin assumption for the band]\label{rem:margin}
In some applications, one can assume the oracle margin: there exist $\alpha_{\rm band}>0$ and $C_{\rm band}<\infty$ such that for all small $u>0$,
\[
\frac{1}{|D|}\int_D \1\{|s_{\theta^\star}^{\rm or}(\tau)|\le u\} \lambda^\star(\tau) \dd\tau\ \le\ C_{\rm band} u^{\alpha_{\rm band}}.
\]
Under this assumption and the oracle-set concentration in \cref{ass:G8}, with high probability \[N(A_{\theta^\star}(2b))/|D|\lesssim C_{\rm band} (2b)^{\alpha_{\rm band}} + O\left(\sqrt{\frac{K_{\rm win}\log|D|}{|D|}}\right)\] for those $b$ with $2b$ in the validity range of the margin bound.
\end{remark}

\begin{definition}[Ambiguous-band mass]\label{def:C_b}
For any $b>0$ define
\[
C_b(\theta^\star)\ :=\ \frac{1}{|D|}\int_D \1\{|s_{\theta^\star}^{\rm or}(\tau)|\le 2b\} \lambda^\star(\tau) \dd\tau.
\]
By \cref{ass:G8}, we have $|D|^{-1}N\!\big(A_{\theta^\star}(2b)\big)=C_b(\theta^\star)+o_{\mathbb P}(1)$ as $|D|\to\infty$.
\end{definition}

\subsubsection{Summary}

Figure \ref{fig:assDAG} summarises the above assumptions, and how they are invoked.


\begin{figure}[!hp]
\centering
\resizebox{\linewidth}{!}{%
\begin{tikzpicture}[
    >=Latex,
    node distance=1.5cm and 1.5cm,
    font=\sffamily,
    block/.style={
        draw=black, 
        fill=white, 
        rectangle, 
        rounded corners, 
        align=center, 
        minimum width=3.5cm, 
        minimum height=1.2cm, 
        inner sep=6pt,
        line width=0.8pt
    },
    wideblock/.style={
        block,
        minimum width=5cm
    },
    groupbox/.style={
        draw=black,
        dashed,
        inner sep=16pt,
        rounded corners=8pt,
        fill=none,
        line width=0.6pt
    },
    line/.style={
        ->,
        draw=black,
        line width=1pt,
        rounded corners=5pt,
        shorten >=2pt,
        shorten <=2pt
    }
]

    
    \def\lvlA{0}      
    \def\lvlB{-3}     
    \def\lvlC{-6}     
    \def\lvlD{-9}     
    \def\lvlE{-13}    
    \def\lvlF{-17}    
    \def\lvlG{-21}    
    \def\lvlH{-24}    
    \def\lvlI{-27}    

    \def\colLL{-14}
    \def\colL{-7}
    \def\colC{0}
    \def\colR{7}
    \def\colRR{14}


    \node[block] (OraclePos) at (\colLL, \lvlA) {Oracle Positivity\\{\footnotesize\cref{ass:oracle-positivity}}};
    \node[block] (LipEnv)    at (\colL, \lvlA)  {Intervention Lipschitz\\{\footnotesize\cref{ass:Lambda-Lip}}};
    \node[block] (Warm)      at (\colRR, \lvlA) {Warm Start\\{\footnotesize\cref{ass:warmstart}}};

    \node[block] (Act)       at (\colLL, \lvlB) {Activity \& Tails\\{\footnotesize\cref{ass:G0}}};
    \node[block] (NonAnt)    at (\colL, \lvlB)  {Predictable\\Updates\\{\footnotesize\cref{ass:G0pred}}};
    \node[block] (SC)        at (\colR, \lvlB)  {Local Concavity\\{\footnotesize\cref{ass:G6}}};
   \node[block] (ExistM)    at (\colRR, \lvlB) {Existence of\\Maximizers\\{\footnotesize\cref{lem:existence-M}}};

    \node[block] (Win)       at (\colLL, \lvlC) {Window Envelope\\{\footnotesize\cref{ass:G2prime}}};
    \node[block] (Margins)   at (\colL, \lvlC)  {Uniform Margins\\{\footnotesize\cref{ass:G1}}};
    \node[block] (Infl)      at (\colC, \lvlB)  {Flip Influence\\{\footnotesize\cref{ass:G3}}};
    \node[block] (LLRLip)    at (\colR, \lvlC)  {LLR Lipschitz\\{\footnotesize\cref{ass:G5}}};
    \node[block] (Lab2Sc)    at (\colRR, \lvlC) {Label-to-Score\\Lipschitz\\{\footnotesize\cref{ass:G7prime}}};

    \node[block] (Align)     at (\colLL, \lvlD) {Score Alignment\\{\footnotesize\cref{ass:G4}}};
    \node[block] (Mass)      at (\colR, \lvlD)  {Mass \&\\Concentration\\{\footnotesize\cref{ass:G8}}};
    \node[block] (Opt)       at (\colRR, \lvlD) {M-step Optimality\\{\footnotesize\cref{ass:G9}}};

    \node[wideblock, line width=1.5pt] (Master) at (\colC, \lvlE) {Master High-Probability Event\\{\footnotesize\cref{lem:master}}};

    \node[block] (PerFlip)   at (\colLL, \lvlF) {Per-flip Bounds\\{\footnotesize\cref{lem:perflip}}};
    \node[block] (EElim)     at (\colL, \lvlF)  {Greedy\\Elimination\\{\footnotesize\cref{lem:E-elim}}};
    \node[block] (SetStab)   at (\colR, \lvlF)  {Set Stability\\{\footnotesize\cref{lem:set-stab}}};
    \node[block] (MStep)     at (\colRR, \lvlF) {M-step Bound\\{\footnotesize\cref{prop:Mstep}}};

    \node[wideblock] (EStep) at (\colL, \lvlG)  {E-step Hamming Bound\\{\footnotesize\cref{prop:Estep}}};
    \node[block]     (InvSC) at (\colRR, \lvlG) {Invariant Set\\{\footnotesize\cref{lem:inv-sc}}};

    \node[wideblock, line width=1.5pt] (Contract) at (\colC, \lvlH) {Theorem: Contraction to a Statistical Floor\\{\footnotesize\cref{thm:contract}}};

    \node[block] (Transfer)  at (\colL, \lvlI)  {Functional\\Transfer\\{\footnotesize\cref{prop:transfer-functional}}};
    \node[block] (Cor)       at (\colR, \lvlI)  {Estimand\\Bounds\\{\footnotesize\cref{cor:compact-estimands}}};

    \begin{scope}[on background layer]
        \node[groupbox, fit=(OraclePos) (Warm) (Act) (Opt) (Align) (Lab2Sc)] (GrpAss) {};
        \node[anchor=north west, font=\bfseries] at (GrpAss.north west) {Setup \& Assumptions};

        \node[groupbox, fit=(Master) (PerFlip) (MStep) (EElim)] (GrpCore) {};
        \node[anchor=north west, font=\bfseries] at (GrpCore.north west) {Core High-Probability Analysis};

        \node[groupbox, fit=(EStep) (InvSC) (Contract) (Cor) (Transfer)] (GrpCons) {};
        \node[anchor=north west, font=\bfseries] at (GrpCons.north west) {Contraction \& Consequences};
    \end{scope}


    \coordinate (MInAct)     at ($(Master.north)+(-2.2,0)$);
    \coordinate (MInNonAnt)  at ($(Master.north)+(-1.65,0)$);
    \coordinate (MInWin)     at ($(Master.north)+(-1.10,0)$);
    \coordinate (MInMargins) at ($(Master.north)+(-0.55,0)$);
    \coordinate (MInInfl)    at ($(Master.north)+( 0.00,0)$);
    \coordinate (MInMass)    at ($(Master.north)+( 0.55,0)$);
    \coordinate (MInLab2Sc)  at ($(Master.north)+( 1.10,0)$);
    \coordinate (MInLLRLip)  at ($(Master.north)+( 1.65,0)$);
    \coordinate (MInSC)      at ($(Master.north)+( 2.20,0)$);

    \coordinate (MOutPerFlip) at ($(Master.south)+(-1.8,0)$);
    \coordinate (MOutEElim)   at ($(Master.south)+(-0.6,0)$);
    \coordinate (MOutEStep)   at ($(Master.south)+( 0.0,0)$);
    \coordinate (MOutSetStab) at ($(Master.south)+( 0.6,0)$);
    \coordinate (MOutMStep)   at ($(Master.south)+( 1.8,0)$);

    \coordinate (ExistMInA) at ($(ExistM.north)+(-0.8,0)$);
    \coordinate (ExistMInB) at ($(ExistM.north)+( 0.8,0)$);

    \coordinate (SetStabInM)    at ($(SetStab.north)+(-0.9,0)$);
    \coordinate (SetStabInMass) at ($(SetStab.north)+( 0.0,0)$);
    \coordinate (SetStabInLip)  at ($(SetStab.north)+( 0.9,0)$);

    \coordinate (MStepInMaster) at ($(MStep.north)+( 0.0,0)$);
    \coordinate (MStepInOpt)    at ($(MStep.north)+(-0.9,0)$);

    \coordinate (EStepInLeft)  at ($(EStep.north)+(-1.6,0)$);
    \coordinate (EStepInMid)   at ($(EStep.north)+( 0.0,0)$);
    \coordinate (EStepInRight) at ($(EStep.north)+( 1.6,0)$);

    \coordinate (ContractInPF)  at ($(Contract.north)+(-1.8,0)$);
    \coordinate (ContractInE)   at ($(Contract.north)+(-0.6,0)$);
    \coordinate (ContractInInv) at ($(Contract.north)+( 0.6,0)$);
    \coordinate (ContractInM)   at ($(Contract.north)+( 1.8,0)$);

    \coordinate (ContractOutL) at ($(Contract.south)+(-1.4,0)$);
    \coordinate (ContractOutR) at ($(Contract.south)+( 1.4,0)$);


    \draw[line] (Act.south)      -- ++(0,-0.35) -| (MInAct);
    \draw[line] (NonAnt.south)   -- ++(0,-0.75) -| (MInNonAnt);

    \draw[line] (Win.south)      -- ++(0,-0.35) -| (MInWin);
    \draw[line] (Margins.south)  -- ++(0,-0.75) -| (MInMargins);
    \draw[line] (Infl.south)     -- ++(0,-1.15) -| (MInInfl);

    \draw[line] (Mass.south)     -- ++(0,-0.55) -| (MInMass);
    \draw[line] (Lab2Sc.south)   -- ++(0,-0.95) -| (MInLab2Sc);
    \draw[line] (LLRLip.south)   -- ++(0,-1.35) -| (MInLLRLip);
    \draw[line] (SC.south)       -- ++(0,-0.35) -| (MInSC);

    \coordinate (MarginsKick) at ($(Margins.north)+(-2.2,0)$);
    \coordinate (MarginsBus)  at ($(MarginsKick |- 0,-1.05)$);
    \draw[line] (Margins.north) -- (MarginsKick) -- (MarginsBus) -| (ExistMInA);

    \coordinate (WinKick) at ($(Win.north)+(-2.2,0)$);
    \coordinate (WinBus)  at ($(WinKick |- 0,-1.35)$);
    \draw[line] (Win.north) -- (WinKick) -- (WinBus) -| (ExistMInB);


    \draw[line] (Align.south) -- ++(0,-0.6) -| (EElim.west);

    \draw[line] (LLRLip.east) -- ++(1.2,0) |- (SetStabInLip);

    \draw[line] (Mass.south) -- (SetStabInMass);

    \draw[line] (ExistM.east) -- ++(1.2,0) |- (MStep.east);
    \draw[line] (Opt.south)   -- ++(0,-0.55) -| (MStepInOpt);

    \draw[line] (Warm.east) -- ++(2.0,0) |- (InvSC.east);

    \draw[line] (MOutPerFlip) -- ++(0,-0.55) -| (PerFlip.north);

    \draw[line] (MOutEElim)   -- ++(0,-0.85) -| (EElim.north);

    \draw[line] (MOutSetStab) -- ++(0,-0.65) -| (SetStabInM);

    \draw[line] (MOutMStep)   -- ++(0,-0.95) -| (MStepInMaster);

    \draw[line] (EElim.south)    -- ++(0,-0.45) -| (EStepInLeft);
    \draw[line] (SetStab.south)  -- ++(0,-0.85) -| (EStepInRight);

    \coordinate (EStepBusX) at (-10, \lvlE-1.6); 
    \draw[line] (MOutEStep) -- ++(0,-0.25) -| (EStepBusX) |- (EStepInMid);

    \draw[line] (EStep.south) -- ++(0,-0.55) -| (ContractInE);

    \draw[line] (PerFlip.west) -- ++(-1.4,0) |- (ContractInPF);

    \draw[line] (InvSC.south) -- ++(0,-0.45) -| (ContractInInv);

    \draw[line] (MStep.east) -- ++(1.4,0) |- (ContractInM);

    \draw[line] (ContractOutL) -- ++(0,-0.65) -| (Transfer.north);
    \draw[line] (ContractOutR) -- ++(0,-0.35) -| (Cor.north);

    \coordinate (LipLane)     at ($(LipEnv.east)+(1.2,0)$);
    \coordinate (LipLaneDown) at ($(LipLane |- Transfer.east)$);
    \draw[line] (LipEnv.east) -- (LipLane) -- (LipLaneDown) -- (Transfer.east);

    \draw[line] (Transfer.east) -- (Cor.west);

\end{tikzpicture}
}
\caption{Dependency structure of the analysis. The diagram flows from top (assumptions) to bottom (theorems). Dashed boxes indicate logical groupings.}
\label{fig:assDAG}
\end{figure}
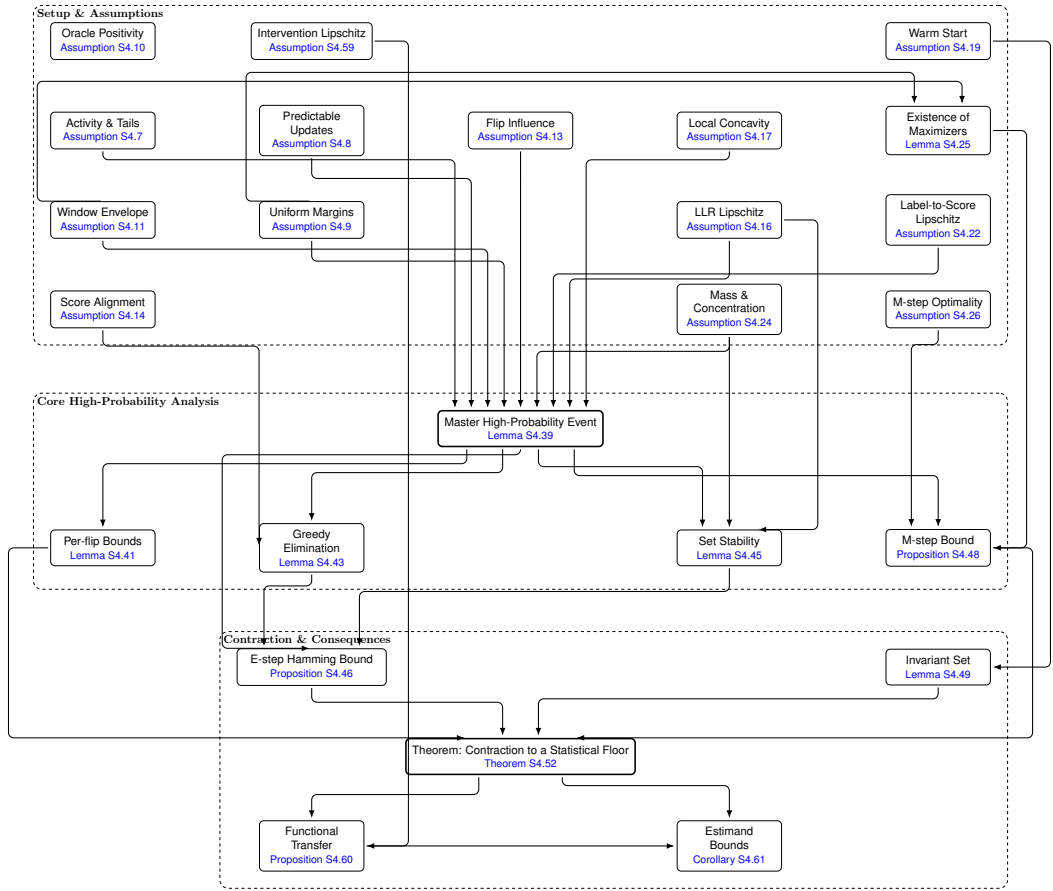
\break

\subsection{Freedman inequality and auxiliary lemmas}\label{sec:freedman}

\begin{lemma}[Indicator of a post-stopping-time interval is predictable]\label{lem:predictable-stoptime}
Let $\sigma$ be a stopping time (with respect to the filtration in force). Then the set $\{(t,\omega):t>\sigma(\omega)\}$ is predictable. Equivalently, $t\mapsto \1\{t>\sigma\}$ is predictable.
\end{lemma}

\begin{proof}
Define $X_t:=\1\{t>\sigma\}$. Since $\sigma$ is a stopping time, $\{\sigma\le q\}\in\mathcal F_q$ for all $q$.
Moreover, for each $t$,
\[
\{X_t=1\}=\{\sigma<t\}=\bigcup_{\substack{q\in\mathbb Q\\ q<t}}\{\sigma\le q\}\in\mathcal F_t,
\]
so $(X_t)$ is adapted. Moreover $X$ is left-continuous: $X_\sigma=0$ and for $t\uparrow\sigma$ we have $X_t=0$. Every adapted left-continuous process is predictable, hence $X$ is predictable.
\end{proof}

\begin{lemma}[Freedman inequality for point-process martingales]\label{lem:Freedman}
Let $M^\star$ be the compensated martingale of $N$ and let $h$ be $\mathcal G$-predictable with $\|h\|_\infty\le H<\infty$. Set
\[
X_T=\int_D h\,\dd M^\star,\qquad
V=\int_D h^2 \lambda^\star\,\dd\tau.
\]
Then, for any $z_{\rm Fr}>0$,
\begin{align*}
\bbP\!\left(X_T \ge \frac{H}{3} z_{\rm Fr} + \sqrt{ 2Vz_{\rm Fr} + \frac{H^2}{9}z_{\rm Fr}^2 }\right)\le& e^{-z_{\rm Fr}},\\
\bbP\!\left(|X_T| \ge \frac{H}{3} z_{\rm Fr} + \sqrt{ 2Vz_{\rm Fr} + \frac{H^2}{9}z_{\rm Fr}^2 }\right)\le& 2e^{-z_{\rm Fr}}.
\end{align*}
The same bounds hold when $M^\star$ is replaced by either component martingale $M_k^\star$ and $\lambda^\star$ by $\lambda_k^\star$.
\end{lemma}

\begin{proof}
This is a standard continuous time analogue of Freedman inequality for martingales with bounded jumps and predictable quadratic characteristic \citep{freedman1975tail}, the properties of which are discussed in depth in \citet{DZHAPARIDZE2001109}.
\end{proof}

\begin{lemma}[Near-zero oracle score at the truth via Freedman]\label{lem:score-nearzero}
Assume oracle positivity (Assumption~\ref{ass:oracle-positivity}) and the oracle gradient bound in
Assumption~\ref{ass:G3}. Fix $z_{\rm Fr}>0$ and work on the window envelope event
$\Omega_{n_*}(L,R)$ of Assumption~\ref{ass:G2prime}, so that $\sup_{\tau\in D}\lambda^\star(\tau)\le K_{\rm win}$.
Then with probability at least $1-4p\,e^{-z_{\rm Fr}}$ (conditional on $\Omega_{n_*}(L,R)$),
\[
\|\nabla f_{r^\star}(\theta^\star)\|
\ \le\ C_{\rm sc}\!\left(\sqrt{\frac{K_{\rm win} z_{\rm Fr}}{|D|}} + \frac{z_{\rm Fr}}{|D|}\right),
\]
where one may take
\[
C_{\rm sc}\ :=\ \frac{8\sqrt{2p}\,S_\infty}{\underline\lambda_{\min}},
\qquad
\underline\lambda_{\min}:=\min\{\underline\lambda_0,\underline\lambda_1\}.
\]
\end{lemma}

\begin{proof}
Write $\lambda_k^\star(\tau):=\lambda_k(\tau\mid\mathcal G_{t-};\theta^\star)$ and recall that under $r^\star$
(the true labels) the complete-data log-likelihood is
\[
\ell_{r^\star}(\theta)=\sum_{k=0}^1\left(\int_D \log\lambda_k(\tau\mid\mathcal G_{t-};\theta)\,dN_k(\tau)
-\int_D \lambda_k(\tau\mid\mathcal G_{t-};\theta)\,d\tau\right).
\]
Differentiating and evaluating at $\theta^\star$ gives
\[
\nabla \ell_{r^\star}(\theta^\star)
=\sum_{k=0}^1\left(\int_D \nabla_\theta\log\lambda_k(\tau\mid\mathcal G_{t-};\theta^\star)\,dN_k(\tau)
-\int_D \nabla_\theta\lambda_k(\tau\mid\mathcal G_{t-};\theta^\star)\,d\tau\right).
\]
Using the Doob--Meyer decomposition $dN_k=dM_k^\star+\lambda_k^\star\,d\tau$ and the identity
$\nabla_\theta\log\lambda_k=(\nabla_\theta\lambda_k)/\lambda_k$, we obtain the cancellation
\[
\int_D \nabla_\theta\log\lambda_k^\star\,\lambda_k^\star\,d\tau
=\int_D \nabla_\theta\lambda_k^\star\,d\tau,
\]
hence
\begin{equation}\label{eq:score-mtg-rep}
\nabla \ell_{r^\star}(\theta^\star)
=\sum_{k=0}^1\int_D \nabla_\theta\log\lambda_k^\star(\tau)\,dM_k^\star(\tau).
\end{equation}
Dividing by $|D|$ yields the same representation for $\nabla f_{r^\star}(\theta^\star)$.

Fix a coordinate $j\in\{1,\dots,p\}$ and define the predictable integrand
\[
h_{k,j}(\tau):=\partial_{\theta_j}\log\lambda_k^\star(\tau)
=\frac{\partial_{\theta_j}\lambda_k^\star(\tau)}{\lambda_k^\star(\tau)}.
\]
By oracle positivity, $\lambda_k^\star(\tau)\ge \underline\lambda_k\ge \underline\lambda_{\min}$, and by the oracle
gradient bound in Assumption~\ref{ass:G3},
$\|\nabla_\theta\lambda_k^\star(\tau)\|\le S_\infty$ a.e.; therefore
\[
|h_{k,j}(\tau)|
\le \frac{\|\nabla_\theta\lambda_k^\star(\tau)\|}{\lambda_k^\star(\tau)}
\le \frac{S_\infty}{\underline\lambda_{\min}}
=:H.
\]
On the window envelope event $\Omega_{n_*}(L,R)$ we have $\lambda_k^\star(\tau)\le \lambda^\star(\tau)\le K_{\rm win}$,
so the predictable quadratic characteristic of $\int_D h_{k,j}\,dM_k^\star$ satisfies
\[
V_{k,j}
:=\int_D h_{k,j}(\tau)^2\,\lambda_k^\star(\tau)\,d\tau
\le \int_D H^2\,\lambda_k^\star(\tau)\,d\tau
\le H^2\,K_{\rm win}\,|D|.
\]
Apply Lemma~\ref{lem:Freedman} (two-sided form) to the scalar martingale integral
$X_{k,j}:=\int_D h_{k,j}\,dM_k^\star$. On $\Omega_{n_*}(L,R)$, with probability at least $1-2e^{-z_{\rm Fr}}$,
\[
|X_{k,j}|
\le \frac{H}{3}z_{\rm Fr}+\sqrt{2V_{k,j}z_{\rm Fr}+\frac{H^2}{9}z_{\rm Fr}^2}
\le \frac{H}{3}z_{\rm Fr}+\sqrt{2H^2K_{\rm win}|D|\,z_{\rm Fr}+\frac{H^2}{9}z_{\rm Fr}^2}.
\]
Using $\sqrt{a+b}\le \sqrt a+\sqrt b$ and $\sqrt{H^2}=H$ gives
\[
|X_{k,j}|
\le H\sqrt{2K_{\rm win}|D|\,z_{\rm Fr}}+\frac{2H}{3}z_{\rm Fr}.
\]
Now union bound over the $2p$ pairs $(k,j)$ yields that, on $\Omega_{n_*}(L,R)$, with probability at least
$1-4p\,e^{-z_{\rm Fr}}$ the above bound holds simultaneously for all $k\in\{0,1\}$ and $j\in\{1,\dots,p\}$.
On this event, for each coordinate $j$ we combine \eqref{eq:score-mtg-rep} with the triangle inequality:
\[
\big|\partial_{\theta_j}\ell_{r^\star}(\theta^\star)\big|
=\left|\sum_{k=0}^1 X_{k,j}\right|
\le \sum_{k=0}^1|X_{k,j}|
\le 2H\sqrt{2K_{\rm win}|D|\,z_{\rm Fr}}+\frac{4H}{3}z_{\rm Fr}.
\]
Dividing by $|D|$ yields
\[
\big|\partial_{\theta_j} f_{r^\star}(\theta^\star)\big|
\le 2H\sqrt{\frac{2K_{\rm win}z_{\rm Fr}}{|D|}}+\frac{4H}{3}\frac{z_{\rm Fr}}{|D|}.
\]
Finally, $\|\nabla f\|\le \sqrt p\,\max_j|\partial_{\theta_j} f|$, hence
\[
\|\nabla f_{r^\star}(\theta^\star)\|
\le \sqrt p\left(2H\sqrt{\frac{2K_{\rm win}z_{\rm Fr}}{|D|}}+\frac{4H}{3}\frac{z_{\rm Fr}}{|D|}\right)
\le \frac{8\sqrt{2p}\,S_\infty}{\underline\lambda_{\min}}
\left(\sqrt{\frac{K_{\rm win}z_{\rm Fr}}{|D|}}+\frac{z_{\rm Fr}}{|D|}\right),
\]
which is the claimed bound with the stated choice of $C_{\rm sc}$.
\end{proof}

\begin{lemma}[Stochastic Fubini for $v$-interpolation]\label{lem:stoch-fubini}
Let $H(\tau,v)$ be jointly measurable in $(\tau,v)\in D\times[0,1]$, $\mathcal G$-predictable in $\tau$ for each $v$, and $\sup_{v\in[0,1]}\|H(\cdot,v)\|_\infty<\infty$. Suppose moreover that
\[
\int_0^1\!\!\int_D H(\tau,v)^2 \lambda^\star(\tau) \dd\tau \dd v\ <\ \infty.
\]
Then
\[
\int_0^1\!\!\int_D H(\tau,v) \dd M^\star(\tau) \dd v\ =\ \int_D\!\Big(\int_0^1 H(\tau,v) \dd v\Big) \dd M^\star(\tau)
\]
almost surely, and the right-hand side is square-integrable \citep{protter2012stochastic}. 
\end{lemma}

\begin{lemma}[Deterministic Lipschitz bounds in $\theta$ for common integrands]\label{lem:theta-Lip-integrands}
Assume Assumptions~\ref{ass:G1} and~\ref{ass:G3}.
Then for every deterministic labelling \(r\) and for a.e.\ \(\tau\in D\) the following hold uniformly over \(\theta,\theta'\in\Theta_\circ\).

\begin{enumerate}[nosep,label=(\alph*)]
\item \textbf{(Log-intensity Lipschitz.)} For each $k\in\{0,1\}$,
\[
\big|\log \widehat\lambda_k^r(\tau;\theta)-\log \widehat\lambda_k^r(\tau;\theta')\big|
\le \frac{S_\infty}{\underline\mu_{\min}}\,\|\theta-\theta'\|.
\]

\item \textbf{(Score-gradient Lipschitz.)} For each $k\in\{0,1\}$ and coordinate $j\in\{1,\dots,p\}$,
\[
\big|\partial_{\theta_j}\log \widehat\lambda_k^r(\tau;\theta)-\partial_{\theta_j}\log \widehat\lambda_k^r(\tau;\theta')\big|
\le \left(\frac{H_\infty}{\underline\mu_{\min}}+\frac{S_\infty^2}{\underline\mu_{\min}^2}\right)\|\theta-\theta'\|.
\]

\item \textbf{(Interpolation-ratio Lipschitz.)} If $r'$ differs from $r$ by a single flip and
$\Delta\widehat\lambda^r:=\widehat\lambda^{r'}-\widehat\lambda^{r}$, then for every $v\in[0,1]$,
\[
\left|\frac{\Delta\widehat\lambda^r(\tau;\theta)}{\widehat\lambda^r(\tau;\theta)+v\Delta\widehat\lambda^r(\tau;\theta)}
-\frac{\Delta\widehat\lambda^r(\tau;\theta')}{\widehat\lambda^r(\tau;\theta')+v\Delta\widehat\lambda^r(\tau;\theta')}\right|
\le L_{\rm frac}\,\|\theta-\theta'\|,
\]
where one may take
\[
L_{\rm frac}:=\frac{\overline B_\infty^\Sigma}{\underline\mu_\Sigma}
+\frac{B_\infty\big(2S_\infty+\overline B_\infty^\Sigma\big)}{\underline\mu_\Sigma^2}.
\]
\end{enumerate}
\end{lemma}

\begin{remark}[Uniform Freedman bounds along a blockwise update path]\label{rem:conditional-freedman}
We work under $P_{\theta^\star}$ and apply union bounds over a finite collection of indices
(updates $m\le\Tmax$, event indices/candidates within the relevant window, components, coordinates, and net points).
By Assumption~\ref{ass:G0pred}, along Algorithm~\ref{alg:blockwise-hardem} the labels used inside any stochastic integral
on a window $D^{(m)}$ are non-anticipating, so every label-induced intensity appearing as an integrand is $\mathcal G$-predictable
on that window.

Moreover, since $\theta^{(m)}$ is $\mathcal G_{u_m}$-measurable, it is known throughout the subsequent block $(u_m,u_{m+1}]$,
so treating $\theta^{(m)}$ as fixed when applying Freedman inequality within that block is legitimate.
We obtain Freedman inequality bounds first at net points $\theta\in\mathcal N_{\rm net}$ and then extend deterministically to all
$\theta\in\Theta_\circ$ using Lipschitz bounds in $\theta$.

Finally, because the deviation bounds are established uniformly over $\theta\in\Theta_\circ$ (via a finite net and a deterministic Lipschitz extension), they may be evaluated at the realized, data-dependent iterate $\theta^{(m)}$ without any additional conditioning argument. Assumption~\ref{ass:G0pred}(a) is used only to ensure that $\theta^{(m)}$ is already known at time $u_m$, so the corresponding integrands are predictable on the subsequent block $(u_m,u_{m+1}]$.
\end{remark}


\begin{lemma}[Predictability of forward cones]\label{lem:predictable-cone}
Let $(\mathcal G_t)_{t\ge 0}$ be a filtration satisfying the usual conditions, and let
$\mathcal P(\mathcal G)$ denote its predictable $\sigma$-algebra on $\Omega\times\mathbb R_+$.
Let $T$ be a $\mathcal G$-stopping time and let $X$ be $\mathcal G_T$-measurable with values in
$\mathbb R^d$. Fix $L,R>0$ and a deterministic $u\ge 0$, and define the (random) forward cone
\begin{align*}
C^{[u]}(T,X)
:=& \Bigl(\{(t,x): T<t\le u,\ t\le T+L\}\Bigr)\cap\bigl(\mathbb R_+\times B_R(X)\bigr),\\
B_R(y):=&\{x\in\mathbb R^d:\|x-y\|\le R\}.
\end{align*}
Then the indicator $(\omega,t,x)\mapsto \mathbf 1_{C^{[u]}(T,X)}(t,x)$ is
$\mathcal P(\mathcal G)\otimes \mathcal B(\mathbb R^d)$-measurable (i.e.\ $\mathcal G$-predictable
as a space--time integrand).
\end{lemma}

\begin{proof}
Write $\mathcal P:=\mathcal P(\mathcal G)$ and $\mathcal B:=\mathcal B(\mathbb R^d)$.

Since $T$ is a stopping time, $\{(\omega,t):T(\omega)<t\}\in\mathcal P$ (e.g.\ $t\mapsto \mathbf 1\{t>T\}$ is adapted and
left-continuous).
Likewise, $T+L$ is a stopping time, hence $\{(\omega,t):T(\omega)+L<t\}=\{t>T+L\}\in\mathcal P$, and therefore
$\{(\omega,t):t\le T+L\}\in\mathcal P$ by closure under complements.
Also $(\Omega\times(0,u])\in\mathcal P$ since $u$ is deterministic.
Thus
\[
A:=\{(\omega,t): T<t\le u,\ t\le T+L\}\in\mathcal P.
\]

The map $(\omega,x)\mapsto \mathbf 1\{\|x-X(\omega)\|\le R\}$ is $\sigma(X)\otimes\mathcal B$-measurable.
Since $X$ is $\mathcal G_T$-measurable, for each rational $q\ge 0$ and each Borel $B\in\mathcal B$,
\[
\{X\in B\}\cap\{T\le q\}\in\mathcal G_q,
\]
so by a monotone-class argument the restriction of $(\omega,x)\mapsto \mathbf 1\{\|x-X(\omega)\|\le R\}$ to $\{T\le q\}$
is $\mathcal G_q\otimes\mathcal B$-measurable.

Now consider
\[
E:=\{(\omega,t,x): T(\omega)<t,\ \|x-X(\omega)\|\le R\}
=
\bigcup_{q\in\mathbb Q_+}\Bigl(\{T\le q\}\cap\{\|x-X\|\le R\}\Bigr)\times(q,\infty).
\]
Each set $\bigl(\{T\le q\}\cap\{\|x-X\|\le R\}\bigr)\times(q,\infty)$ belongs to $\mathcal P\otimes\mathcal B$
because $\{T\le q\}\cap\{\|x-X\|\le R\}\in\mathcal G_q\otimes\mathcal B$ and $\mathcal P$ is generated by rectangles
$F\times(s,t]$ with $F\in\mathcal G_s$.
Hence $E\in\mathcal P\otimes\mathcal B$.

Finally,
\[
\mathbf 1_{C^{[u]}(T,X)}(\omega,t,x)=\mathbf 1_A(\omega,t)\,\mathbf 1_E(\omega,t,x),
\]
and since $A\in\mathcal P$ and $E\in\mathcal P\otimes\mathcal B$, we obtain
$\mathbf 1_{C^{[u]}(T,X)}\in\mathcal P\otimes\mathcal B$.
\end{proof}

\subsection{Penalized objective, per-flip bound, and E-step elimination}\label{sec:perflip}

We first introduce our oracle aided penalized hard-EM algorithm. Fix an update index $m$ and consider the cumulative (prefix) window $D^{(m)}=D^{[u_m]}=(t^\ast,u_m]\times\mathcal S$. In Lemma~\ref{lem:perflip} and in the Freedman inequality bounds that follow, all stochastic integrals and event-sums are
understood to be taken over the active window $D^{(m)}$ (equivalently, the likelihood is $\ell_r^{[u_m]}$ and the
normalization is by $|D^{(m)}|$).
This is exactly the regime where Assumption~\ref{ass:G0pred} guarantees $\mathcal G$-predictability of the integrands. Any dependence on the data-dependent iterate $\theta^{(m)}$ is handled by establishing Freedman inequality bounds uniformly over $\theta\in\Theta_\circ$ (via a finite net and deterministic Lipschitz extension), as detailed in \cref{rem:conditional-freedman}. The special case $m=\Tmax$ corresponds to the full window $u_m=T$. Given $b>0$ and $\alpha>0$, define the oracle-aided penalized objective on a prefix window $D^{[u]}$ as $\mathcal J_\alpha^{[u]}(r,\theta)$ in \eqref{eq:penobj}. In this section, at update $m$ we work with $\mathcal J_\alpha^{[u_m]}$ on the active window $D^{(m)}$.

\begin{equation}\label{eq:penobj}
\mathcal J_\alpha^{[u]}(r,\theta)
:=\ell_r^{[u]}(\theta)
-\alpha K_{\rm win}\cdot
\#\Big\{\gamma_i:\ t_i\le u,\ \gamma_i\in S_\theta^+(b)\cup S_\theta^-(b),\ r_i=\text{minority}(\gamma_i;\theta,b)\Big\}.
\end{equation}
We also write $\mathcal J_\alpha(r,\theta):=\mathcal J_\alpha^{[T]}(r,\theta)$ for the full-window objective.

We now describe the predictable (greedy) E-step under the blockwise schedule
$t^\ast=u_0<u_1<\cdots<u_{\Tmax}=T$.
At update $m\in\{0,\dots,\Tmax-1\}$ we freeze the parameter $\theta^{(m)}$
(which is $\mathcal G_{u_m}$-measurable by Assumption~\ref{ass:G0pred}(a))
and assign labels only for events in the next block $D_{m+1}=(u_m,u_{m+1}]\times\mathcal S$,
sequentially in event time, while keeping all past labels fixed.

Specifically, for each event $\gamma_i=(t_i,x_i)$ with $u_m<t_i\le u_{m+1}$,
having already fixed $r^{(m+1)}_{1:i-1}$ (which includes all labels from previous blocks and earlier events in the current block),
set
\begin{align}
\label{eq:greedy-estep}
r_i^{(m+1)} \in \argmax_{k \in \{0,1\}}
\Big\{
&\log \widehat\lambda_{k}^{r^{(m+1)}}(\gamma_i;\theta^{(m)}) \nonumber\\
&\quad
-\alpha K_{\rm win}\,
\1\!\left\{
\gamma_i \in S^+_{\theta^{(m)}}(b)\cup S^-_{\theta^{(m)}}(b),\ 
k=\mathrm{minority}(\gamma_i;\theta^{(m)},b)
\right\}
\Big\}.
\end{align}
with a deterministic tie-breaking rule (and at decisive events ties are broken in favor of the non-minority label).
For all indices $j$ with $t_j\le u_m$ we set $r_j^{(m+1)}:=r_j^{(m)}$ (labels are not revised).

\begin{figure}[t]
\centering
\begin{minipage}{0.97\linewidth}
\hrule\vspace{0.6em}
\textbf{Algorithm 1 (Analyzed operator): Predictable blockwise penalized hard-EM.}\\[0.3em]
\textbf{Input:} deterministic partition $t^\ast=u_0<\cdots<u_{\Tmax}=T$; warm start $\theta^{(0)}\in\Theta_{\rm sc}$;
threshold $b\ge b_{\rm th}$; penalty $\alpha>0$.\\
\textbf{Initialize:} $r^{(0)}$ (arbitrary; no labels are fixed before the first block).\\[0.2em]
\textbf{For $m=0,1,\dots,\Tmax-1$:}\\
\quad\emph{(E-step on the next block $D_{m+1}$)} Assign labels for events with $u_m<t_i\le u_{m+1}$ sequentially by
\eqref{eq:greedy-estep}, keeping labels with $t_j\le u_m$ fixed.\\
\quad\emph{(M-step at the block end)} Update the parameter using the cumulative window $D^{(m+1)}=(t^\ast,u_{m+1}]\times\mathcal S$:
\[
\theta^{(m+1)} := \hat\theta_{\Theta_{\rm sc}}^{[u_{m+1}]}\!\big(r^{(m+1)}\big)
\in \arg\max_{\theta\in\Theta_{\rm sc}}\ell_{r^{(m+1)}}^{[u_{m+1}]}(\theta).
\]
\textbf{Output:} $\theta^{(\Tmax)}$ and the (progressively assigned) labels $r^{(\Tmax)}$ on $D$.
\vspace{0.6em}\hrule
\end{minipage}
\caption{Predictable blockwise penalized hard-EM operator used in the analysis. The practical SEM implementation may include additional offline smoothing passes, but the theoretical argument only requires the time-respecting operator displayed here.}
\label{alg:blockwise-hardem}
\end{figure}

\begin{lemma}[Non-anticipation of the blockwise greedy E-step]\label{lem:greedy-G0pred}
Fix $m\in\{0,\dots,\Tmax-1\}$ and suppose:
(i) $\theta^{(m)}$ is $\mathcal G_{u_m}$-measurable, and
(ii) the current labels $r^{(m)}$ satisfy Assumption~\ref{ass:G0pred}(b) up to time $u_m$. Construct $r^{(m+1)}$ by keeping all labels with $t_j\le u_m$ fixed and assigning labels for events $\gamma_i$ with $u_m<t_i\le u_{m+1}$ sequentially by the greedy rule \eqref{eq:greedy-estep}, using a deterministic tie-breaking rule.  Then $r^{(m+1)}$ satisfies Assumption~\ref{ass:G0pred}(b) up to time $u_{m+1}$.
\end{lemma}

\begin{proof}
For indices $i$ with $t_i\le u_m$ we have $r_i^{(m+1)}=r_i^{(m)}$, hence $\mathcal G_{t_i}$-measurability holds by the
induction (progressive) hypothesis in Assumption~\ref{ass:G0pred}(b).

Now consider the events in the new block $D_{m+1}$ and list them in increasing time order.
We argue by induction over this within-block ordering.
For the first such event $\gamma_i$ (smallest $t_i$ with $u_m<t_i\le u_{m+1}$), the objective in \eqref{eq:greedy-estep}
depends only on:
(i) the observed history $\mathcal H_{t_i-}\subset \mathcal G_{t_i}$,
(ii) already-fixed past labels $r^{(m+1)}_{1:i-1}$ (which are $\mathcal G_{t_i}$-measurable by construction),
and (iii) $\theta^{(m)}$, which is $\mathcal G_{u_m}$-measurable and hence $\mathcal G_{t_i}$-measurable because $u_m<t_i$
and $(\mathcal G_t)$ is increasing.
With deterministic tie-breaking, $r_i^{(m+1)}$ is therefore $\mathcal G_{t_i}$-measurable.

Assume the claim holds for all earlier events in $D_{m+1}$. For the next event $\gamma_{i'}$ in the block,
the greedy objective depends on $\mathcal H_{t_{i'}-}$, the already-chosen labels in the block,
and $\theta^{(m)}$, all of which are $\mathcal G_{t_{i'}}$-measurable. Deterministic tie-breaking again yields that
$r_{i'}^{(m+1)}$ is $\mathcal G_{t_{i'}}$-measurable. This closes the induction.
\end{proof}

\begin{remark}[Why the blockwise timing condition is needed]\label{rem:theta-conditioning}
If one computes $\theta^{(m)}$ from the entire offline window $D=(t^\ast,T]\times\mathcal S$ and then uses it to
assign labels to early events $\gamma_i$ with $t_i$ close to $t^\ast$, then $r_i^{(m+1)}$ generally depends on future events
through $\theta^{(m)}$ and need not be $\mathcal G_{t_i}$-measurable. This breaks Assumption~\ref{ass:G0pred}
(and hence $\mathcal G$-predictability of the label-induced intensities inside martingale integrals).
The deterministic blockwise schedule in Algorithm~\ref{alg:blockwise-hardem} enforces the required timing:
$\theta^{(m)}$ is computed using data only up to $u_m$ and is therefore known when labelling any event in $(u_m,u_{m+1}]$.
\end{remark}

Lemma~\ref{lem:greedy-G0pred} shows that, under Assumption~\ref{ass:G0pred}(a), the blockwise sequential construction produces labels that satisfy the progressive $\mathcal G_{t_i}$-measurability in Assumption~\ref{ass:G0pred}(b). In case of a tie at a decisive event $\gamma_i\in S^+_{\theta^{(m)}}(b)\cup S^-_{\theta^{(m)}}(b)$ we break ties by choosing the non-minority label. Given $r^{(m+1)}$, the M-step (localized for the analysis under Assumption~\ref{ass:G6}) at the block end $u_{m+1}$ is
\[
\theta^{(m+1)}\ :=\ \hat\theta_{r^{(m+1)}}^{[u_{m+1}]}\ \in\ \arg\max_{\theta\in\Theta_{\rm sc}}\ell_{r^{(m+1)}}^{[u_{m+1}]}(\theta),
\qquad
\Theta_{\rm sc}:=\Theta_\circ\cap B(\theta^\star,r_{\rm sc}).
\]

We do not assume the M-step maximiser is unique. Whenever $\arg\max_{\theta\in K}\ell_r^{[u]}(\theta)$ is set-valued, the algorithm may choose any maximiser (e.g.\ returned by a numerical routine). All bounds below hold uniformly for
every such choice. For notational convenience we fix a selection $\hat\theta_K^{[u]}(r)\in\operatorname{ArgMax}_K^{[u]}(r)$ and write $\theta^{(m+1)}:=\hat\theta_{\Theta_{\rm sc}}^{[u_{m+1}]}(r^{(m+1)})$.

This ball localization is imposed only for the analysis, to ensure that the selected M-step maximiser lies in the region where Assumption~\ref{ass:G6} controls the curvature of $f_{r^\star}$. Once the two-line recursion holds with $\rho<1$ and the statistical floor is sufficiently small, Lemma~\ref{lem:inv-sc} implies that the iterates remain in a strict sub-neighbourhood of $\theta^\star$ (hence in $\operatorname{int}(\Theta_{\rm sc})$), so the additional ball constraint does not bind along the contraction path.

We notate the update count and subsequent union bounds as $\Tmax=\lceil C_{\rm it}\log|D|\rceil$ as in \cref{ass:G8}. Under Algorithm~\ref{alg:blockwise-hardem}, labels are assigned progressively and never revised.
Hence each event $\gamma_i$ contributes at most one label assignment over the entire update path.
On the activity event $E_{\rm act}$ of \cref{ass:G0} we therefore have the total update count bound
\begin{equation}\label{eq:flipcount}
M_{\rm flips}\ \le\ N(D)\ \le\ (1+\eta_{\rm act}) C_{\rm act} |D|.
\end{equation}
(Algorithm~\ref{alg:blockwise-hardem} assigns each event label once, progressively, and never revisits past events. Any additional smoothing passes used in practice are outside the analysed operator.)

We next introduce the master high-probability event.
Fix $z_{\rm Fr}>0$ and define
\[
E_{\mathrm{master}}(z_{\rm Fr})
:=E_{\rm act}\cap \Omega_{n_*}(L,R)\cap E^{\rm align}_{|D|}
\cap E^{\rm Lip}_{|D|}\cap E^{\rm sc}_{|D|}\cap E^{\nabla}_{|D|}\cap E^{\rm mass}_{|D|}
\cap E^{\rm Fr}_{|D|}(z_{\rm Fr}),
\]
where $E^{\rm Fr}_{|D|}(z_{\rm Fr})$ denotes the intersection of all Freedman deviation events (with level $z_{\rm Fr}$)
invoked in Lemma~\ref{lem:master}(ii) (per-flip martingales and M-step anchor gradient-coordinate martingales,
uniformly over $m\le\Tmax$, the relevant event indices, components, and net points).
On $E_{\mathrm{master}}(z_{\rm Fr})$, all inequalities used in the subsequent E-step/M-step arguments hold simultaneously.

Moreover, $\Pr(E_{\mathrm{master}}(z_{\rm Fr})^c)$ is controlled by a finite union bound over the constituent events
(and over the finite family of Freedman deviations at level $z_{\rm Fr}$), yielding an overall failure probability
of order $|D|^{-\eta''}$ for some $\eta''>0$. For notational convenience, set
\[
E_{\rm main}
:=E_{\rm act}\cap \Omega_{n_*}(L,R)\cap E^{\rm align}_{|D|}
\cap E^{\rm Lip}_{|D|}\cap E^{\rm sc}_{|D|}\cap E^{\nabla}_{|D|}\cap E^{\rm mass}_{|D|},
\qquad n_*(|D|)=\lceil c_1\log|D|\rceil.
\]

\begin{lemma}[Master high-probability event]\label{lem:master}
Assume \cref{ass:G0,ass:G0pred,ass:G1,ass:G2prime,ass:G3,ass:G4,ass:G5,ass:G6,ass:G7prime,ass:G8}. There exists $\eta''>0$ such that, for all large $|D|$, with probability at least $1-|D|^{-\eta''}$, the following hold simultaneously:
\begin{enumerate}[nosep,label=(\roman*)]
\item $E_{\rm main}$ occurs;
\item For every update index $m\le \Tmax$, every event index $j \le N^{(m)}$ (i.e.\ with $t_j\le u_m$),
every selected component $k\in\{0,1\}$, and every net point $\theta\in\mathcal N_{\rm net}$,
the Freedman inequality bounds with the same parameter $z_{\rm Fr}$ hold simultaneously for the scalar martingales used in Lemma~\ref{lem:perflip} (on the active window $D^{(m)}$). Moreover, for every update index $m\in\{0,\dots,\Tmax-1\}$ (with $u=u_{m+1}$ and $r=r^{(m+1)}$), and for every gradient coordinate used in the Assumption~\ref{ass:G7prime} verification at the M-step anchor $\hat\theta_{r}^{[u]}=\hat\theta_{r^{(m+1)}}^{[u_{m+1}]}$, the corresponding Freedman inequality bounds (with the same $z_{\rm Fr}$) hold simultaneously.
\end{enumerate}
\end{lemma}

\begin{proof}
We establish the result in two steps: first, a union bound over a finite net to handle the martingale noise; second, a deterministic Lipschitz extension to cover the entire parameter space $\Theta_\circ$.

Recall the fixed, deterministic $\varepsilon_{\rm net}$-net $\mathcal{N}_{\rm net} \subset \Theta_\circ$ from \cref{sec:assumptions}. For concreteness one may take $\varepsilon_{\rm net}:=|D|^{-2}$ (which satisfies $\varepsilon_{\rm net}\le |D|^{-1}$ for $|D|\ge 1$, which is eventually true). Since $p$ is fixed, $|\mathcal{N}_{\rm net}| \le (C \operatorname{diam}(\Theta_\circ)/\varepsilon_{\rm net})^p$, so $\log N_{\rm net} = O(\log |D|)$.
Item (i) is exactly the event $E_{\rm main}$, whose failure probability is controlled by a finite union bound using the
tail bounds in \cref{ass:G0,ass:G2prime,ass:G4,ass:G5,ass:G6,ass:G7prime,ass:G8}. For (ii), we consider the collection of all scalar martingales evaluated by the algorithm at update index \(m\), event index $j$, component $k$, and parameter $\theta$. By Assumption~\ref{ass:G0pred}, for any fixed $\theta \in \mathcal{N}_{\rm net}$, the integrand defining these martingales depends only on the history and is $\mathcal{G}$-predictable on the relevant cumulative window.

Let $N_{\rm total}$  denote the total number of (two-sided) Freedman inequality tail events that we union bound over. On $E_{\rm act}$, the number of prefix event indices we union over is at most $\sum_{m=1}^{\Tmax}N^{(m)}\le \Tmax\,N(D)\le (1+\eta_{\rm act})C_{\rm act}|D|\,\Tmax$. Each candidate contributes at most $(p+2)$ martingales (two scalar terms in Lemma~\ref{lem:perflip} and $p$ gradient coordinates), and we take a two-sided deviation bound. Therefore, for an absolute constant $c_0$ (absorbing these fixed factors, unions over update indices, and any ancillary Freedman inequality calls included in the master union bound),
\[
N_{\rm total}\ \le\ c_0\,(p+2)\,N_{\rm net}\,(1+\eta_{\rm act})C_{\rm act}|D|\,\Tmax.
\]
The parameter $z_{\rm Fr}$ in \eqref{eq:z-def} is chosen so that $e^{-z_{\rm Fr}}$ dominates $N_{\rm total}^{-1}$ (up to the extra $|D|^{-\eta_x}$ factor), making the union bound over all $N_{\rm total}$ Freedman events polynomially small in $|D|$. Applying \cref{lem:Freedman} and summing the failure probabilities yields that, with probability at least $1 - |D|^{-\eta''}$, the Freedman inequality bounds hold simultaneously for all $\theta \in \mathcal{N}_{\rm net}$ and all algorithmic steps.

We now extend to $\Theta_\circ$ via Lipschitz continuity. Working on the intersection of  the activity event $E_{\rm act}$, and the window envelope event $\Omega_{n_*}(L,R)$, consider an arbitrary $\theta \in \Theta_\circ$. Let $\theta_g \in \mathcal{N}_{\rm net}$ be the closest net point, so $\|\theta - \theta_g\| \le \varepsilon_{\rm net}$. Let $X(\theta) = \int_D H(\tau; \theta) \dd M^\star(\tau)$ be any of the generic martingale terms (e.g., from \cref{lem:perflip} or the gradient in Assumption \ref{ass:G7prime}). We decompose the value at $\theta$ as:
\[
|X(\theta)| \le |X(\theta_g)| + |X(\theta) - X(\theta_g)|.
\]
The first term is bounded by the grid result in the above union bound. For the second term, we use the decomposition $\dd M^\star = \dd N - \lambda^\star \dd \tau$:
\[
|X(\theta) - X(\theta_g)| 
= \left| \int_D (H(\tau;\theta) - H(\tau;\theta_g)) \dd M^\star \right|
\le \int_D \Delta_H(\tau) \dd N(\tau) + \int_D \Delta_H(\tau) \lambda^\star(\tau) \dd\tau,
\]
where $\Delta_H(\tau) = |H(\tau;\theta) - H(\tau;\theta_g)|$.
By Lemma~\ref{lem:theta-Lip-integrands} (and the margins in Assumption~\ref{ass:G1}), the relevant integrands $H(\tau;\cdot)$ are Lipschitz in $\theta$ on $\Theta_\circ$ with a deterministic constant $L_H<\infty$ (dependent on model constants and $K_{\rm win}$, but independent of $|D|$). Thus $\Delta_H(\tau) \le L_H \varepsilon_{\rm net}$.
Substituting this bound:
\[
|X(\theta) - X(\theta_g)| \le L_H \varepsilon_{\rm net} \big( N(D) + \Lambda^\star(D) \big).
\]
On $E_{\rm act}$ we have $N(D)=O(|D|)$, and on $\Omega_{n_*}$ we have $\Lambda^\star(D)=\int_D \lambda^\star \le K_{\rm win}|D|$. Since we chose $\varepsilon_{\rm net}=|D|^{-2}$, the discretization error is
\[
O\!\left(L_H\,\varepsilon_{\rm net}\,(N(D)+\Lambda^\star(D))\right)
=O_{\mathbb P}\!\left(\frac{L_H\,(1+K_{\rm win})}{|D|}\right).
\]
Since we chose $\varepsilon_{\rm net}=|D|^{-2}$, this discretization error is $O_{\mathbb P}\!\big(L_H(1+K_{\rm win})/|D|\big)$. When the corresponding term appears with the $|D|^{-1}$ normalization (as in $f_r=|D|^{-1}\ell_r$), the discretization contribution is $O_{\mathbb P}\!\big(L_H(1+K_{\rm win})/|D|^2\big)$, hence negligible.

\end{proof}

\subsubsection{Per-flip decomposition and explicit constants}

We adopt $\sigma_j=+1$ for a $0\!\to\!1$ flip and $\sigma_j=-1$ for a $1\!\to\!0$ flip at $\zeta=\gamma_j$. Recall $\underline\mu_{\min}=\min\{\underline\mu_0,\underline\mu_1\}$ and $\underline\mu_\Sigma$ from Assumption~\ref{ass:G1}.

\begin{lemma}[Single flip preserves non-anticipation on the active prefix]\label{lem:singleflip-nonant}
Fix $m\le \Tmax$ and work on the active window $D^{(m)}$ (as per the convention above).
Let $r=(r_i)_{i\ge 1}$ be a (possibly random) labelling such that $r_i$ is $\mathcal G_{t_i}$-measurable for every
$i\le N^{(m)}$.
Fix an index $j\le N^{(m)}$ and define the single-flip labelling $r^{(\gamma_j)}$ by
\[
r^{(\gamma_j)}_i=
\begin{cases}
r_i,& i\neq j,\\
1-r_j,& i=j.
\end{cases}
\]
Then $r^{(\gamma_j)}_i$ is $\mathcal G_{t_i}$-measurable for every $i\le N^{(m)}$.
\end{lemma}

\begin{proof}
For $i\neq j$, $r^{(\gamma_j)}_i=r_i$ is $\mathcal G_{t_i}$-measurable by assumption.
For $i=j$, $r^{(\gamma_j)}_j=1-r_j$ is a measurable function of $r_j$ and hence $\mathcal G_{t_j}$-measurable.
\end{proof}

\begin{lemma}[Per-flip identity and bounds for $\mathcal J_\alpha$]\label{lem:perflip}
Assume \cref{ass:G1,ass:G2prime,ass:G3}, and Assumption~\ref{ass:G0pred} along the update path, and fix $m\le \Tmax$. Throughout this lemma, all likelihoods, event-sums, and stochastic integrals are taken over the active prefix window
$D^{(m)}=(t^\ast,u_m]\times\mathcal S$ (cf.\ the active-window convention at the start of this section).
Fix $j \le N^{(m)}$, and let $r'=r^{(\zeta)}$ be a single flip at $\zeta=\gamma_j$. By Lemma~\ref{lem:singleflip-nonant}, $r'$ is also non-anticipating, hence all label-induced intensities below are $\mathcal G$-predictable. Then
\begin{equation}\label{eq:perflip-id}
\ell_{r'}^{[u_m]}(\theta)-\ell_r^{[u_m]}(\theta)
=\sigma_j\,\tilde s_\theta(\zeta;r)+\mathcal M_1+\mathcal D_1-\mathcal D_2+\mathcal R_{\rm sel},
\end{equation}
where, for $v\in[0,1]$, writing $\Delta\widehat\lambda^r=\widehat\lambda^{r'}-\widehat\lambda^{r}$,
\begin{align*}
\mathcal M_1&:=\int_0^1\!\!\int_{D^{(m)}} \frac{\Delta\widehat\lambda^r(\tau)}{\widehat\lambda^r(\tau)+v\Delta\widehat\lambda^r(\tau)} \dd M^\star(\tau)\dd v,\\
\mathcal D_1&:=\int_0^1\!\!\int_{D^{(m)}} \frac{\Delta\widehat\lambda^r(\tau)}{\widehat\lambda^r(\tau)+v\Delta\widehat\lambda^r(\tau)}\big(\lambda^\star(\tau)-\widehat\lambda^r(\tau)\big)\dd\tau\dd v,\\
\mathcal D_2&:=\int_0^1\!\!\int_{D^{(m)}} \frac{v (\Delta\widehat\lambda^r(\tau))^2}{\widehat\lambda^r(\tau)+v\Delta\widehat\lambda^r(\tau)} \dd\tau \dd v\ \ (\ge 0).
\end{align*}
and the selection correction is the (pathwise) sum over events
\[
\mathcal R_{\rm sel}
:=\sum_{i:\,\gamma_i\in D^{(m)}}\Big(
\log \widehat\lambda^{r'}_{r_i}\!\big(\gamma_i;\theta\big)
-\log \widehat\lambda^{r}_{r_i}\!\big(\gamma_i;\theta\big)
-\big(\log\widehat\lambda^{r'}(\gamma_i;\theta)-\log\widehat\lambda^{r}(\gamma_i;\theta)\big)
\Big),
\]
The decomposition \eqref{eq:perflip-id} is purely algebraic and holds path-wise, without any probabilistic event.
The quantitative bounds \eqref{eq:drift-bnds}--\eqref{eq:mtgsel-bnd} stated below hold on the master event of Lemma~\ref{lem:master}, since they use only the lower margins in Assumption~\ref{ass:G1}, the window envelope in Assumption~\ref{ass:G2prime}, and the Freedman inequalities included in that master event.
\end{lemma}
\begin{proof}
Fix two labellings $r,r'$ that differ at exactly one index $j$ (the ``flipped'' event), and write
\[
\Delta\widehat\lambda_k^r(\tau):=\widehat\lambda_k^{r'}(\tau)-\widehat\lambda_k^{r}(\tau),\qquad
\Delta\widehat\lambda^{r}(\tau):=\widehat\lambda^{r'}(\tau)-\widehat\lambda^{r}(\tau),
\]
where $\widehat\lambda^{r}=\widehat\lambda_0^{r}+\widehat\lambda_1^{r}$. For the remainder of this proof, we write $D:=D^{(m)}$ for brevity. For $v\in[0,1]$ also define the interpolated intensity
\[
\widehat\lambda^{r,v}(\tau):=\widehat\lambda^{r}(\tau)+v\,\Delta\widehat\lambda^{r}(\tau).
\]

Using the decomposition
\[
\sum_{k\in\{0,1\}}\int_D \log\widehat\lambda_k^r(\tau)\,dN_k^r(\tau)
=
\int_D \log\widehat\lambda^{r}(\tau)\,dN(\tau)
+
\sum_{k\in\{0,1\}}\int_D \log\!\Big(\frac{\widehat\lambda_k^r(\tau)}{\widehat\lambda^{r}(\tau)}\Big)\,dN_k^r(\tau),
\]
we may write the objective (restricted to $D$) as a sum of a total-intensity log-likelihood term and a selection term. Hence
\[
\ell_{r'}^{[u_m]}(\theta)-\ell_r^{[u_m]}(\theta)
=
\Big[\mathcal L_{\mathrm{tot}}(r')-\mathcal L_{\mathrm{tot}}(r)\Big]
+
\Big[\mathcal L_{\mathrm{sel}}(r')-\mathcal L_{\mathrm{sel}}(r)\Big],
\]
where $\mathcal L_{\mathrm{tot}}(r):=\int_D \log\widehat\lambda^{r}\,dN-\int_D \widehat\lambda^{r}\,d\tau$.

For the log term, apply the identity
\[
\log a-\log b
=
(a-b)\int_0^1\frac{1}{b+v(a-b)}\,dv
\qquad(a,b>0),
\]
pointwise with $a=\widehat\lambda^{r'}(\tau)$ and $b=\widehat\lambda^{r}(\tau)$ to obtain
\[
\log\widehat\lambda^{r'}(\tau)-\log\widehat\lambda^{r}(\tau)
=
\int_0^1 \frac{\Delta\widehat\lambda^{r}(\tau)}{\widehat\lambda^{r,v}(\tau)}\,dv.
\]
Therefore,
\begin{align*}
\mathcal L_{\mathrm{tot}}(r')-\mathcal L_{\mathrm{tot}}(r)
&=
\int_D\big(\log\widehat\lambda^{r'}-\log\widehat\lambda^{r}\big)\,dN
-
\int_D \Delta\widehat\lambda^{r}\,d\tau\\
&=
\int_0^1\int_D \frac{\Delta\widehat\lambda^{r}(\tau)}{\widehat\lambda^{r,v}(\tau)}\,dN(\tau)\,dv
-
\int_D \Delta\widehat\lambda^{r}(\tau)\,d\tau.
\end{align*}
Now decompose $dN=dM^\star+\lambda^\star(\tau)\,d\tau$ to get
\begin{align*}
\mathcal L_{\mathrm{tot}}(r')-\mathcal L_{\mathrm{tot}}(r)
&=
\underbrace{\int_0^1\int_D \frac{\Delta\widehat\lambda^{r}(\tau)}{\widehat\lambda^{r,v}(\tau)}\,dM^\star(\tau)\,dv}_{:=\,\mathcal M_1}
\\&\quad+
\int_0^1\int_D \frac{\Delta\widehat\lambda^{r}(\tau)}{\widehat\lambda^{r,v}(\tau)}\,\lambda^\star(\tau)\,d\tau\,dv
-
\int_D \Delta\widehat\lambda^{r}(\tau)\,d\tau\\
&=
\mathcal M_1
+
\int_0^1\int_D \frac{\Delta\widehat\lambda^{r}(\tau)}{\widehat\lambda^{r,v}(\tau)}\big(\lambda^\star(\tau)-\widehat\lambda^{r,v}(\tau)\big)\,d\tau\,dv.
\end{align*}
By Lemma~\ref{lem:stoch-fubini} (stochastic Fubini), we may interchange the $v$-integration and the martingale integral and write
\[
\mathcal M_1=\int_D\left(\int_0^1 \frac{\Delta\widehat\lambda^{r}(\tau)}{\widehat\lambda^{r,v}(\tau)}\,dv\right)\,dM^\star(\tau),
\]
so $\mathcal M_1$ is a well-defined martingale integral.

Using $\lambda^\star-\widehat\lambda^{r,v}=(\lambda^\star-\widehat\lambda^{r})-v\,\Delta\widehat\lambda^{r}$, the last integral becomes
\begin{align*}
&\int_0^1\int_D \frac{\Delta\widehat\lambda^{r}(\tau)\big(\lambda^\star(\tau)-\widehat\lambda^{r}(\tau)\big)}{\widehat\lambda^{r,v}(\tau)}\,d\tau\,dv
-
\int_0^1\int_D \frac{v\big(\Delta\widehat\lambda^{r}(\tau)\big)^2}{\widehat\lambda^{r,v}(\tau)}\,d\tau\,dv\\
&\qquad=: \mathcal D_1-\mathcal D_2.
\end{align*}
Thus $\mathcal L_{\mathrm{tot}}(r')-\mathcal L_{\mathrm{tot}}(r)=\mathcal M_1+\mathcal D_1-\mathcal D_2$, matching the lemma's definitions.

\medskip
We now examine the selection part and the single-flip term. Write the selection term as an event-sum:
\[
\mathcal L_{\mathrm{sel}}(r)
 =
 \sum_{k\in\{0,1\}}\int_D \log\!\Big(\frac{\widehat\lambda_k^{r}(\tau)}{\widehat\lambda^{r}(\tau)}\Big)\,dN_k^r(\tau)
 =
 \sum_{i:\,\gamma_i\in D} \log\!\Big(\frac{\widehat\lambda_{r_i}^{r}(\gamma_i)}{\widehat\lambda^{r}(\gamma_i)}\Big).
\]
Since $r$ and $r'$ differ only at $j$, split the difference into the contribution from $i\neq j$ and the contribution at index $j$:
\begin{align*}
\mathcal L_{\mathrm{sel}}(r')-\mathcal L_{\mathrm{sel}}(r)
&=
\sum_{\substack{i:\,\gamma_i\in D\\ i\neq j}}
\Bigg[
\log\!\Big(\frac{\widehat\lambda_{r_i}^{r'}(\gamma_i)}{\widehat\lambda^{r'}(\gamma_i)}\Big)
-
\log\!\Big(\frac{\widehat\lambda_{r_i}^{r}(\gamma_i)}{\widehat\lambda^{r}(\gamma_i)}\Big)
\Bigg]\\
&\quad+
\Bigg[
\log\!\Big(\frac{\widehat\lambda_{r'_{j}}^{r'}(\gamma_{j})}{\widehat\lambda^{r'}(\gamma_{j})}\Big)
-
\log\!\Big(\frac{\widehat\lambda_{r_{j}}^{r}(\gamma_{j})}{\widehat\lambda^{r}(\gamma_{j})}\Big)
\Bigg].
\end{align*}
By definition, the first sum is exactly the remainder term $\mathcal R_{\mathrm{sel}}$ in the lemma.

For the single flipped event, use predictability: at time $\gamma_{j}$, the pre-$t_j$ histories under $r$ and $r'$ coincide, so $\widehat\lambda_k^{r'}(\gamma_{j})=\widehat\lambda_k^{r}(\gamma_{j})$ for $k\in\{0,1\}$, and likewise $\widehat\lambda^{r'}(\gamma_{j})=\widehat\lambda^{r}(\gamma_{j})$. Therefore the bracket simplifies to
\[
\log\!\Big(\frac{\widehat\lambda_{r'_j}^{\,r}(\gamma_j;\theta)}{\widehat\lambda_{r_j}^{\,r}(\gamma_j;\theta)}\Big)
\;=\;
\sigma_j\,\log\!\Big(\frac{\widehat\lambda_1^{\,r}(\gamma_j;\theta)}{\widehat\lambda_0^{\,r}(\gamma_j;\theta)}\Big)
\;=\;
\sigma_j\,\tilde s_\theta(\zeta;r),
\]
where $\sigma_j=+1$ for a $0\to1$ flip and $\sigma_j=-1$ for a $1\to0$ flip (as defined in the lemma), and $\tilde s_\theta(\zeta;r)$ is the log-ratio score at event $\zeta$.

Combining Step 1 and Step 2 yields the claimed per-flip identity.
\end{proof}

On the window envelope event $\Omega_{n_*}(L,R)$ from Assumption~\ref{ass:G2prime} (and hence on the master event of Lemma~\ref{lem:master}), the denominators in $\mathcal D_1,\mathcal D_2$ are bounded below by $\underline\mu_\Sigma$ by Assumption~\ref{ass:G1}, while the numerators are controlled by the single flip locality bounds in Assumption~\ref{ass:G3}. A direct computation gives
\begin{align}
|\mathcal D_1|
&\le\frac{2B_1}{\underline\mu_\Sigma} K_{\rm win}, \notag\\
\mathcal D_2
&\le\frac{B_2}{2}. \notag
\end{align}
Therefore, on $\Omega_{n_*}(L,R)$ (and hence on the master event of Lemma~\ref{lem:master}), the drift terms satisfy
\begin{align}
|\mathcal D_1| \le \frac{2B_1}{\underline\mu_\Sigma} K_{\rm win},\qquad
\mathcal D_2 \le \frac{B_2}{2}.
\label{eq:drift-bnds}
\end{align}

Using $| \log a-\log b|\le |a-b|/\min\{a,b\}$, Assumption~\ref{ass:G1}, and $|\Delta\widehat\lambda_{\mathrm{sel}_r}|\le |\Delta\widehat\lambda_0^r|+|\Delta\widehat\lambda_1^r|$ at event times, we have the pathwise domination
\[
|\mathcal R_{\rm sel}|
\ \le\ \int_D h(\tau)\,\dd N(\tau),
\qquad
h(\tau):=\frac{|\Delta\widehat\lambda^r(\tau)|}{\underline\mu_\Sigma}+\frac{|\Delta\widehat\lambda_0^r(\tau)|+|\Delta\widehat\lambda_1^r(\tau)|}{\underline\mu_{\min}}.
\]
Since $h$ depends only on past labels through $\Delta\widehat\lambda_k^r$, it is $\mathcal G$-predictable under Assumption~\ref{ass:G0pred}. Decompose
\[
\int_D h\,\dd N=\underbrace{\int_D h\,\dd M^\star}_{=:\ \mathcal M_{\rm sel}}+\underbrace{\int_D h(\tau)\lambda^\star(\tau)\dd\tau}_{=:\ \mathcal D_{\rm sel}}.
\]
On $\Omega_{n_*}(L,R)$, $\|\lambda^\star\|_\infty\le K_{\rm win}$, hence
\[
|\mathcal D_{\rm sel}|
\ \le\ K_{\rm win}\int_D h
\ \le\ \left(\frac{B_1}{\underline\mu_\Sigma}+\frac{B_1^\Sigma}{\underline\mu_{\min}}\right)K_{\rm win}.
\]
Moreover, $\|h\|_\infty\le H_{\rm sel}$ with
\[
H_{\rm sel}:=\frac{B_\infty}{\underline\mu_\Sigma}+\frac{2B_\infty^{\rm comp}}{\underline\mu_{\min}}.
\]

Applying \cref{lem:Freedman} to $\mathcal M_{\rm sel}$ gives
\begin{align}
|\mathcal M_1|
&\le \sqrt{ \frac{2B_2}{\underline\mu_\Sigma} K_{\rm win} z_{\rm Fr}\ +\ \frac{B_\infty^2}{9\underline\mu_\Sigma^2} z_{\rm Fr}^2 }\ +\ \frac{B_\infty}{3\underline\mu_\Sigma} z_{\rm Fr}
\ \le\ \sqrt{\frac{2B_2}{\underline\mu_\Sigma} K_{\rm win} z_{\rm Fr}}\ +\ \frac{2B_\infty}{3\underline\mu_\Sigma} z_{\rm Fr},\label{eq:mtg1-bnd}
\end{align}
where the final inequality uses $\sqrt{a+b}\le \sqrt a+\sqrt b$, producing an additional $(B_\infty/(3\underline\mu_\Sigma))z_{\rm Fr}$ term and hence the coefficient $2/3$ on the linear-in-$z_{\rm Fr}$ part.

Similarly,
\begin{align}
|\mathcal M_{\rm sel}|
&\le \sqrt{ 2K_{\rm win}\!\left(2\frac{B_2}{\underline\mu_\Sigma} + 4\frac{B_2^\Sigma\,\underline\mu_\Sigma}{\underline\mu_{\min}^2}\right) z_{\rm Fr}\ +\ \frac{H_{\rm sel}^2}{9}z_{\rm Fr}^2 }\ +\ \frac{H_{\rm sel}}{3}z_{\rm Fr}\notag\\
&\le 2\sqrt{\left(\frac{B_2}{\underline\mu_\Sigma} + 2\frac{B_2^\Sigma\,\underline\mu_\Sigma}{\underline\mu_{\min}^2}\right)K_{\rm win} z_{\rm Fr}}\ +\ \frac{2}{3}\left(\frac{B_\infty}{\underline\mu_\Sigma} + \frac{2B_\infty^{\rm comp}}{\underline\mu_{\min}}\right)z_{\rm Fr}.
\label{eq:mtgsel-bnd}
\end{align}
Consequently, setting
\[
C_{\rm dr,0}:=\frac{B_2}{2},\qquad
C_{\rm dr,1}:=\frac{3B_1}{\underline\mu_\Sigma}+\frac{B_1^{\Sigma}}{\underline\mu_{\min}},
\]
and
\[
C_{\rm mtg,1}:=\sqrt{\frac{2B_2}{\underline\mu_\Sigma}} + 2\sqrt{\frac{B_2}{\underline\mu_\Sigma}+2\frac{B_2^{\Sigma}\,\underline\mu_\Sigma}{\underline\mu_{\min}^2}},\qquad
C_{\rm mtg,2}:=\frac{4}{3}\!\left(\frac{B_\infty}{\underline\mu_\Sigma}+\frac{B_\infty^{\rm comp}}{\underline\mu_{\min}}\right),
\]
we have the uniform per-flip bound
\begin{equation}\label{eq:perflip-bound}
\Big|\ell_{r'}^{[u_m]}(\theta)-\ell_r^{[u_m]}(\theta)-\sigma_j\,\tilde s_\theta(\zeta;r)\Big|
\ \le\ C_{\rm dr,0}\ +\ C_{\rm dr,1} K_{\rm win}\ +\ C_{\rm mtg,1}\sqrt{K_{\rm win} z_{\rm Fr}}\ +\ C_{\rm mtg,2} z_{\rm Fr}.
\end{equation}
Hence, if the flip at $\zeta\in S_\theta^\pm(b)$ creates a minority, then
\[
\mathcal J_\alpha^{[u_m]}(r',\theta)-\mathcal J_\alpha^{[u_m]}(r,\theta)\ \le\ (C_{\rm dr,1}-\alpha) K_{\rm win}\ +\ \Big(C_{\rm dr,0}+C_{\rm mtg,1}\sqrt{K_{\rm win}z_{\rm Fr}}+C_{\rm mtg,2}z_{\rm Fr}-(b-\Delta_s)\Big).
\]
On \(E^{\rm align}_{|D|}\), Assumption~\ref{ass:G4} implies that for \(\zeta\in S_\theta^+(b)\) we have \(\tilde s_\theta(\zeta;r)\ge b-\Delta_s\), and for \(\zeta\in S_\theta^-(b)\) we have \(\tilde s_\theta(\zeta;r)\le -b+\Delta_s\).If it eliminates a minority, then
\[
\mathcal J_\alpha^{[u_m]}(r',\theta)-\mathcal J_\alpha^{[u_m]}(r,\theta)\ \ge\ (\alpha-C_{\rm dr,1}) K_{\rm win}\ -\ \big(C_{\rm dr,0}+C_{\rm mtg,1}\sqrt{K_{\rm win}z_{\rm Fr}}+C_{\rm mtg,2}z_{\rm Fr}\big)\ +\ (b-\Delta_s).
\]
Therefore the following sufficient condition ensures elimination and non-creation of minorities on $S_\theta^\pm(b)$:
\begin{align}
&(\alpha-C_{\rm dr,1}) K_{\rm win}\ \ge\ C_{\rm mtg,1}\sqrt{K_{\rm win}z_{\rm Fr}}+C_{\rm mtg,2}z_{\rm Fr}+\kappa_b,\qquad
b \ge \Delta_s + C_{\rm dr,0} - \kappa_b,\label{eq:EsufficientB}
\end{align}
for some \(\kappa_b\in[0,C_{\rm dr,0}]\), where \(\Delta_s=\Delta_s(|D|)\) is the alignment slack from Assumption~\ref{ass:G4}, written in shorthand on \(E^{\rm align}_{|D|}\); it is independent of \(b\). Note that the greedy scan \eqref{eq:greedy-estep} eliminates minorities on decisive sets under the weaker condition \(b+\alpha K_{\rm win}\ge \Delta_s\) (Lemma~\ref{lem:E-elim}); \eqref{eq:EsufficientB} is tailored to per-flip comparisons of \(\mathcal J_\alpha^{[u_m]}\) that control future effects.

\begin{remark}[Penalty scaling]\label{rem:alpha-scaling}
On the master event, $K_{\rm win}=O(\log|D|)$ by Assumption~\ref{ass:G2prime}, and with fixed parameter dimension we have $z_{\rm Fr}=\Theta(\log|D|)$ by \eqref{eq:z-def}. Hence the right-hand side of \eqref{eq:EsufficientB} is of order $\log|D|$ (up to polylog factors through $K_{\rm win}$). The effective penalty in \eqref{eq:penobj} is $(\alpha-C_{\rm dr,1})K_{\rm win}$, so whenever $K_{\rm win}$ grows at least on the order of $\log|D|$ (the typical feedback regime captured by $C_{\rm loc}>0$ in \cref{ass:G2prime}), there exists a finite $\alpha_0>C_{\rm dr,1}$ (depending only on the model constants and on the asymptotic ratio $z_{\rm Fr}/K_{\rm win}$) such that any fixed $\alpha\ge\alpha_0$ makes
\[
(\alpha-C_{\rm dr,1})K_{\rm win}
\ge C_{\rm mtg,1}\sqrt{K_{\rm win} z_{\rm Fr}}+C_{\rm mtg,2}z_{\rm Fr}+\kappa_b
\]
for all sufficiently large $|D|$. In other words, in this common regime $\alpha$ does not need to grow with $|D|$.
\end{remark}

\begin{lemma}[Greedy E-step eliminates minorities]\label{lem:E-elim}
Work on the event \(E^{\rm align}_{|D|}\) of Assumption~\ref{ass:G4}. Suppose moreover that
\begin{equation}\label{eq:Eelim-cond}
b+\alpha K_{\rm win}\ \ge\ \Delta_s
\qquad\text{(equivalently, }\alpha K_{\rm win}\ge \Delta_s-b\text{)}.
\end{equation}
If ties in \eqref{eq:greedy-estep} at decisive events are broken in favor of the non-minority label, then at update $m$ the greedy E-step rule \eqref{eq:greedy-estep} assigns no minority labels among events in the processed block $D_{m+1}$:
for every event index $i$ with $u_m<t_i\le u_{m+1}$,
\[
\gamma_i\in S_{\theta^{(m)}}^+(b)\ \Rightarrow\ r_i^{(m+1)}=1,
\qquad
\gamma_i\in S_{\theta^{(m)}}^-(b)\ \Rightarrow\ r_i^{(m+1)}=0.
\]
\end{lemma}

\begin{proof}
Fix an event index $i$ with $u_m<t_i\le u_{m+1}$ and $\gamma_i\in S^+_{\theta^{(m)}}(b)$. The minority label is $0$, hence the greedy score difference between choosing $1$ and choosing $0$ equals
\[
\Big(\log \widehat\lambda_{1}^{r^{(m+1)}}(\gamma_i;\theta^{(m)})\Big)
-\Big(\log \widehat\lambda_{0}^{r^{(m+1)}}(\gamma_i;\theta^{(m)})-\alpha K_{\rm win}\Big)
=\tilde s_{\theta^{(m)}}(\gamma_i;r^{(m+1)})+\alpha K_{\rm win}.
\]

By Assumption~\ref{ass:G4}, $\tilde s_{\theta^{(m)}}(\gamma_i;r^{(m+1)})\ge b-\Delta_s$, hence
\[
\tilde s_{\theta^{(m)}}(\gamma_i;r^{(m+1)})+\alpha K_{\rm win}\ \ge\ b-\Delta_s+\alpha K_{\rm win}\ \ge\ 0
\]
by \eqref{eq:Eelim-cond}. Therefore the non-minority label is chosen (ties broken toward the non-minority label). If $\gamma_i\in S^-_{\theta^{(m)}}(b)$ then $\tilde s_{\theta^{(m)}}(\gamma_i;r^{(m+1)})\le -b+\Delta_s$, so
\[
(-\tilde s_{\theta^{(m)}}(\gamma_i;r^{(m+1)}))+\alpha K_{\rm win}
\ \ge\ b-\Delta_s+\alpha K_{\rm win}\ \ge\ 0,
\]
and the greedy rule chooses label $0$ (ties broken toward the non-minority label). Hence no minorities are assigned on
$S^+_{\theta^{(m)}}(b)\cup S^-_{\theta^{(m)}}(b)$ within the processed block.
\end{proof}

\begin{remark}[Interpretation under non-revisiting labels]\label{rem:no-revisit-meaning}
Under Assumption~\ref{ass:G0pred}(b) labels are never revised, so Lemma~\ref{lem:E-elim} only asserts that the greedy rule assigns no minority labels  among the newly processed events in $D_{m+1}$  relative to the current parameter $\theta^{(m)}$. Previously assigned labels (in earlier blocks) may become minority relative to later iterates $\theta^{(\ell)}$, but the analysis does not require re-elimination: it controls the blockwise mislabelling increments and propagates them to prefix error via the averaging identity in Lemma~\ref{lem:avg-hamming}.
\end{remark}

\subsection{Decisive-set stability and E-step Hamming bound}\label{sec:Estep}

\begin{lemma}[Decisive-set stability on the next processed block]\label{lem:set-stab}
Let $B_\theta$ be as in \cref{def:Btheta}. Fix $m\in\{0,\dots,\Tmax-1\}$ and consider the next block
$D_{m+1}=(u_m,u_{m+1}]\times\mathcal S$.
On $E^{\rm Lip}_{|D|}\cap E^{\rm mass}_{|D|}$ (with $E^{\rm mass}_{|D|}$ including the compensator deviation for
$B_{\theta^{(m)}}\cap D_{m+1}$ as in Assumption~\ref{ass:G8}) we have
\[
N(B_{\theta^{(m)}}\cap D_{m+1})\ \le\ \frac{L_s}{b}\,|D_{m+1}|\,\|\theta^{(m)}-\theta^\star\|
\ +\ N\!\big(A_{\theta^\star}(2b)\cap D_{m+1}\big)\ +\ 2 C_{\rm it}\sqrt{K_{\rm win}|D_{m+1}|\log|D|}.
\]
\end{lemma}

\begin{proof}
Let $\Delta_{\rm LLR}:=s_{\theta^\star}^{\rm or}-s_{\theta^{(m)}}^{\rm or}$.
The symmetric-difference inclusions yield the pointwise bound on $D$:
\[
\1_{B_{\theta^{(m)}}}
\ \le\ 
\1_{A_{\theta^\star}(2b)} + \1_{\{|\Delta_{\rm LLR}|\ge b\}}.
\]
Intersect with $W$ and multiply by $\lambda^\star$ to obtain
\[
\int_{B_{\theta^{(m)}}\cap W}\lambda^\star(\tau)\dd\tau
\ \le\ 
\int_{A_{\theta^\star}(2b)\cap W}\lambda^\star(\tau)\dd\tau
\ +\ 
\frac{1}{b}\int_W |\Delta_{\rm LLR}(\tau)|\,\lambda^\star(\tau)\dd\tau.
\]
On $E^{\rm Lip}_{|D|}$, Assumption~\ref{ass:G5} (applied to the subwindow $W$) gives
\[
\int_W |\Delta_{\rm LLR}(\tau)|\,\lambda^\star(\tau)\dd\tau
\ \le\ 
L_s\,|W|\,\|\theta^{(m)}-\theta^\star\|.
\]
On $E^{\rm mass}_{|D|}$ (strengthened to include the restricted sets) we have the compensator deviations
\[
N(B_{\theta^{(m)}}\cap W)\ \le\ \int_{B_{\theta^{(m)}}\cap W}\lambda^\star(\tau)\dd\tau\ +\ C_{\rm it}\sqrt{K_{\rm win}|W|\log|D|},
\]
\[
N(A_{\theta^\star}(2b)\cap W)\ \ge\ \int_{A_{\theta^\star}(2b)\cap W}\lambda^\star(\tau)\dd\tau\ -\ C_{\rm it}\sqrt{K_{\rm win}|W|\log|D|}.
\]
Combine the last three displays to conclude.
\end{proof}

\begin{proposition}[E-step Hamming bound (block-normalized increment)]\label{prop:Estep}
On the master event of \cref{lem:master} (strengthened if needed by the deterministic block/prefix restrictions
described in Assumption~\ref{ass:G8}), for any update index $m\in\{0,\dots,\Tmax-1\}$ and
$r^{(m+1)}$ produced by the greedy E-step \eqref{eq:greedy-estep}, the block increment
\[
\Delta_{m+1}
:=d_H^{[u_{m+1}]}\!\big(r^{(m+1)},r^\star\big)-d_H^{[u_m]}\!\big(r^{(m)},r^\star\big)
\]
satisfies
\[
\frac{\Delta_{m+1}}{|D_{m+1}|}
\ \le\ \eta_+\big(|D_{m+1}|;b\big) + \eta_-\big(|D_{m+1}|;b\big)
\ +\ \frac{L_s}{b}\,\|\theta^{(m)}-\theta^\star\|
\ +\ \frac{2}{|D_{m+1}|}N\!\big(A_{\theta^\star}(2b)\cap D_{m+1}\big)\ +\ \delta_m,
\]
where
\[
\delta_m\ :=\ C_{\rm H,1}\sqrt{\frac{K_{\rm win}\,z_{\rm Fr}}{|D_{m+1}|}},
\qquad C_{\rm H,1}:=6C_{\rm it}.
\]
\end{proposition}

\begin{proof}
Work on the master event of \cref{lem:master}, so in particular on
$E^{\rm Lip}_{|D|}\cap E^{\rm mass}_{|D|}$ (with the deterministic subwindow restrictions included as in
Assumption~\ref{ass:G8}), and assume \cref{lem:E-elim} applies to the greedy E-step on the processed block $D_{m+1}$.

By Assumption~\ref{ass:G0pred}(b), labels are not revised: $r^{(m+1)}_i=r^{(m)}_i$ for all $i$ with $t_i\le u_m$.
Hence $\Delta_{m+1}$ counts disagreements between $r^{(m+1)}$ and $r^\star$ only among events in $D_{m+1}$.

Define the sets (as subsets of $D$)
\[
R_1:=S_{\theta^{(m)}}^+(b)\cap S_{\theta^\star}^+(b),\quad
R_2:=S_{\theta^{(m)}}^-(b)\cap S_{\theta^\star}^-(b),\quad
R_3:=B_{\theta^{(m)}},\quad
R_4:=A_{\theta^\star}(b).
\]
By \cref{lem:E-elim}, within $D_{m+1}$ all decisive-set labels match the non-minority choice at $\theta^{(m)}$.
Therefore any disagreement within $D_{m+1}$ must occur either as an oracle minority on $R_1$/$R_2$,
or in $B_{\theta^{(m)}}$, or in the ambiguous band $A_{\theta^\star}(b)$, hence
\[
\Delta_{m+1}
\ \le\ N_0(R_1\cap D_{m+1})+N_1(R_2\cap D_{m+1})+N(R_3\cap D_{m+1})+N(R_4\cap D_{m+1}).
\]
Since $R_1\subseteq S_{\theta^\star}^+(b)$ and $R_2\subseteq S_{\theta^\star}^-(b)$,
\[
N_0(R_1\cap D_{m+1})\le N_0\!\big(S_{\theta^\star}^+(b)\cap D_{m+1}\big),\qquad
N_1(R_2\cap D_{m+1})\le N_1\!\big(S_{\theta^\star}^-(b)\cap D_{m+1}\big).
\]
Also $R_3\cap D_{m+1}\subseteq B_{\theta^{(m)}}\cap D_{m+1}$ and
$R_4\cap D_{m+1}\subseteq A_{\theta^\star}(2b)\cap D_{m+1}$, hence
\[
N(R_3\cap D_{m+1})\le N(B_{\theta^{(m)}}\cap D_{m+1}),\qquad
N(R_4\cap D_{m+1})\le N\!\big(A_{\theta^\star}(2b)\cap D_{m+1}\big).
\]

Bound the oracle-minority terms using Assumption~\ref{ass:G8} applied on the block window $D_{m+1}$ and the corresponding
mean deviations included in $E^{\rm mass}_{|D|}$:
\[
\frac{N_0(S_{\theta^\star}^+(b)\cap D_{m+1})+N_1(S_{\theta^\star}^-(b)\cap D_{m+1})}{|D_{m+1}|}
\ \le\ \eta_+(|D_{m+1}|;b)+\eta_-(|D_{m+1}|;b)\ +\ 4C_{\rm it}\sqrt{\frac{K_{\rm win}\log|D|}{|D_{m+1}|}}.
\]
For the $B_{\theta^{(m)}}$ term, apply Lemma~\ref{lem:set-stab} with $W=D_{m+1}$ and divide by $|D_{m+1}|$:
\[
\frac{N(B_{\theta^{(m)}}\cap D_{m+1})}{|D_{m+1}|}
\ \le\ \frac{L_s}{b}\,\|\theta^{(m)}-\theta^\star\|
\ +\ \frac{1}{|D_{m+1}|}N\!\big(A_{\theta^\star}(2b)\cap D_{m+1}\big)
\ +\ 2C_{\rm it}\sqrt{\frac{K_{\rm win}\log|D|}{|D_{m+1}|}}.
\]
Combine the last three displays and use $\log|D|\le z_{\rm Fr}$ for all large $|D|$ to obtain the stated bound
with $\delta_m=C_{\rm H,1}\sqrt{K_{\rm win}z_{\rm Fr}/|D_{m+1}|}$ and $C_{\rm H,1}=6C_{\rm it}$.
\end{proof}

\subsection{M-step bound via tightened label-to-score Lipschitz}\label{sec:Mstep}

\begin{lemma}[Strong concavity: gradient inequality]\label{lem:sc-mon}
Let $\Theta\subset\bbR^p$ be convex and let $f:\Theta\to\bbR$ be differentiable and $m_{\rm sc}$--strongly concave on
$\Theta$. Then for any $\theta,\theta'\in\Theta$,
\[
\big(\nabla f(\theta')-\nabla f(\theta)\big)^\top(\theta'-\theta)\ \le\ -m_{\rm sc}\|\theta'-\theta\|^2.
\]
\end{lemma}

\begin{proof}
Strong concavity of $f$ is equivalent to $m_{\rm sc}$--strong convexity of $-f$. Hence for all $\theta,\theta'\in\Theta$,
\[
\big(\nabla(-f)(\theta')-\nabla(-f)(\theta)\big)^\top(\theta'-\theta)\ \ge\ m_{\rm sc}\|\theta'-\theta\|^2.
\]
Multiplying by $-1$ gives the claim.
\end{proof}

\begin{proposition}[M-step parameter bound]\label{prop:Mstep}
On the master event of \cref{lem:master} and under \cref{ass:G6,ass:G7prime,ass:G9} (applied to the cumulative window $D^{(m+1)}$),
for any update index $m\in\{0,\dots,\Tmax-1\}$ and $r=r^{(m+1)}$, with any maximiser
\[
\hat\theta_r^{[u_{m+1}]}\ \in\ \arg\max_{\theta\in\Theta_{\rm sc}}\ell_r^{[u_{m+1}]}(\theta),
\]
we have
\[
\norm{\hat\theta_r^{[u_{m+1}]}-\theta^\star}\ \le\ 
\underbrace{\frac{\bar C_0 K_{\rm win}+\bar C_1\sqrt{K_{\rm win} z_{\rm Fr}}+\bar C_2 z_{\rm Fr}}{m_{\rm sc}}}_{=:B_{\rm main}}
\frac{d_H^{[u_{m+1}]}(r,r^\star)}{|D^{(m+1)}|}
\ +\ 
\underbrace{\frac{C_{\rm sc}}{m_{\rm sc}}\!\left(\sqrt{\frac{K_{\rm win} z_{\rm Fr}}{|D^{(m+1)}|}} + \frac{z_{\rm Fr}}{|D^{(m+1)}|}\right)}_{=:\xi_{|D^{(m+1)}|}}.
\]
In particular, under \cref{ass:G2prime} we have $K_{\rm win}=O(\log|D|)$, and for any sequence of update indices
with $|D^{(m+1)}|\to\infty$ we have $\xi_{|D^{(m+1)}|}\to 0$.
\end{proposition}

\begin{proof}
Let $u:=u_{m+1}$ and write $\hat\theta:=\hat\theta_r^{[u]}$ for brevity.
Since $\hat\theta,\theta^\star\in\Theta_{\rm sc}$ and (on the master event) $f_{r^\star}^{[u]}$ is $m_{\rm sc}$--strongly concave on $\Theta_{\rm sc}$
(Assumption~\ref{ass:G6}, applied to the restricted objective on $D^{[u]}$), apply \cref{lem:sc-mon} on $\Theta=\Theta_{\rm sc}$ with
$f=f_{r^\star}^{[u]}$, $\theta=\theta^\star$ and $\theta'=\hat\theta$ to get
\[
-m_{\rm sc}\|\hat\theta-\theta^\star\|^2\ \ge\ \big(\nabla f_{r^\star}^{[u]}(\hat\theta)-\nabla f_{r^\star}^{[u]}(\theta^\star)\big)^\top(\hat\theta-\theta^\star).
\]
Rearranging,
\[
m_{\rm sc}\|\hat\theta-\theta^\star\|^2\ \le\ \big(\nabla f_{r^\star}^{[u]}(\theta^\star)-\nabla f_{r^\star}^{[u]}(\hat\theta)\big)^\top(\hat\theta-\theta^\star).
\]
Insert and subtract $\nabla f_r^{[u]}(\hat\theta)$:
\begin{align*}
m_{\rm sc}\|\hat\theta-\theta^\star\|^2
&\le \big(\nabla f_{r^\star}^{[u]}(\theta^\star)-\nabla f_r^{[u]}(\hat\theta)\big)^\top(\hat\theta-\theta^\star)
+ \big(\nabla f_r^{[u]}(\hat\theta)-\nabla f_{r^\star}^{[u]}(\hat\theta)\big)^\top(\hat\theta-\theta^\star).
\end{align*}
By Assumption~\ref{ass:G9} applied to the restricted objective $f_r^{[u]}$ with $\theta=\theta^\star$, we have
$(\theta^\star-\hat\theta)^\top \nabla f_r^{[u]}(\hat\theta)\le 0$, i.e.\ $(\hat\theta-\theta^\star)^\top \nabla f_r^{[u]}(\hat\theta)\ge 0$. Hence
\[
\big(\nabla f_{r^\star}^{[u]}(\theta^\star)-\nabla f_r^{[u]}(\hat\theta)\big)^\top(\hat\theta-\theta^\star)
\le \nabla f_{r^\star}^{[u]}(\theta^\star)^\top(\hat\theta-\theta^\star)
\le \|\nabla f_{r^\star}^{[u]}(\theta^\star)\|\,\|\hat\theta-\theta^\star\|.
\]
Also,
\[
\big(\nabla f_r^{[u]}(\hat\theta)-\nabla f_{r^\star}^{[u]}(\hat\theta)\big)^\top(\hat\theta-\theta^\star)
\le \|\nabla f_r^{[u]}(\hat\theta)-\nabla f_{r^\star}^{[u]}(\hat\theta)\|\,\|\hat\theta-\theta^\star\|.
\]
Dividing by $m_{\rm sc}\|\hat\theta-\theta^\star\|$ (or arguing trivially if $\hat\theta=\theta^\star$) yields
\[
m_{\rm sc}\|\hat\theta-\theta^\star\|
\le \|\nabla f_r^{[u]}(\hat\theta)-\nabla f_{r^\star}^{[u]}(\hat\theta)\| + \|\nabla f_{r^\star}^{[u]}(\theta^\star)\|.
\]
On the master event, the first norm is bounded by Assumption~\ref{ass:G7prime}, and the second by Assumption~\ref{ass:G6}(b) (applied on the prefix window $D^{[u]}$):
\[
\|\nabla f_r^{[u]}(\hat\theta)-\nabla f_{r^\star}^{[u]}(\hat\theta)\|
\ \le\ \frac{d_H^{[u]}(r,r^\star)}{|D^{[u]}|}\Big(\bar C_0 K_{\rm win} + \bar C_1\sqrt{K_{\rm win} z_{\rm Fr}} + \bar C_2 z_{\rm Fr}\Big),
\]
\[
\|\nabla f_{r^\star}^{[u]}(\theta^\star)\|
\ \le\ C_{\rm sc}\!\left(\sqrt{\frac{K_{\rm win} z_{\rm Fr}}{|D^{[u]}|}} + \frac{z_{\rm Fr}}{|D^{[u]}|}\right).
\]
Divide by $m_{\rm sc}$ and note $|D^{[u]}|=|D^{(m+1)}|$.
\end{proof}

\subsection{Two-line recursion and contraction to a statistical floor}\label{sec:contraction}

\begin{lemma}[Invariant local neighborhood]\label{lem:inv-sc}
Fix $r_{\rm sc}>0$ and let $\Theta_{\rm sc}:=\Theta_\circ\cap B(\theta^\star,r_{\rm sc})$.
Assume that on an event $E$ the iterate sequence satisfies, for all $m<\Tmax$,
\[
\|\theta^{(m+1)}-\theta^\star\|\ \le\ \rho\,\|\theta^{(m)}-\theta^\star\| + s,
\qquad\text{with }0\le\rho<1,\ s\ge 0.
\]
If additionally $\|\theta^{(0)}-\theta^\star\|\le r_{\rm sc}$ and $s\le (1-\rho)\,r_{\rm sc}$, then on $E$,
\[
\|\theta^{(m)}-\theta^\star\|\ \le\ r_{\rm sc}\qquad\forall m\le \Tmax,
\]
so in particular $\theta^{(m)}\in\Theta_{\rm sc}$ for all $m\le\Tmax$.
\end{lemma}

\begin{proof}
By induction. The base case is $\|\theta^{(0)}-\theta^\star\|\le r_{\rm sc}$.
If $\|\theta^{(m)}-\theta^\star\|\le r_{\rm sc}$ then
\[
\|\theta^{(m+1)}-\theta^\star\|
\le \rho\,r_{\rm sc}+s
\le \rho\,r_{\rm sc}+(1-\rho)\,r_{\rm sc}
= r_{\rm sc}.
\]
\end{proof}

\begin{lemma}[Averaging identity for cumulative Hamming error]\label{lem:avg-hamming}
For each $m\in\{0,\dots,\Tmax-1\}$ define the block increment
\[
\Delta_{m+1}
:=d_H^{[u_{m+1}]}\!\big(r^{(m+1)},r^\star\big)-d_H^{[u_m]}\!\big(r^{(m)},r^\star\big),
\]
which counts disagreements between $r^{(m+1)}$ and $r^\star$ among events in the new block
$D_{m+1}=(u_m,u_{m+1}]\times\mathcal S$.

Define the prefix mislabel rates by $e_0:=0$ and, for $m\ge 1$,
\[
e_m:=\frac{d_H^{[u_m]}\!\big(r^{(m)},r^\star\big)}{|D^{(m)}|}.
\]
Define the block mislabel rate $\varepsilon_{m+1}:=\Delta_{m+1}/|D_{m+1}|$.
Then for every $m\in\{0,\dots,\Tmax-1\}$ we have the exact identity
\[
e_{m+1}
=
\omega_m\,e_m
+
(1-\omega_m)\,\varepsilon_{m+1},
\qquad
\omega_m:=\frac{|D^{(m)}|}{|D^{(m+1)}|}\in[0,1),
\]
(with $\omega_0=0$ since $|D^{(0)}|=0$).
\end{lemma}

\begin{proof}
Since labels are not revised (Assumption~\ref{ass:G0pred}(b)), we have
$d_H^{[u_{m+1}]}=d_H^{[u_m]}+\Delta_{m+1}$ for every $m$.
Divide by $|D^{(m+1)}|=|D^{(m)}|+|D_{m+1}|$ to obtain
\[
\frac{d_H^{[u_{m+1}]}}{|D^{(m+1)}|}
=
\frac{|D^{(m)}|}{|D^{(m+1)}|}\cdot \frac{d_H^{[u_m]}}{|D^{(m)}|}
+
\frac{|D_{m+1}|}{|D^{(m+1)}|}\cdot \frac{\Delta_{m+1}}{|D_{m+1}|}.
\]
For $m\ge 1$ this is exactly the claimed identity with $e_m=d_H^{[u_m]}/|D^{(m)}|$.
For $m=0$ the same formula holds by the convention $e_0:=0$ and $\omega_0=|D^{(0)}|/|D^{(1)}|=0$.
\end{proof}

\begin{lemma}[When localised and unlocalised M-steps coincide]\label{lem:Mstep-coincide}
Let $\Theta_\circ\subset\Theta$ be compact and convex, and let $\Theta_{\rm loc}\subseteq \Theta_\circ$ be nonempty
with $\theta^\star\in\Theta_{\rm loc}$. Fix $u\in[t^\ast,T]$ and $\Delta_{\rm gap}>0$ and assume the oracle gap on the
restricted objective:
\begin{equation}\label{eq:oracle-gap}
\sup_{\theta\in \Theta_\circ\setminus \Theta_{\rm loc}} f_{r^\star}^{[u]}(\theta)
\ \le\ f_{r^\star}^{[u]}(\theta^\star)-\Delta_{\rm gap}.
\end{equation}
Let $r$ be any labelling (possibly random) such that on an event $E_{\rm gap}^{[u]}$,
\begin{equation}\label{eq:uniform-perturb-gap}
\sup_{\theta\in\Theta_\circ}\big|f_r^{[u]}(\theta)-f_{r^\star}^{[u]}(\theta)\big|
\ \le\ \frac{\Delta_{\rm gap}}{4}.
\end{equation}
Then on $E_{\rm gap}^{[u]}$ every maximiser of $f_r^{[u]}$ over $\Theta_\circ$ lies in $\Theta_{\rm loc}$, and moreover
\[
\arg\max_{\theta\in\Theta_\circ} f_r^{[u]}(\theta)
\ =\ \arg\max_{\theta\in\Theta_{\rm loc}} f_r^{[u]}(\theta).
\]
Equivalently, the localised and unlocalised restricted M-steps coincide:
\[
\arg\max_{\theta\in\Theta_\circ}\ell_r^{[u]}(\theta)
\ =\ \arg\max_{\theta\in\Theta_{\rm loc}}\ell_r^{[u]}(\theta).
\]
\end{lemma}

\begin{proof}
Fix $\theta\in\Theta_\circ\setminus\Theta_{\rm loc}$. On $E_{\rm gap}^{[u]}$,
\[
f_r^{[u]}(\theta)\ \le\ f_{r^\star}^{[u]}(\theta)+\frac{\Delta_{\rm gap}}{4}
\ \le\ f_{r^\star}^{[u]}(\theta^\star)-\Delta_{\rm gap}+\frac{\Delta_{\rm gap}}{4}
\ =\ f_{r^\star}^{[u]}(\theta^\star)-\frac{3}{4}\Delta_{\rm gap}.
\]
Also on $E_{\rm gap}^{[u]}$,
\[
f_r^{[u]}(\theta^\star)\ \ge\ f_{r^\star}^{[u]}(\theta^\star)-\frac{\Delta_{\rm gap}}{4}.
\]
Therefore
\[
f_r^{[u]}(\theta)\ \le\ f_r^{[u]}(\theta^\star)-\frac{\Delta_{\rm gap}}{2},
\]
so no maximiser of $f_r^{[u]}$ over $\Theta_\circ$ can lie in $\Theta_\circ\setminus\Theta_{\rm loc}$.
Since $\Theta_{\rm loc}\subseteq\Theta_\circ$, the maximiser sets over $\Theta_\circ$ and $\Theta_{\rm loc}$ coincide.
\end{proof}

\begin{theorem}[Recursion and contraction to a statistical floor]\label{thm:contract}
Assume \cref{ass:G0,ass:G0pred,ass:G1,ass:G2prime,ass:G3,ass:G4,ass:G5,ass:G6,ass:G7prime,ass:G8,ass:G9,ass:warmstart}. 
For the E-step, run the greedy update \eqref{eq:greedy-estep}. Pick any \(b>0\) such that Assumption~\ref{ass:G8} holds at that threshold, and choose \(\alpha\) and \(\kappa_b\in[0,C_{\rm dr,0}]\) so that \eqref{eq:EsufficientB} holds.
Set
\[
A:=\frac{L_s}{b},\qquad B_{\rm main}:=\frac{\bar C_0 K_{\rm win}+\bar C_1\sqrt{K_{\rm win} z_{\rm Fr}}+\bar C_2 z_{\rm Fr}}{m_{\rm sc}}.
\]
Assume additionally that the deterministic schedule satisfies the uniform growth bound
\begin{equation}\label{eq:block-omega}
\omega\ :=\ \max_{m\le \Tmax-1}\ \frac{|D^{(m)}|}{|D^{(m+1)}|}\ <\ 1.
\end{equation}
(For instance, a dyadic/doubling schedule has $\omega\approx\tfrac12$ for all $m$.)

Define the cumulative prefix mislabel rates
\[
e_m\ :=\ \frac{d_H^{[u_m]}\!\big(r^{(m)},r^\star\big)}{|D^{(m)}|}\qquad(m=1,\dots,\Tmax),
\qquad e_0:=0,
\]
and define, for each block $D_{m+1}=(u_m,u_{m+1}]\times\mathcal S$, the block mislabel rate
\[
\varepsilon_{m+1}\ :=\ \frac{\Delta_{m+1}}{|D_{m+1}|}
\ =\ \frac{d_H^{[u_{m+1}]}\!\big(r^{(m+1)},r^\star\big)-d_H^{[u_m]}\!\big(r^{(m)},r^\star\big)}{|D_{m+1}|}.
\]
Finally define the (blockwise) statistical terms
\begin{align*}
\varepsilon_{b,m}\ :=\ \eta_+(|D_{m+1}|;b)+\eta_-(|D_{m+1}|;b)\ +\ \frac{2}{|D_{m+1}|}N\!\big(A_{\theta^\star}(2b)\cap D_{m+1}\big),\\
\delta_{m}\ :=&\ C_{\rm H,1}\sqrt{\frac{K_{\rm win} z_{\rm Fr}}{|D_{m+1}|}},
\end{align*}
and the (prefix) M-step noise level
\[
\xi_m\ :=\ \frac{C_{\rm sc}}{m_{\rm sc}}\!\left(\sqrt{\frac{K_{\rm win} z_{\rm Fr}}{|D^{(m)}|}} + \frac{z_{\rm Fr}}{|D^{(m)}|}\right)
\qquad(m=1,\dots,\Tmax),
\]
with $C_{\rm H,1}=6C_{\rm it}$.

On the same high-probability event as in \cref{lem:master}, strengthened if needed by a finite union over the
deterministic family of windows $\{D^{(m)}\}_{m\le\Tmax}$ and blocks $\{D_m\}_{m\le\Tmax}$ for the fixed-set concentration
events invoked in \cref{ass:G8}, the following hold simultaneously:
\begin{enumerate}[nosep,label=(\roman*)]
\item \textbf{(Block error bound.)} For every $m\in\{0,\dots,\Tmax-1\}$,
\begin{equation}\label{eq:block-mislabel-rate}
\varepsilon_{m+1}\ \le\ A\,\|\theta^{(m)}-\theta^\star\|\ +\ \varepsilon_{b,m}\ +\ \delta_m.
\end{equation}

\item \textbf{(Prefix averaging recursion.)} For every $m\in\{0,\dots,\Tmax-1\}$,
\begin{equation}\label{eq:prefix-avg-rec}
e_{m+1}
\ \le\ \omega_m\,e_m\ +\ (1-\omega_m)\,\varepsilon_{m+1},
\qquad
\omega_m:=\frac{|D^{(m)}|}{|D^{(m+1)}|}\ \le\ \omega,
\end{equation}
and hence, combining with \eqref{eq:block-mislabel-rate},
\begin{equation}\label{eq:e-recursion}
e_{m+1}
\ \le\ \rho_m\,e_m\ +\ (1-\omega_m)\big(A\,\xi_m+\varepsilon_{b,m}+\delta_m\big)
\qquad(m\ge 1),
\end{equation}
where
\[
\rho_m:=\omega_m+(1-\omega_m)AB_{\rm main},
\qquad
\rho:=\max_{0\le j\le \Tmax-1}\rho_j.
\]
Note that if $AB_{\rm main}\le 1$ (in particular whenever $\rho<1$), then $\rho=\omega+(1-\omega)AB_{\rm main}$ and $\rho_m\le\rho$. (The base block $m=0$ is initialized by \eqref{eq:block-mislabel-rate} with $\theta^{(0)}$.)

\item \textbf{(M-step bound on prefixes.)} For every $m\in\{0,\dots,\Tmax-1\}$,
\begin{equation}\label{eq:theta-from-e}
\|\theta^{(m+1)}-\theta^\star\|
\ \le\ B_{\rm main}\,e_{m+1}\ +\ \xi_{m+1}.
\end{equation}
\end{enumerate}

In particular, if $\rho<1$ then for every $m\in\{1,\dots,\Tmax\}$,
\[
e_m\ \le\ \rho^{m-1}e_1\ +\ \frac{\max_{1\le j\le m-1}\big(A\,\xi_j+\varepsilon_{b,j}+\delta_j\big)}{1-\rho},
\qquad
\|\theta^{(m)}-\theta^\star\|\ \le\ B_{\rm main}\,e_m+\xi_m.
\]
Thus \(\theta^{(\Tmax)}\) contracts to a data-dependent statistical floor. Passing from this finite-sample statement to consistency requires the additional array-level compatibility conditions isolated in Proposition~\ref{prop:compatibility}.
\end{theorem}

\begin{proof}[Proof of \cref{thm:contract}]
Work on the stated high-probability event. Fix $m\in\{0,\dots,\Tmax-1\}$. The quantity $\Delta_{m+1}$ counts mislabels in the block $D_{m+1}$ produced by the greedy
rule \eqref{eq:greedy-estep} at parameter $\theta^{(m)}$.
Proposition~\ref{prop:Estep} gives \eqref{eq:block-mislabel-rate} directly (with $\varepsilon_{m+1}=\Delta_{m+1}/|D_{m+1}|$). Identity \eqref{eq:prefix-avg-rec} is exactly Lemma~\ref{lem:avg-hamming} written in terms of $(e_m)$ and $(\varepsilon_{m+1})$.
Combining \eqref{eq:prefix-avg-rec} with \eqref{eq:block-mislabel-rate} gives
\[
e_{m+1}\ \le\ \omega_m e_m+(1-\omega_m)\Big(A\|\theta^{(m)}-\theta^\star\|+\varepsilon_{b,m}+\delta_m\Big).
\]
For $m\ge 1$, apply Proposition~\ref{prop:Mstep} at index $m-1$ to bound
$\|\theta^{(m)}-\theta^\star\|\le B_{\rm main}e_m+\xi_m$, yielding \eqref{eq:e-recursion}.

Equation \eqref{eq:theta-from-e} is exactly Proposition~\ref{prop:Mstep} rewritten with $e_{m+1}=d_H^{[u_{m+1}]}/|D^{(m+1)}|$. Finally, if $\rho<1$, iterating \eqref{eq:e-recursion} (with $\rho_m\le\rho$) gives the displayed geometric bound on $e_m$,
and plugging into \eqref{eq:theta-from-e} yields the bound on $\|\theta^{(m)}-\theta^\star\|$.
The final consistency statement follows by combining the displayed bound on \(e_m\) with \eqref{eq:theta-from-e}: under
\[
B_{\rm main}\,\rho^{\Tmax-1} e_1 \xrightarrow{\mathbb P}0,\qquad
B_{\rm main}\,\frac{\max_{1\le j\le \Tmax-1}\big(A\,\xi_j+\varepsilon_{b,j}+\delta_j\big)}{1-\rho}\xrightarrow{\mathbb P}0,\qquad
\xi_{\Tmax}\xrightarrow{\mathbb P}0,
\]
we obtain \(\|\theta^{(\Tmax)}-\theta^\star\|\xrightarrow{\mathbb P}0\).
\end{proof}

\begin{remark}[Array-level sufficient conditions]\label{rem:array-rho}
In the terminal conclusion of \cref{thm:contract}, the quantities that must vanish are exactly
\[
B_{\rm main}\,\rho^{\Tmax-1} e_1,\qquad
B_{\rm main}\,\frac{\max_{1\le j\le \Tmax-1}\big(A\,\xi_j+\varepsilon_{b,j}+\delta_j\big)}{1-\rho},
\qquad
\xi_{\Tmax}.
\]
A convenient sufficient regime along the triangular array is
\[
\sup_{|D|\ \text{large}}\rho\le \bar\rho<1,\qquad
B_{\rm main}\,e_1=O_{\mathbb P}(1),\qquad
B_{\rm main}\,\max_{1\le j\le \Tmax-1}\big(A\,\xi_j+\varepsilon_{b,j}+\delta_j\big)\xrightarrow{\mathbb P}0,\qquad
\xi_{\Tmax}\xrightarrow{\mathbb P}0,
\]
because then \(1-\rho\ge 1-\bar\rho>0\) and, since \(\Tmax\asymp\log|D|\), also
\(\rho^{\Tmax-1}\le \bar\rho^{\Tmax-1}\to 0\).

Using the definitions of \(\xi_j\) and \(\delta_j\), the middle condition is implied by
\[
B_{\rm main}\,\max_{1\le j\le \Tmax-1}\left[
A\left(\sqrt{\frac{K_{\rm win} z_{\rm Fr}}{|D^{(j)}|}}+\frac{z_{\rm Fr}}{|D^{(j)}|}\right)
+\varepsilon_{b,j}
+\sqrt{\frac{K_{\rm win} z_{\rm Fr}}{|D_{j+1}|}}
\right]\xrightarrow{\mathbb P}0.
\]
In particular, a bound of the form \(\max_m \omega_m<1\) alone does not guarantee the terminal consistency claim.
\end{remark}

\begin{remark}[Localized versus unlocalized M-step]\label{rem:localized-theorem}
The recursion in \cref{thm:contract} is proved for Algorithm~\ref{alg:blockwise-hardem}, i.e.\ with the localized M-step
\[
\theta^{(m+1)}\in \arg\max_{\theta\in\Theta_{\rm sc}}\ell_{r^{(m+1)}}^{[u_{m+1}]}(\theta).
\]
To identify this with the ordinary unconstrained M-step over \(\Theta_\circ\), one additionally needs the gap condition of \cref{lem:Mstep-coincide}.
\end{remark}

\begin{remark}[How \(b\) enters the contraction factor]\label{rem:rho-choice}
The contraction factor satisfies
\[
\rho=\omega+(1-\omega)\frac{L_s}{b}\,B_{\rm main}.
\]
Thus taking \(b\) larger improves contraction. This is only a contraction-oriented choice, however. A consistency choice must simultaneously control the ambiguous-band term and, in Hawkes models, dominate the alignment slack \(\Delta_s\). Proposition~\ref{prop:compatibility} makes this joint requirement explicit. In particular, choosing \(b\asymp K_{\rm win}\) may help ensure \(\rho<1\), but by itself it does not imply a vanishing-floor regime.
\end{remark}

\subsection{When does the floor vanish? Parameter regimes}\label{sec:param-regimes}

We give two sufficient regimes for the truncated kernel Hawkes model of \S\ref{app:hawkes} that do not require strict separation. In \cref{thm:contract}, the dominant statistical floor enters through the blockwise ambiguous-band term $|D_{m+1}|^{-1}N(A_{\theta^\star}(2b)\cap D_{m+1})$. By Assumption~\ref{ass:G8} this concentrates around its compensator, so it suffices to control the $\lambda^\star$-mass of the ambiguous band on deterministic blocks. Accordingly, all bounds below hold verbatim with $D$ replaced by any deterministic block $D_{m+1}$ (or prefix $D^{(m)}$). Throughout, $C_b(\theta^\star)$ is as in \cref{def:C_b} and probabilities/expectations are under $P_{\theta^\star}$.

Write $\mu_\Sigma^{\rm base}(\tau):=\mu_0(\tau;\theta^\star)+\mu_1(\tau;\theta^\star)$ and
\[
U(\tau)\ :=\ \sum_{\ell=0}^1\big(G_{\bullet\ell}*N_\ell\big)(\tau)\ \ (\ge 0),\qquad
\kappa\ :=\ \operatorname*{ess sup}_{\tau\in D}\ \frac{U(\tau)}{\mu_\Sigma^{\rm base}(\tau)}.
\]
Note that $\kappa$ is random through $U(\tau)$. In the truncated Hawkes setting with memory $h$ and radius $R$,
on the window-count event $\Omega_n(h,R)$ we have
\[
U(\tau)=\sum_{\ell=0}^1 (G_{\bullet\ell}*N_\ell)(\tau)
\le G_{\max}\,N\big((t-h,t)\times B_R(x)\big)\ \le\ G_{\max}\,n,
\qquad \tau=(t,x).
\]
Hence, if $\underline\mu^{\rm base}_\Sigma:=\inf_{\tau\in D}\mu_\Sigma^{\rm base}(\tau)>0$, then
\[
\kappa\ \le\ \frac{G_{\max}\,n}{\underline\mu^{\rm base}_\Sigma}\qquad\text{on }\Omega_n(h,R).
\]
Under Assumption~\ref{ass:CSC}, $a_{\rm csc}u_k U\ \le\ \sum_{\ell} (g_{k\ell}*N_\ell)\ \le\ b_{\rm csc}u_k U$ for $k\in\{0,1\}$.

We first define the $\lambda^\star$-weighted law on $D$. Let
\[
\Lambda^\star(D)\ :=\ \int_D \lambda^\star(\tau)\,\dd\tau.
\]
Define the probability measure on $D$ with density proportional to $\lambda^\star$ by
\[
\mathbb P_{\lambda^\star}(A)
\ :=\ \frac{1}{\Lambda^\star(D)}\int_D \mathbf 1\{\tau\in A\}\,\lambda^\star(\tau)\,\dd\tau,
\qquad A\subseteq D\ \text{measurable}.
\]
Equivalently, for any measurable $g$,
\[
\frac{1}{|D|}\int_D g(\tau)\,\lambda^\star(\tau)\,\dd\tau
\ =\ \frac{\Lambda^\star(D)}{|D|}\,\mathbb E_{\lambda^\star}[g(\tau)],
\qquad \tau\sim \mathbb P_{\lambda^\star}.
\]
In particular,
\[
C_b(\theta^\star)=\frac{\Lambda^\star(D)}{|D|}\,
\mathbb P_{\lambda^\star}\!\big(|s^{\rm or}_{\theta^\star}(\tau)|\le 2b\big).
\]

\begin{lemma}[Regime A: weak excitation]\label{lem:regA}
Assume \cref{ass:CSC}. Suppose there exists $\epsilon_0>0$ and $L_0<\infty$ such that the $\lambda^\star$-weighted distribution of the baseline log-odds
\[
S_0(\tau):=\log\frac{\mu_1(\tau;\theta^\star)}{\mu_0(\tau;\theta^\star)}
\]
has a density bounded by $L_0$ on $[-\epsilon_0,\epsilon_0]$. In particular, under $\tau\sim\mathbb P_{\lambda^\star}$,
the random variable $S_0(\tau)$ admits a density $f_{S_0,\lambda^\star}$ on $[-\epsilon_0,\epsilon_0]$ satisfying
\[
\sup_{|x|\le \epsilon_0} f_{S_0,\lambda^\star}(x)\ \le\ L_0.
\]

Define
\[
C_A'\ :=\ \left(\frac{b_{\rm csc}u_1}{\underline\mu_1} + \frac{b_{\rm csc}u_0}{\underline\mu_0}\right)\cdot \sup_{\tau}\mu_\Sigma^{\rm base}(\tau).
\]
Then for any $b>0$ with $2b+C_A'\kappa\le \epsilon_0$,
\[
C_b(\theta^\star)
\ \le\ 2L_0\,(2b+C_A'\kappa)\,\frac{\Lambda^\star(D)}{|D|}.
\]
The bound is most informative when $\Lambda^\star(D)/|D|=O_{\mathbb P}(1)$ (e.g.\ under stationary/ergodic regimes where $\Lambda^\star(D)/|D|\to \bar\lambda$ in probability). In our finite-sample analysis we always have the envelope $\Lambda^\star(D)/|D|\le K_{\rm win}$ on $\Omega_{n_*}(h,R)$, while in typical stationary regimes $K_{\rm win}$ is only a worst-case upper bound.

The trivial envelope $\Lambda^\star(D)/|D|\le K_{\rm win}$ on the window event (hence $K_{\rm win}=O(\log|D|)$)is a worst-case bound and should be viewed only as an envelope, not as the typical scaling. It follows that
\[
C_b(\theta^\star)\ \le\ 2L_0\,(2b+C_A'\kappa)\,K_{\rm win}\qquad\text{on }\Omega_{n_*}(h,R).
\]
Thus $C_b(\theta^\star)$ can be made small by taking $b=b(|D|)\downarrow 0$ and ensuring $\kappa$ is small (e.g.\ on a high-probability envelope event), trading off against the band-mass term.

\end{lemma}

\begin{proof}
Under \cref{ass:CSC} and the bounds $a_{\rm csc}u_k U\le \sum_\ell (g_{k\ell}*N_\ell)\le b_{\rm csc}u_k U$, we obtain
\[
\log\left(\frac{\mu_1+a_{\rm csc}u_1 U}{\mu_0+b_{\rm csc}u_0 U}\right)\ \le\ s^{\rm or}_{\theta^\star}\ \le\ \log\left(\frac{\mu_1+b_{\rm csc}u_1 U}{\mu_0+a_{\rm csc}u_0 U}\right).
\]
Write $s^{\rm or}_{\theta^\star}-S_0=\log(1+\delta_1)-\log(1+\delta_0)$ with
\(
\delta_1=\frac{c_1u_1U}{\mu_1},\ \delta_0=\frac{c_0u_0U}{\mu_0},
\)
where $c_1,c_0\in[a_{\rm csc},b_{\rm csc}]$ and $\delta_0,\delta_1\ge0$. Using $\log(1+u)\le u$ for $u\ge0$,
\[
\big|s^{\rm or}_{\theta^\star}-S_0\big|
\ \le\ \log(1+\delta_1)+\log(1+\delta_0)
\ \le\ \delta_1+\delta_0
\ \le\ \left(\frac{b_{\rm csc}u_1}{\mu_1}+\frac{b_{\rm csc}u_0}{\mu_0}\right)U.
\]
Since $\mu_k\ge \underline\mu_k$, $U\le \kappa \mu_\Sigma^{\rm base}(\tau)$, and $\mu_\Sigma^{\rm base}(\tau)\le \sup_\tau \mu_\Sigma^{\rm base}(\tau)$, we get
\[
\big|s^{\rm or}_{\theta^\star}-S_0\big|\ \le\ C_A' \kappa.
\]
Therefore,
\[
\{ |s^{\rm or}_{\theta^\star}|\le 2b \}\ \subseteq\ \{ |S_0|\le 2b+C_A'\kappa \}.
\]
Multiplying by $\lambda^\star$, integrating, and dividing by $|D|$ yields
\begin{align*}
C_b(\theta^\star)
&\le \frac{1}{|D|}\int_D \1\{|S_0(\tau)|\le 2b+C_A'\kappa\}\,\lambda^\star(\tau)\,\dd\tau \\
&= \frac{\Lambda^\star(D)}{|D|}\,
\mathbb P_{\lambda^\star}\!\big(|S_0(\tau)|\le 2b+C_A'\kappa\big)
\ \le\ 2L_0\,(2b+C_A'\kappa)\,\frac{\Lambda^\star(D)}{|D|},
\end{align*}
whenever $2b+C_A'\kappa\le \epsilon_0$.
\end{proof}

\begin{lemma}[Regime B: dominant column bias]\label{lem:regB}
Assume \cref{ass:CSC} and baseline comparability $c_\mu\le \mu_1/\mu_0\le C_\mu$ as in \S\ref{app:hawkes-align}. Let $q:=b_{\rm csc}/a_{\rm csc}$ and define the bias gap
\[
\Delta_{\rm bias}\ :=\ \log\frac{u_1}{u_0}\ -\ \log q.
\]
Suppose $C_\mu\ <\ (u_1/u_0)/q$ (dominant bias), and assume further that the baseline is aligned with the bias:
\[
\mu_1(\tau;\theta^\star)\ \ge\ \mu_0(\tau;\theta^\star)\qquad\text{for a.e.\ }\tau\in D,
\]
equivalently $S_0(\tau)=\log(\mu_1/\mu_0)\ge 0$ a.e. Then for any $b>0$,
\[
\{ |s^{\rm or}_{\theta^\star}|\le 2b \}\ \subseteq\ \{ |S_0|\le 2b \}.
\]
Consequently, if the $\lambda^\star$-weighted distribution of $S_0$ has a density bounded by $L_0$ on $[0,2b]$, then
\[
C_b(\theta^\star)
\ \le\ 2L_0 b\,\frac{\Lambda^\star(D)}{|D|}
\ \le\ 2L_0 b\,K_{\rm win}\qquad\text{on }\Omega_{n_*}(h,R).
\]
\end{lemma}

\begin{proof}
For any $\tau$, define $\phi_{\rm low}(S):=\log\frac{\mu_1(\tau;\theta^\star)+a_{\rm csc}u_1 S}{\mu_0(\tau;\theta^\star)+b_{\rm csc}u_0 S}$ for $S\ge 0$. Its derivative is
\[
\phi_{\rm low}'(S)\ =\ \frac{a_{\rm csc}u_1 \mu_0\ -\ b_{\rm csc}u_0 \mu_1}{(\mu_1+a_{\rm csc}u_1 S)(\mu_0+b_{\rm csc}u_0 S)}.
\]
The dominant-bias condition $C_\mu < (u_1/u_0)/q$ implies $\frac{\mu_1}{\mu_0} < \frac{a_{\rm csc} u_1}{b_{\rm csc} u_0}$, so the numerator $a_{\rm csc} u_1 \mu_0 - b_{\rm csc} u_0 \mu_1$ is strictly positive. Thus $\phi_{\rm low}$ is strictly increasing in $S\ge 0$, with $\phi_{\rm low}(0)=S_0(\tau)$ and $\lim_{S\to\infty}\phi_{\rm low}(S)=\log(a_{\rm csc}u_1/b_{\rm csc}u_0)=\Delta_{\rm bias}$.

By the bracketing inequality from \cref{lem:regA}, we have $s^{\rm or}_{\theta^\star}(\tau)\ge \phi_{\rm low}(U(\tau))$. Since $\phi_{\rm low}$ is increasing, $\phi_{\rm low}(U)\ge \phi_{\rm low}(0)=S_0(\tau)$. Under the alignment assumption, $S_0(\tau)\ge 0$, hence
\[
s^{\rm or}_{\theta^\star}(\tau)\ \ge\ S_0(\tau)\ \ge\ 0,
\]
so $|s^{\rm or}_{\theta^\star}(\tau)| = s^{\rm or}_{\theta^\star}(\tau)$. Therefore, if $|s^{\rm or}_{\theta^\star}(\tau)|\le 2b$ then
\[
0\ \le\ S_0(\tau)\ \le\ s^{\rm or}_{\theta^\star}(\tau)\ \le\ 2b,
\]
which implies $|S_0(\tau)|\le 2b$. This proves the inclusion. The bound on $C_b(\theta^\star)$ follows by integrating $\1\{|S_0|\le 2b\}$ against $\lambda^\star$ and using the density bound.
\end{proof}

\begin{remark}[Necessity of alignment for Regime B]
If the baseline favored label 0 (so $S_0\ll -2b$) while the feedback bias favored label 1 (so $\Delta_{\rm bias}>2b$), then the score trajectory could cross zero. In the crossing region one may have $|s^{\rm or}_{\theta^\star}|\le 2b$ even though $|S_0|>2b$, invalidating the inclusion in Lemma~\ref{lem:regB}. The alignment condition rules out this cancellation mechanism.
\end{remark}

Strict separation, $\inf_\tau |s^{\rm or}_{\theta^\star}(\tau)|>b$, is sufficient but much stronger than necessary. Lemmas~\ref{lem:regA}--\ref{lem:regB} show that it suffices to make the $\lambda^\star$-mass of the ambiguous band small; the algorithmic penalties then eliminate minorities on decisive sets, and the blockwise recursion of \cref{thm:contract} handles the rest.

Throughout, we fixed a scalar threshold \(b>0\) to define the decisive sets \(S_\theta^\pm(b)\) and the ambiguous band \(A_\theta(b)\). A mild generalization is possible if, at a given update \(m\), one uses a deterministic threshold field
\(b_m:D\to(0,\infty)\) on the active block. In that case the indicators
\(\mathbf 1\{s^{\rm or}_{\theta^{(m)}}(\tau)\ge b_m(\tau)\}\) and
\(\mathbf 1\{s^{\rm or}_{\theta^{(m)}}(\tau)\le -b_m(\tau)\}\)
remain \(\mathcal G\)-predictable, and the arguments below go through with
\[
b_{\min}^{(m)}:=\operatorname*{ess inf}_{\tau\in D} b_m(\tau),
\qquad
b_{\max}^{(m)}:=\operatorname*{ess sup}_{\tau\in D} b_m(\tau).
\]
 Define
\begin{align*}
S_\theta^+\!\big(b_m(\cdot)\big)=&\{\tau:\ s^{\rm or}_\theta(\tau)\ge b_m(\tau)\}\\
S_\theta^-\!\big(b_m(\cdot)\big)=&\{\tau:\ s^{\rm or}_\theta(\tau)\le - b_m(\tau)\}\\
A_\theta\!\big(b_m(\cdot)\big)=&\{\tau:\ |s^{\rm or}_\theta(\tau)|\le b_m(\tau)\}
\end{align*}
and let \(b_{\min}^{(m)}:=\operatorname*{ess inf}_\tau b_m(\tau)\) and \(b_{\max}^{(m)}:=\operatorname*{ess sup}_\tau b_m(\tau)\).
Assumption~\ref{ass:G4} holds pointwise on all of $D$ (and hence in particular on $S_\theta^\pm\!\big(b_m(\cdot)\big)$), and the per-flip bounds of Lemma~\ref{lem:perflip} do not depend on $b$. Consequently, the sufficient condition \eqref{eq:EsufficientB} with \(b\) replaced by \(b_{\min}^{(m)}\) implies the E step elimination conclusion of Lemma~\ref{lem:E-elim} for the sets \(S_{\theta^{(m)}}^\pm\!\big(b_m(\cdot)\big)\).

The decisive set stability bound also adapts with a minimal change. Writing \(\Delta_{\rm LLR}:=s^{\rm or}_{\theta^\star}-s^{\rm or}_{\theta^{(m)}}\), the pointwise inclusions
\begin{align*}
\1_{S_{\theta^{(m)}}^+(b_m) \Delta S_{\theta^\star}^+(b_m)}\ \le&\ \1_{\{|s^{\rm or}_{\theta^\star}|\le 2b_m(\tau)\}}+\1_{\{|\Delta_{\rm LLR}|\ge b_m(\tau)\}},\\
\1_{S_{\theta^{(m)}}^-(b_m) \Delta S_{\theta^\star}^-(b_m)}\ \le&\ \1_{\{|s^{\rm or}_{\theta^\star}|\le 2b_m(\tau)\}}+\1_{\{|\Delta_{\rm LLR}|\ge b_m(\tau)\}}
\end{align*}
yield, after multiplying by \(\lambda^\star\) and integrating, the analogue of Lemma~\ref{lem:set-stab} with \(b\) replaced by \(b_{\min}^{(m)}\) in the Lipschitz term and \(A_{\theta^\star}(2b)\) replaced by \(A_{\theta^\star}(2b_{\max}^{(m)})\). On the event \(E^{\rm Lip}_{|D|}\cap E^{\rm mass}_{|D|}\),
\[
N\!\big(B_{\theta^{(m)}}\cap D_{m+1}\big)\ \le\ \frac{L_s}{b_{\min}^{(m)}} |D_{m+1}| \|\theta^{(m)}-\theta^\star\|
\ +\ N\!\big(A_{\theta^\star}(2b_{\max}^{(m)})\cap D_{m+1}\big)\ +\ 2 C_{\rm it}\sqrt{K_{\rm win}|D_{m+1}|\log|D|}.
\]
Accordingly, the blockwise recursion in \cref{thm:contract} holds with \(A=L_s/b\) replaced by \(A_m=L_s/b_{\min}^{(m)}\),
with \(N\!\big(A_{\theta^\star}(2b)\cap D_{m+1}\big)\) replaced by \(N\!\big(A_{\theta^\star}(2b_{\max}^{(m)})\cap D_{m+1}\big)\),
and with the contraction factors read as
\[
\rho_m\ =\ \omega_m+(1-\omega_m)\,A_m\,B_{\rm main},
\qquad \omega_m=\frac{|D^{(m)}|}{|D^{(m+1)}|}.
\]
Thus the extension is immediate for deterministic threshold fields. Extending the argument to random/predictable threshold fields would require additional concentration assumptions for the resulting random oracle sets, so we do not pursue that generalization here.

\subsection{Causal estimands as functionals}
\label{sec:functional-transfer}

Fix a tessellation $\{\mathcal I_j\}_{j=1}^J$ of \(D=(t^\ast,T]\times\mathcal S\) with \(|\mathcal I_j|>0\).
For a fixed allocation \(z\in\{0,1\}^J\), let \(\Nproc^z\) denote the post-treatment potential process under \(\mathrm{do}(Z=z)\). Under the regime-stable superposition model of Assumption~\ref{ass:id-structural},
\[
\Nproc^z=\Nzero^z+\None^z,
\]
with predictable component intensities
\[
\lambda_{k,\theta}^z(\tau\mid \mathcal F_{t-}^z,\CovX),\qquad k\in\{0,1\},
\]
and total intensity
\[
\lambda_\theta^z(\tau\mid \mathcal F_{t-}^z,\CovX)
=
\lambda_{0,\theta}^z(\tau\mid \mathcal F_{t-}^z,\CovX)
+
\lambda_{1,\theta}^z(\tau\mid \mathcal F_{t-}^z,\CovX).
\]
Here \(\mathcal F_t^z\) is the post-treatment filtration under allocation \(z\) from \cref{sec:identification}. The allocation \(z\) may affect both component intensities; it does not select one active component cellwise. Consequently, spillover allows \(\Nzero^z\)-events to occur in treated cells and \(\None^z\)-events to occur in untreated cells.

Define the interventional mean measure
\[
\Lambda_\theta^z(B)
:=
\bbE_\theta\!\left[\Nproc^z(B)\right]
=
\bbE_\theta\!\left[
\int_B \lambda_\theta^z(\tau\mid \mathcal F_{t-}^z,\CovX)\,\dd\tau
\right],
\qquad B\subseteq D\ \text{Borel}.
\]
Equivalently, if one wishes to keep the conditioning from \cref{sec:identification}, one may write
\[
\Lambda_\theta^z(B\mid \mathcal B_0)
:=
\bbE_\theta\!\left[\Nproc^z(B)\mid \mathcal B_0\right]
=
\bbE_\theta\!\left[
\int_B \lambda_\theta^z(\tau\mid \mathcal F_{t-}^z,\CovX)\,\dd\tau
\,\Big|\, \mathcal B_0
\right].
\]

All causal estimands in \cref{sec:causalEstimands} are finite linear functionals of \(z\mapsto \Lambda_\theta^z\). For example,
\[
\psi_j(z)=\Lambda_\theta^{z^{j=1}}(\mathcal I_j)-\Lambda_\theta^{z^{j=0}}(\mathcal I_j),
\qquad
\psi(z)=\frac{1}{J}\sum_{j=1}^J \psi_j(z),
\]
and
\[
\psi(z_a,z_b)
=
\Lambda_\theta^{z_a}(D)-\Lambda_\theta^{z_b}(D).
\]
More generally, for a finite linear functional
\[
F(\theta)=\sum_{\ell=1}^M c_\ell \Lambda_\theta^{z_\ell}(B_\ell),
\]
the transfer argument below applies verbatim.

\begin{assumption}[Intervention Lipschitz envelope]\label{ass:Lambda-Lip}
There exists \(L_\lambda<\infty\) such that, for all \(\theta,\theta'\in\Theta_\circ\), all fixed allocations \(z\in\{0,1\}^J\),
and all Borel \(B\subseteq D\),
\begin{equation}\label{eq:Lambda-Lip}
\big| \Lambda_\theta^z(B)-\Lambda_{\theta'}^z(B) \big|
\ \le\ L_\lambda |B| \|\theta-\theta'\|.
\end{equation}
\end{assumption}

\begin{proposition}[Functional transfer principle]\label{prop:transfer-functional}
Under \cref{ass:Lambda-Lip} and the blockwise recursion of \cref{thm:contract}, for any finite collection
$\{(c_\ell,z_\ell,B_\ell)\}_{\ell=1}^M$ define the functional
$F(\theta):=\sum_{\ell=1}^M c_\ell \Lambda_\theta^{z_\ell}(B_\ell)$.
Then, on the same high probability event as \cref{thm:contract}, for every update index $m\le \Tmax$,
\[
\big| F(\theta^{(m)})-F(\theta^\star) \big|
\ \le\ L_\lambda \Big(\sum_{\ell=1}^M |c_\ell| |B_\ell|\Big)\,\|\theta^{(m)}-\theta^\star\|.
\]
In particular, with $e_m=d_H^{[u_m]}(r^{(m)},r^\star)/|D^{(m)}|$ and $\xi_m$ as in \cref{thm:contract},
\[
\big| F(\theta^{(m)})-F(\theta^\star) \big|
\ \le\ L_\lambda \Big(\sum_{\ell=1}^M |c_\ell| |B_\ell|\Big)\,\big(B_{\rm main}\,e_m+\xi_m\big).
\]
If additionally $\rho<1$ in \cref{thm:contract}, then for $m\ge 1$,
\[
\big| F(\theta^{(m)})-F(\theta^\star) \big|
\ \le\ L_\lambda \Big(\sum_{\ell=1}^M |c_\ell| |B_\ell|\Big)\left(
B_{\rm main}\left(\rho^{m-1}e_1+\frac{\max_{1\le j\le m-1}\big(A\xi_j+\varepsilon_{b,j}+\delta_j\big)}{1-\rho}\right)+\xi_m\right).
\]
\end{proposition}

\begin{proof}
By linearity and \eqref{eq:Lambda-Lip},
$|F(\theta)-F(\theta')|\le L_\lambda(\sum_\ell |c_\ell||B_\ell|)\|\theta-\theta'\|$.
Apply this with $\theta=\theta^{(m)}$ and $\theta'=\theta^\star$.
The bounds in terms of $(e_m)$ and $(\xi_m)$ follow directly from \eqref{eq:theta-from-e} in \cref{thm:contract}.
\end{proof}

\begin{corollary}[Compact bounds for ITE/AITE/DAITE/TAITE/DTAITE]\label{cor:compact-estimands}
With $L_\lambda$ as above, the plug-in estimators at $\theta^{(m)}$ satisfy
\[
\big|\widehat\psi_j(z)-\psi_j(z)\big|\ \le\ 2L_\lambda |\mathcal I_j|\,\|\theta^{(m)}-\theta^\star\|,
\]
\[
\big|\widehat\psi(z)-\psi(z)\big|\ \le\ \frac{2L_\lambda}{J}\Big(\sum_{j=1}^J |\mathcal I_j|\Big)\,\|\theta^{(m)}-\theta^\star\|,
\]
\[
\big|\widehat\psi(z_a,z_b)-\psi(z_a,z_b)\big|\ \le\ 2L_\lambda |D|\,\|\theta^{(m)}-\theta^\star\|,
\]
\[
\left|\frac{\widehat\psi(z_a,z_b)-\psi(z_a,z_b)}{|D|}\right|
\ \le\ 2L_\lambda\,\|\theta^{(m)}-\theta^\star\|.
\]
and, more generally, any TAITE/DTAITE formed as a finite average of cell-level contrasts inherits the AITE scaling by linearity: its plug-in error is bounded by a constant multiple of \((|D|/J)\,\|\theta^{(m)}-\theta^\star\|\), where the constant depends only on the averaging scheme. In particular, on the event of \cref{thm:contract},
\[
\|\theta^{(m)}-\theta^\star\|\ \le\ B_{\rm main}\,e_m+\xi_m.
\]
Hence, whenever \(\|\theta^{(m)}-\theta^\star\|\xrightarrow{\mathbb P}0\), the following plug-in consistency statements hold:
\begin{enumerate}[nosep,label=(\alph*)]
\item ITEs are consistent for any fixed cell, more generally whenever \(\sup_j |\mathcal I_j|<\infty\);
\item AITE, and likewise TAITE/DTAITE, are consistent whenever the average cell volume is uniformly bounded, i.e.
\[
\frac{1}{J}\sum_{j=1}^J |\mathcal I_j|=\frac{|D|}{J}=O(1);
\]
\item for DAITE, parameter consistency alone yields consistency only for the normalized estimand \(|D|^{-1}\psi(z_a,z_b)\); the unnormalized DAITE additionally requires
\[
|D|\,\|\theta^{(m)}-\theta^\star\|\xrightarrow{\mathbb P}0.
\]
\end{enumerate}
\end{corollary}

\section{Verification for parametric families}
\subsection{Specialization: truncated-kernel Hawkes}\label{app:hawkes}

This section verifies the model-dependent parts of \cref{ass:G0,ass:G1,ass:G2prime,ass:G4,ass:G5} for Hawkes models with nonnegative kernels of finite temporal memory \(h\) and finite spatial range \(R\). Further, it derives \cref{ass:G7prime} from the standing gradient-level regularity part of \cref{ass:G3}. Predictability of the analysed blockwise operator is enforced algorithmically (see \cref{ass:G0pred} and Lemma~\ref{lem:greedy-G0pred} in \cref{sec:perflip}). The value-level part of \cref{ass:G3} is verified explicitly below; the gradient-level part of \cref{ass:G3} (bounds on \(\Delta(\nabla_\theta\widehat\lambda)\), \(S_\infty\), and \(H_\infty\)) is used through \cref{lem:score-nearzero,lem:theta-Lip-integrands,prop:hawkes-G7prime} and should either be checked separately from the smoothness of \(\mu_k\) and the kernels or retained as an additional standing regularity assumption. Assumption~\ref{ass:G8} remains a separate model-dependent high-signal/concentration requirement; \S\ref{sec:param-regimes} only gives sufficient regimes for small ambiguous-band mass and does not constitute a full verification of \cref{ass:G8}. For convenience:
\begin{itemize}[nosep]
\item Activity and tails: \cref{ass:G0} (Proposition~\ref{prop:hawkes-G0}).
\item Window envelope: \cref{ass:G2prime} (Lemma~\ref{lem:hawkes-window} and Proposition~\ref{prop:hawkes-tail}).
\item Uniform margins and value-level single-flip locality: \cref{ass:G1} and the value-level part of \cref{ass:G3} (Lemmas~\ref{lem:hawkes-margin} and \ref{lem:B-constants}).
\item Alignment and LLR Lipschitz: \cref{ass:G4,ass:G5} (see \S\ref{app:hawkes-align} and \cref{lem:hawkes-Lip-avg}).
\item Label-to-score Lipschitz: \cref{ass:G7prime} from \cref{ass:G3} via Proposition~\ref{prop:hawkes-G7prime}.
\end{itemize}

\subsubsection{Model and parameter space}
For $k\in\{0,1\}$, the $\mathcal G$-predictable Hawkes process intensities are
\[
\lambda_k(\tau\mid \mathcal G_{t-};\theta)
=\mu_k(\tau;\theta)+\sum_{\ell=0}^1\int g_{k\ell}(\tau-\tau';\theta)\,\dd N_\ell(\tau'),
\qquad \tau=(t,x),
\]
where each \(g_{k\ell}(\cdot;\theta)\ge 0\) is jointly measurable, supported on \((0,h]\times B_R(0)\), and satisfies
\[
\sup_{\theta\in\Theta_\circ}\|g_{k\ell}(\cdot;\theta)\|_\infty<\infty,\qquad
\sup_{\theta\in\Theta_\circ}\|g_{k\ell}(\cdot;\theta)\|_{L^1}<\infty,\qquad
\sup_{\theta\in\Theta_\circ}\|g_{k\ell}(\cdot;\theta)\|_{L^2}<\infty.
\]
The truncated exponential family
\[
g_{k\ell}(t,x;\theta)=\alpha_{k\ell}e^{-\beta t}w_{k\ell}(x)\,\mathbf 1_{(0,h]}(t)\,\mathbf 1_{B_R(0)}(x)
\]
is a special case.

Under a candidate labelling $r$, the label-induced intensities (built from $\widehat N_\ell^r$) take the explicit form
\[
\widehat\lambda_k^r(\tau;\theta)
=\mu_k(\tau;\theta)+\sum_{\ell=0}^1\int g_{k\ell}(\tau-\tau';\theta)\,\dd \widehat N_\ell^r(\tau').
\]

Assume a compact parameter neighborhood $\Theta_\circ$ with
\begin{align*}
0\le \alpha_{k\ell}\le \alpha_{\max},\qquad \sup_{\theta\in\Theta_\circ,\ \tau\in D}\ \mu_k(\tau;\theta)\le \overline\mu_k\\
\inf_{\theta\in\Theta_\circ,\ \tau\in D}\ \mu_k(\tau;\theta)\ \ge\ \underline\mu_k>0
\end{align*}
for \(k\in\{0,1\}\) and write \(M_\Sigma^{\rm base}:=\overline\mu_0+\overline\mu_1\). Further assume fixed $\mathcal S\subset\bbR^d$ is bounded with finite Lebesgue measure $\mu$, and $B_R(x)$ denotes the Euclidean ball of radius $R$.

We work throughout conditional on an observed history that covers the memory window (e.g.\ \((t^\ast-h,t^\ast]\times\mathcal S\) for temporal support \(h\)), in the same sense as \S\ref{sec:setup}. For notational convenience we often specialize to the empty-history case
\begin{equation}\label{eq:Hpre}
\tag{Hpre}
N\big((t^\ast -h,t^\ast]\times\mathcal S\big)=0.
\end{equation}
In this case the pseudo-baseline defined below reduces to \(\tilde\mu_k=\mu_k\).
All bounds used in the Hawkes process verification extend verbatim to a conditioned start, provided the conditioned history covers the memory window. The analysis extends immediately to the case where the process is observed on a history interval $[0,t^\ast)$ (or $(-\infty,t^\ast)$) and we maximise the likelihood conditioned on this history.

In this setting, the intensity on $D$ admits the decomposition
\[
\lambda_k(\tau) = \underbrace{\left(\mu_k(\tau;\theta) + \sum_{\ell=0}^1 \int_{(-\infty,t^\ast)} g_{k\ell}(\tau-\tau') \dd N_\ell(\tau')\right)}_{=:\ \tilde\mu_k(\tau;\theta)} + \sum_{\ell=0}^1 \int_{[t^\ast,t)} g_{k\ell}(\tau-\tau') \dd N_\ell(\tau').
\]
The term $\tilde\mu_k$ acts as a spatially varying, predictable pseudo-baseline. Since the pre-$t^\ast$ events are fixed observations, they do not participate in the E-step labelling updates.
Furthermore, the activity bounds in Proposition~\ref{prop:hawkes-tail} rely on domination by a stationary process $\overline N^{\rm stat}$ (which possesses a history on $(-\infty, t^\ast]$), so the window envelopes $K_{\rm win}$ and concentration results remain valid without modification under a conditioned start.

We write \(\mathcal B_R:=B_R(0)\subset\bbR^d\). Let
\[
G_{\bullet\ell}:=g_{0\ell}+g_{1\ell},\qquad
\Gamma_\ell:=\int_0^h\!\!\int_{\mathcal B_R} G_{\bullet\ell}(t,x) \dd\mu(x)\dd t,\qquad
G_{\max}:=\max_\ell \|G_{\bullet\ell}\|_\infty.
\]
Assume the column $L^1$ masses $\Gamma_\ell$ are finite.

\begin{assumption}[Hawkes subcriticality]\label{ass:H-stab}
Let $A\in\bbR^{2\times 2}$ have entries $A_{k\ell}=\int_0^h\!\!\int_{\mathcal B_R} g_{k\ell}(t,x) \dd\mu(x)\dd t$. Assume the spectral radius $\xi(A)<1$.
\end{assumption}

\subsubsection{Predictability of the analysed operator (verifies \cref{ass:G0pred})}

Assumption~\ref{ass:G0pred} is enforced by the analyzed operator (Algorithm~\ref{alg:blockwise-hardem}) and is therefore model-agnostic.
In particular, Lemma~\ref{lem:greedy-G0pred} shows that the blockwise greedy E-step assigns labels non-anticipatively (eventwise $\mathcal G_{t_i}$-measurability on each processed block), and the M-step at block end $u_m$ produces $\theta^{(m)}$ that is $\mathcal G_{u_m}$-measurable since it depends only on the restriction of the data to the cumulative window $D^{(m)}=(t^\ast,u_m]\times\mathcal S$.
Thus, in the Hawkes process setting no additional model-specific verification of Assumption~\ref{ass:G0pred} is required.
We next verify the model-dependent activity/tail condition of Assumption~\ref{ass:G0}.

\begin{proposition}[Activity and exponential tails for $N(D)$]
\label{prop:hawkes-G0}
Consider the truncated kernel Hawkes model of \cref{app:hawkes} on $D=(t^\ast,T]\times \mathcal{S}$, with $g_{k\ell}\ge 0$ supported in $(0,h]\times B_R(0)$ and bounded baselines $\mu_k(\tau;\theta^\star)\le \overline\mu_k$ for $k\in\{0,1\}$. Assume Hawkes process subcriticality Assumption~\ref{ass:H-stab}, i.e.\ the reproduction matrix $A$ satisfies $\xi(A)<1$. Further assume the empty pre-history condition \eqref{eq:Hpre}. Then there exist constants $C_{\rm act}<\infty$, $\eta_{\rm act}\in(0,1)$ and $c>0$ (depending only on the Hawkes process parameters and on $\Theta_\circ$, but not on $D$) such that for all sufficiently large $|D|$,
\[
\bbE N(D)\ \le\ C_{\rm act}\,|D|,
\qquad
\bbP\!\big(N(D)\ge (1+\eta_{\rm act})C_{\rm act}|D|\big)\ \le\ e^{-c|D|}.
\]
In particular, Assumption~\ref{ass:G0} holds.
\end{proposition}

\begin{proof}
Because $g_{k\ell}\ge 0$ and $\xi(A)<1$, the multivariate Hawkes process admits the standard Poisson cluster (branching) representation on $\bbR\times\mathcal S$ (see, e.g., the stability/cluster constructions used already in Proposition~\ref{prop:hawkes-tail}). In this representation, events arise from immigrants generated by an inhomogeneous Poisson process with intensity $\mu_k(\tau;\theta^\star)$ for each mark $k$, and offspring generated recursively, wherein a parent of mark $\ell$ produces children of mark $k$ according to an inhomogeneous Poisson process with intensity kernel $g_{k\ell}$ shifted to the parent location.

Let $N^{\rm imm}(D)$ denote the number of immigrants in $D$ (all marks). Since
$\mu_0(\tau;\theta^\star)+\mu_1(\tau;\theta^\star)\le \overline\mu_0+\overline\mu_1=:M_\Sigma^{\rm base}$ pointwise,
\[
N^{\rm imm}(D)\ \preceq\ \mathrm{Pois}(M_\Sigma^{\rm base}\,|D|),
\qquad
\bbE N^{\rm imm}(D)\ \le\ M_\Sigma^{\rm base}\,|D|.
\]

For $k\in\{0,1\}$, let $Z^{(k)}$ be the total progeny size (all marks, including the ancestor) of a
cluster started by a single immigrant of mark $k$, in the unrestricted process (no truncation by $D$).
Under $\xi(A)<1$ and compactly supported kernels, $(Z^{(k)})$ has finite mean and admits a nontrivial
exponential moment: there exists $\theta_0>0$ such that
\[
m_1\ :=\ \max_{k\in\{0,1\}}\bbE Z^{(k)}\ <\ \infty,
\qquad
m_\theta\ :=\ \max_{k\in\{0,1\}}\bbE\big[e^{\theta Z^{(k)}}\big]\ <\ \infty
\quad\text{for all }\theta\in(0,\theta_0].
\]
(Heuristically: the offspring counts are Poisson with means $A_{k\ell}$, and subcriticality implies the
associated multitype Galton--Watson total population has a finite exponential moment for small $\theta>0$.)

Now dominate $N(D)$ by the total number of points generated by immigrants in $D$:
if $(Z_i)$ are i.i.d.\ copies of a random variable $Z$ satisfying $Z\overset{d}{\ge} Z^{(k)}$ for both $k$
(e.g.\ take $Z$ to be a mixture maximising the mgf), then
\[
N(D)\ \le\ \sum_{i=1}^{N^{\rm imm}(D)} Z_i
\qquad\text{a.s.}
\]
Taking expectations gives
\[
\bbE N(D)\ \le\ \bbE N^{\rm imm}(D)\,\bbE Z\ \le\ M_\Sigma^{\rm base}\,m_1\,|D|.
\]
Thus Assumption~\ref{ass:G0} holds with $C_{\rm act}:=M_\Sigma^{\rm base}m_1$.

For the tail bound, use the compound-Poisson mgf: for any $\theta\in(0,\theta_0]$,
\begin{align*}
\bbE\big[e^{\theta N(D)}\big]
\ \le&\ \bbE\!\left[\exp\!\left(\theta\sum_{i=1}^{N^{\rm imm}(D)} Z_i\right)\right]\\
\ =&\ \exp\!\Big(\bbE N^{\rm imm}(D)\,( \bbE[e^{\theta Z}]-1)\Big)\\
\ \le&\ \exp\!\big(M_\Sigma^{\rm base}|D|\,(m_\theta-1)\big).
\end{align*}
A Chernoff bound yields, for any $x>0$,
\[
\bbP\!\big(N(D)\ge x|D|\big)
\ \le\ \exp\!\Big(-|D|\big(\theta x - M_\Sigma^{\rm base}(m_\theta-1)\big)\Big).
\]
Since $x\mapsto \sup_{\theta\in(0,\theta_0]}(\theta x - M_\Sigma^{\rm base}(m_\theta-1))$ is strictly positive for
all $x$ above the mean rate, there exist $\eta_{\rm act}\in(0,1)$ and $c>0$ such that taking
$x=(1+\eta_{\rm act})C_{\rm act}$ gives the claimed $e^{-c|D|}$ bound for all large $|D|$.
\end{proof}

\subsubsection{Margins, single-flip locality, and window envelopes (verifies \cref{ass:G1,ass:G3,ass:G2prime})}
\label{app:hawkes-B}

Lemma~\ref{lem:B-constants} provides explicit single-flip influence bounds for  $\Delta\widehat\lambda^r$ (and componentwise analogues), verifying the value-level part of Assumption~\ref{ass:G3}. Lemma~\ref{lem:hawkes-margin} yields the uniform lower margins required in Assumption~\ref{ass:G1}. Lemma~\ref{lem:hawkes-window} gives the deterministic envelope on $\Omega_n(h,R)$ and Proposition~\ref{prop:hawkes-tail} supplies the polynomial tail bound for $\Omega_{n_*}(h,R)^c$, together establishing Assumption~\ref{ass:G2prime}.

\begin{lemma}[Single-flip influence constants]\label{lem:B-constants}
Let \(G_{\bullet\ell}(\cdot;\theta):=g_{0\ell}(\cdot;\theta)+g_{1\ell}(\cdot;\theta)\). For a single flip at \(\zeta\),
\[
\int_D |\Delta\widehat\lambda^r| \dd\tau\ \le\ \|G_{\bullet,1}(\cdot;\theta)-G_{\bullet,0}(\cdot;\theta)\|_{L^1},
\]
\[
\int_D \frac{(\Delta\widehat\lambda^r)^2}{\underline\mu_\Sigma} \dd\tau\ \le\ \frac{\|G_{\bullet,1}(\cdot;\theta)-G_{\bullet,0}(\cdot;\theta)\|_{L^2}^2}{\underline\mu_\Sigma},
\qquad
\|\Delta\widehat\lambda^r\|_\infty\ \le\ \|G_{\bullet,1}(\cdot;\theta)-G_{\bullet,0}(\cdot;\theta)\|_\infty.
\]
Analogous componentwise inequalities hold with \(G_{\bullet,1}-G_{\bullet,0}\) replaced by \(g_{k1}-g_{k0}\). In particular, Assumption~\ref{ass:G3} holds with
\[
B_1:=\sup_{\theta\in\Theta_\circ}\|G_{\bullet,1}(\cdot;\theta)-G_{\bullet,0}(\cdot;\theta)\|_{L^1},
\]
\[
B_2:=\sup_{\theta\in\Theta_\circ}\frac{\|G_{\bullet,1}(\cdot;\theta)-G_{\bullet,0}(\cdot;\theta)\|_{L^2}^2}{\underline\mu_\Sigma},
\qquad
B_\infty:=\sup_{\theta\in\Theta_\circ}\|G_{\bullet,1}(\cdot;\theta)-G_{\bullet,0}(\cdot;\theta)\|_\infty,
\]
and similarly for \(B_1^{\rm comp},B_2^{\rm comp},B_\infty^{\rm comp}\).
\end{lemma}

\begin{lemma}[Per-flip componentwise margin]\label{lem:hawkes-margin}
Assume $g_{k\ell}\ge 0$ and $\inf_{\tau\in D}\mu_k(\tau;\theta)\ge \underline\mu_k>0$ for $k\in\{0,1\}$. For any single flip and $u\in[0,1]$, almost everywhere on $D$,
\[
\widehat\lambda^r(\tau)+u \Delta\widehat\lambda^r(\tau)\ \ge\ \mu_0(\tau;\theta)+\mu_1(\tau;\theta)\ \ge\ \underline\mu_0+\underline\mu_1=:\underline\mu_\Sigma,
\]
and
\[
\widehat\lambda_k^r(\tau)+u \Delta\widehat\lambda_k^r(\tau)\ \ge\ \mu_k(\tau;\theta)\ \ge\ \underline\mu_k\qquad (k=0,1).
\]
\end{lemma}

\begin{lemma}[Window envelope for Hawkes]\label{lem:hawkes-window}
Let $L=h$ and $R$ as in the kernel support. On the event $\Omega_{n}(h,R)$ and under \eqref{eq:Hpre},
\[
\sup_{\tau\in D}\ (\lambda_0^\star+\lambda_1^\star)(\tau)\ \le\ M_\Sigma^{\rm base}+G_{\max}n,\qquad
\sup_{\tau\in D}\ (\widehat\lambda_0^r+\widehat\lambda_1^r)(\tau;\theta)\ \le\ M_\Sigma^{\rm base}+G_{\max}n
\]
for all $\theta\in\Theta_\circ$ and all $r\in\mathcal R_{\rm path}^{\rm flip}$. Hence \cref{ass:G2prime} holds with $C_{\rm base}=M_\Sigma^{\rm base}$ and $C_{\rm loc}=G_{\max}$.
\end{lemma}

\begin{proposition}[Uniform window occupancy tail on the observation window]\label{prop:hawkes-tail}
Consider the truncated kernel Hawkes model of \cref{app:hawkes} on the observation window
\[
D := (t^\ast,T]\times\mathcal{S} \subset \bbR\times\bbR^d,
\]
with $K=2$ marks. Assume that the interaction kernels $g_{k\ell}(t,x)\ge 0$ are supported in $(0,h]\times B_R(0)$; the baseline intensities are uniformly bounded:
    \(
      \sup_{k,t,x}\mu_k(t,x) < \infty,
    \)
and that Assumption~\ref{ass:H-stab} holds, i.e. the reproduction matrix $A$ satisfies $\xi(A)<1$.

Let $|D|:=(T-t^\ast)\,\mu(\mathcal{S})$ denote the spacetime volume of $D$, and for $n\in\bbN$ define

\[
  \Omega_n(h,R)
  := \Bigl\{
      \sup_{(t,x)\in D}
      N\bigl((t-h,t)\times B_R(x)\bigr)\le n
    \Bigr\}.
\]
For technical convenience, also define the closed-window variant
\[
  \widetilde\Omega_n(h,R)
  := \Bigl\{
      \sup_{(t,x)\in D}
      N\bigl([t-h,t]\times B_R(x)\bigr)\le n
    \Bigr\}.
\]
Since \(N((t-h,t)\times B_R(x))\le N([t-h,t]\times B_R(x))\), we have \(\widetilde\Omega_n(h,R)\subseteq \Omega_n(h,R)\) and hence \(\bbP(\Omega_n(h,R)^c)\le \bbP(\widetilde\Omega_n(h,R)^c)\).

Then for every $\eta>0$ there exist constants $c_1(\eta),c_2(\eta)>0$ (depending only on the model parameters, on $h,R$ and on $\eta$) such that,
with
\[
  n_*(|D|) := \bigl\lceil c_1(\eta)\,\log|D| \bigr\rceil,
\]
we have
\[
  \bbP\bigl(\Omega_{n_*}(h,R)^c\bigr)
  \;\le\; c_2(\eta)\,|D|^{-\eta}.
\]
\end{proposition}

\begin{proof}
Throughout the proof, constants may change from line to line but depend only on the Hawkes-process parameters, on $h,R$, and on the geometry of $\mathcal S$, not on $D$. Let $N_k$ be the marked counting measure for mark $k\in\{0,1\}$, and let
\[
  N := N_0 + N_1
\]
be the total (unmarked) counting measure, so that for any Borel set $B\subset D$ we have $N(B)=N_0(B)+N_1(B)$.

For each $k\in\{0,1\}$, define the spatially aggregated (ground) process
\[
  \widetilde N_k(I)
  := N_k(I\times\mathcal{S}), \qquad I\subset\bbR\ \text{Borel},
\]
and set
\[
  \widetilde N(I) := \widetilde N_0(I)+\widetilde N_1(I).
\]
Clearly, for any $(t,x)\in D$,
\begin{equation}\label{eq:spatial-majorisation-app}
  N\bigl([t-h,t]\times B_R(x)\bigr)
  \;\le\; \widetilde N\bigl([t-h,t]\bigr).
\end{equation}

Let $\lambda_k(t,x)$ denote the intensity of $N_k$ at $(t,x)$, and let $\widetilde\lambda_k(t)$ be the intensity of $\widetilde N_k$:
\[
  \widetilde\lambda_k(t)
    = \int_{\mathcal{S}} \lambda_k(t,x)\,\dd\mu(x).
\]
By the model definition and Fubini,
\[
  \widetilde\lambda_k(t)
  = \int_{\mathcal{S}} \mu_k(t,x)\,\dd\mu(x)
    + \sum_{\ell=0}^1 \int_{(t-h,t)}\!\!\int_{\mathcal{S}}\!\int_{\mathcal{S}}
          g_{k\ell}(t-s,x-y)\,N_\ell(\dd s,\dd y)\,\dd\mu(x).
\]

For each pair $(k,\ell)$ and lag $u>0$, define the envelope kernel
\[
  \varphi^{\mathrm{env}}_{k\ell}(u)
  := \int_{\bbR^d} g_{k\ell}(u,z)\,\dd\mu(z),
\]
and note that for any $y\in\mathcal{S}$,
\[
  \int_{\mathcal{S}} g_{k\ell}(u,x-y)\,\dd\mu(x)
  \;\le\; \int_{\bbR^d} g_{k\ell}(u,z)\,\dd\mu(z)
  \;=\; \varphi^{\mathrm{env}}_{k\ell}(u).
\]
Thus, using \(\mu(\mathcal{S})<\infty\), set
\[
  \overline\mu_k^{\rm agg}
  := \mu(\mathcal{S})\,\sup_{t\in[t^\ast,T],\,x\in\mathcal{S}} \mu_k(t,x)
  \ <\ \infty.
\]
Then
\[
  \widetilde\lambda_k(t)
  \;\le\; \overline\mu_k^{\rm agg}
         + \sum_{\ell=0}^1 \int_{(t-h,t)}
                 \varphi^{\mathrm{env}}_{k\ell}(t-s)\,
                 \widetilde N_\ell(\dd s).
\]

Define the purely temporal \(2\)-dimensional linear Hawkes process \(\overline N=(\overline N_0,\overline N_1)\) on \(\bbR\) with baselines \(\overline\mu_k^{\rm agg}\) and interaction kernels \(\varphi^{\mathrm{env}}_{k\ell}\), i.e.

\[
  \overline\lambda_k(t)
  = \overline\mu_k^{\rm agg}
    + \sum_{\ell=0}^1 \int_{(t-h,t)}
           \varphi^{\mathrm{env}}_{k\ell}(t-s)\,\overline N_\ell(\dd s).
\]
Its reproduction matrix is exactly $A=(A_{k\ell})$ as in Assumption~\ref{ass:H-stab}:
\[
  A_{k\ell}
  = \int_0^h \varphi^{\mathrm{env}}_{k\ell}(u)\,\dd u
  = \int_0^h\int_{\bbR^d} g_{k\ell}(u,z)\,\dd\mu(z)\,\dd u.
\]

Let $(\Pi_k)_{k\in\{0,1\}}$ be independent Poisson random measures on $[t^\ast,\infty)\times\mathbb{R}_+$ with intensity $dt\,dz$. We construct $\overline N$ by the standard Poisson embedding for (linear) Hawkes processes with nonnegative kernels; see, e.g., \cite[Proposition~14.7.I]{daley2007introduction}, \citet[Section~3]{bremaud1996stability}, and \citet{bremaud2002rate}:
\[
  \overline N_k(dt)
  :=\int_{\mathbb{R}_+}\mathbf 1_{\{z\le \overline\lambda_k(t-)\}}\,
  \Pi_k(dt,dz),
  \qquad
  \overline\lambda_k(t)
  =\overline\mu_k^{\rm agg}+\sum_\ell\int \varphi^{\rm env}_{k\ell}(t-s)\,\overline N_\ell(ds).
\]
On the same Poisson basis, define $\widetilde N$ analogously by
\[
  \widetilde N_k(dt)
  :=\int_{\mathbb{R}_+}\mathbf 1_{\{z\le \widetilde\lambda_k(t-)\}}\,
  \Pi_k(dt,dz),
\]
which is the coupled construction used in \cite[Lemma~3]{bremaud1996stability}.

Define the first-violation time
\[
  \tau:=\inf\Bigl\{t\ge t^\ast:\exists k \text{ with }\widetilde N_k([t^\ast,t])>\overline N_k([t^\ast,t])\Bigr\}.
\]
On $\{\tau>t\}$ we have $\widetilde N_\ell|_{[t^\ast,t)}\le \overline N_\ell|_{[t^\ast,t)}$ for all $\ell$, and hence, by the nonnegativity of $\varphi^{\rm env}$ (monotonicity of the linear Hawkes  process map),
\[
  \widetilde\lambda_k(t-)
  \le \overline\mu_k^{\rm agg}+\sum_\ell \int \varphi^{\rm env}_{k\ell}(t-s)\,\widetilde N_\ell(ds)
  \le \overline\mu_k+\sum_\ell\int \varphi^{\rm env}_{k\ell}(t-s)\,\overline N_\ell(ds)
  =\overline\lambda_k(t-).
\]
Therefore the acceptance region for $\widetilde N_k$ is contained in that of $\overline N_k$ at time $t$ (cf.\ the monotone coupling argument in \citet{bremaud1996stability}), so $\tau=\infty$ almost surely and $\widetilde N\le \overline N$ pathwise. In particular,
\begin{equation}\label{eq:env-domination-app}
  \widetilde N_k(J)\le \overline N_k(J)
  \quad\text{for all Borel intervals }J\subset[t^\ast,\infty),\ \ k\in\{0,1\}.
\end{equation}
Intuitively, $\widetilde N$ is obtained by thinning the envelope $\overline N$ via the common Poisson embedding \cite[Proposition~14.7.I]{daley2007introduction}.

Because the kernels are nonnegative and $\xi(A)<1$, the process $\overline N$ admits a unique stationary version $\overline N^{\mathrm{stat}}$ on $\bbR$ \citep{bremaud1996stability}. Using the standard Poisson cluster representation, we may construct the empty-history envelope $\overline N$ on $[t^\ast,\infty)$ by removing all clusters whose ancestral point lies before $t^\ast$ from $\overline N^{\mathrm{stat}}$; hence, for any interval $J\subset[t^\ast,\infty)$,
\begin{equation}\label{eq:stationary-domination-app}
  \overline N_k(J) \;\le\; \overline N^{\mathrm{stat}}_k(J)
  \quad\text{a.s.}
\end{equation}

Combining \eqref{eq:spatial-majorisation-app},
\eqref{eq:env-domination-app} and \eqref{eq:stationary-domination-app}, we obtain
the pathwise bound
\begin{equation}\label{eq:window-domination-app}
  N\bigl([t-h,t]\times B_{R'}(x)\bigr)
  \;\le\; \overline N^{\mathrm{stat}}\bigl([t-h,t]\bigr)
  := \sum_{k=0}^1 \overline N^{\mathrm{stat}}_k\bigl([t-h,t]\bigr)
\end{equation}
for every $t\in[t^\ast,T]$, every $x\in\mathcal{S}$, and for every radius $R'\ge 0$
(since the right-hand side does not depend on $R'$).

We next obtain an exponential tail bound for a single temporal window. The process $\overline N^{\mathrm{stat}}$ is a standard multivariate linear Hawkes process with nonnegative, compactly supported kernels and $\xi(A)<1$.  Therefore, we can apply \cite[Proposition~2]{hansen2015lasso} to $\overline N^{\mathrm{stat}}$ on the interval $[-h,0)$: there exists $\theta>0$ such that
\[
  \bbE\bigl[\exp\bigl( \theta\,\overline N^{\mathrm{stat}}([-h,0))\bigr)\bigr]
  < \infty.
\]
By stationarity, for any interval $I$ of length $h$ we have
\(
  \overline N^{\mathrm{stat}}(I)
  \stackrel{d}{=} \overline N^{\mathrm{stat}}([-h,0)).
\)
Fix such an interval $I$ and set
\[
  M := \bbE\bigl[\exp\bigl( \theta\,\overline N^{\mathrm{stat}}(I)\bigr)\bigr]
  < \infty.
\]
A Chernoff bound then gives
\begin{equation}\label{eq:single-window-tail-app}
  \bbP\bigl( \overline N^{\mathrm{stat}}(I) > n \bigr)
  \;\le\; M e^{-\theta n}
  \qquad\text{for all }n\ge 0.
\end{equation}
The constants $M$ and $\theta$ depend only on the Hawkes process parameters and on
$h$, not on the particular choice of $I$.

It remains to discretize $D$ and apply a union bound. Let $\Delta_t:=h/2$ and define
\[
  \mathsf{T}
  := \bigl\{t_i:=t^\ast+i\Delta_t:\ i=0,1,\dots,I_T\bigr\},
  \quad
  I_T:=\Bigl\lceil\frac{T-t^\ast}{\Delta_t}\Bigr\rceil.
\]
Then
\(
  |\mathsf{T}|
  \le 2\,(T-t^\ast)/h + 2
  \le C_t (T-t^\ast)/h
\)
for some absolute constant $C_t$.

Because $\mathcal{S}\subset\bbR^d$ is bounded, there exists a finite set $\mathcal X=\{x_1,\dots,x_J\}\subset\mathcal{S}$ such that
\[
  \mathcal{S} \subset \bigcup_{j=1}^J B_{R/2}(x_j).
\]
We keep $J$ implicit; it depends only on $\mathcal{S}$, $R$ and the dimension $d$, and in particular does not depend on the observation window $D$.

Set
\[
  \mathcal G := \mathsf{T}\times \mathcal X,\qquad
  |\mathcal G| = |\mathsf{T}|\,|\mathcal X|
  \;\le\; C_t J\,\frac{T-t^\ast}{h}
  \;=\; C_{\rm grid}\,|D|
\]
for some constant $C_{\rm grid}>0$ depending only on $h,R$ and the geometry of $\mathcal{S}$ (we used $|D|=(T-t^\ast)\,\mu(\mathcal{S})$ here).

For the covering/union-bound argument we work with enlarged spatial balls \(B_{2R}(y)\), so that any target ball \(B_R(x)\) is contained in some \(B_{2R}(y)\) from the spatial grid. For each $(s,y)\in\mathcal G$, consider the spacetime window
\[
  W_{s,y}:=[s-h,s]\times B_{2R}(y).
\]
By \eqref{eq:window-domination-app} (with radius $2R$) and \eqref{eq:single-window-tail-app}, we have
\[
  \bbP\bigl(N(W_{s,y})>n\bigr)
  \;\le\;
  \bbP\bigl(\overline N^{\mathrm{stat}}([s-h,s])>n\bigr)
  \;\le\; M e^{-\theta n}
\]
for every $(s,y)\in\mathcal G$ and all $n\ge 0$.  A union bound over the finite grid $\mathcal G$ yields
\begin{equation}\label{eq:grid-union-app}
  \bbP\Bigl(
    \max_{(s,y)\in\mathcal G} N(W_{s,y})>n
  \Bigr)
  \;\le\; |\mathcal G|\,M e^{-\theta n}
  \;\le\; C_{\rm grid} M\,|D|\,e^{-\theta n}.
\end{equation}

For any $(t,x)\in D$, choose $s\in\mathsf{T}$ and $y\in\mathcal X$ such that
\[
  |t-s|\le \Delta_t = h/2,\qquad \|x-y\|\le R/2.
\]
Then $B_R(x)\subset B_{3R/2}(y)\subset B_{2R}(y)$ and
\[
  [t-h,t]
  \subset [s-2h,s+h].
\]
The interval $[s-2h,s+h]$ is contained in the union of three intervals of length $h$, for example $[s-2h,s-h]$, $[s-h,s]$, and $[s,s+h]$.  Hence
\[
  [t-h,t]\times B_R(x)
  \subset \bigcup_{\ell=1}^3
    \bigl([s_\ell-h,s_\ell]\times B_{2R}(y)\bigr),
\]
for suitable $s_\ell\in\mathsf{T}'$ in a time grid $\mathsf{T}'$ obtained by shifting $\mathsf{T}$ by multiples of $h$.  Enlarging $\mathcal G$ if necessary to
\(
  \mathcal G':=\mathsf{T}'\times\mathcal X,
\)
we obtain
\[
  N\bigl([t-h,t]\times B_R(x)\bigr)
  \;\le\; 3\max_{(s,y)\in\mathcal G'} N\bigl([s-h,s]\times B_{2R}(y)\bigr).
\]
The cardinality of $\mathcal G'$ is still bounded by a constant multiple of $|D|$: $|\mathcal G'|\le C'_{\rm grid}|D|$ for some $C'_{\rm grid}>0$.

Taking a supremum over $(t,x)\in D$ yields
\[
  \sup_{(t,x)\in D} N\bigl([t-h,t]\times B_R(x)\bigr)
  \;\le\; 3\max_{(s,y)\in\mathcal G'} N\bigl([s-h,s]\times B_{2R}(y)\bigr).
\]
Therefore,
\[
  \bbP\!\left(
    \sup_{(t,x)\in D} N\bigl([t-h,t]\times B_R(x)\bigr) > 3n
  \right)
  \;\le\;
  \bbP\!\left(
    \max_{(s,y)\in\mathcal G'} N\bigl([s-h,s]\times B_{2R}(y)\bigr) > n
  \right).
\]
Applying \eqref{eq:grid-union-app} (with $\mathcal G'$ and possibly a different grid constant) gives
\[
  \bbP\!\left(
    \sup_{(t,x)\in D} N\bigl([t-h,t]\times B_R(x)\bigr) > 3n
  \right)
  \;\le\; C''_{\rm grid} M\,|D|\,e^{-\theta n}
\]
for some $C''_{\rm grid}>0$.

Finally, let $\eta>0$ be arbitrary and choose
\[
  n = n_0(|D|)
  := \Bigl\lceil \frac{\eta+1}{\theta}\,\log|D| \Bigr\rceil.
\]
Then
\[
  C''_{\rm grid}M\,|D|\,e^{-\theta n}
  \;\le\; C''_{\rm grid}M\,|D|\,e^{-(\eta+1)\log|D|}
  \;=\; C''_{\rm grid}M\,|D|^{-\eta}.
\]
Absorbing the multiplicative constant $3$ in front of $n$ into the definition of $c_1(\eta)$ (i.e.\ replacing $n_0(|D|)$ by $3n_0(|D|)$) and renaming constants
\[
  c_1(\eta):=\frac{3(\eta+1)}{\theta},
  \qquad
  c_2(\eta):=C''_{\rm grid}M,
\]
we obtain
\[
  n_*(|D|)
  := \bigl\lceil c_1(\eta)\,\log|D| \bigr\rceil
\]
Thus we have shown \(\bbP(\widetilde\Omega_{n_*}(h,R)^c)\le c_2(\eta)\,|D|^{-\eta}\). Since \(\Omega_{n_*}(h,R)^c\subseteq \widetilde\Omega_{n_*}(h,R)^c\), the same bound holds for \(\bbP(\Omega_{n_*}(h,R)^c)\).
\end{proof}

\begin{remark}[Local strong concavity for multivariate Hawkes]\label{rem:hawkes-local-sc}
For multivariate Hawkes processes, uniform local strong concavity of the (complete-data) log-likelihood around the truth, including for parameterizations that contain nonlinear kernel-shape parameters, is established in \citet{DAVISKRESIN_HPB}. The proof follows the local curvature strategy used in \citet{hansen2015lasso} (control of the empirical Hessian/Gram term on a fixed-radius neighborhood), and in the present spatiotemporal truncated setting the same conclusion can be obtained by combining those local curvature arguments with the same monotone coupling / window-envelope control used in our Hawkes process verification (e.g.\ the coupling behind Proposition~\ref{prop:hawkes-tail}) to transfer the needed uniform concentration bounds to the spatiotemporal observation window, verifying Assumption~\ref{ass:G6}.
\end{remark}

\subsubsection{Verification of Assumptions~\ref{ass:G4}, \ref{ass:G5} and \ref{ass:G7prime}}\label{app:hawkes-align}


\begin{assumption}[Column similarity (sufficient for uniform alignment)]
\label{ass:col-sim}
There exist deterministic numbers $\varepsilon_0,\varepsilon_1\ge 0$ such that, uniformly over $\theta\in\Theta_\circ$,
\[
\sup_{(t,x)\in(0,h]\times \mathcal B_R}
\big|g_{k0}(t,x;\theta)-g_{k1}(t,x;\theta)\big|
\ \le\ \varepsilon_k,
\qquad k\in\{0,1\}.
\]
Set $\varepsilon_{\rm col}:=\varepsilon_0+\varepsilon_1$.
\end{assumption}

\begin{remark}[Checking the column-similarity constants $\varepsilon_k$]\label{rem:epsk-check}
For the truncated exponential parametrization
\[
g_{k\ell}(t,x;\theta)
=\alpha_{k\ell}e^{-\beta t}w_{k\ell}(x)\,\mathbf 1_{(0,h]}(t)\,\mathbf 1_{\mathcal B_R}(x),
\qquad \mathcal B_R:=B_R(0),
\]
a convenient choice in Assumption~\ref{ass:col-sim} is
\[
\varepsilon_k
:=\sup_{\theta\in\Theta_\circ}\ \sup_{(t,x)\in(0,h]\times\mathcal B_R}
\big|g_{k0}(t,x;\theta)-g_{k1}(t,x;\theta)\big|.
\]
Since $e^{-\beta t}\le 1$ on $(0,h]$, we have
\begin{align*}
\varepsilon_k
&\le \sup_{\theta\in\Theta_\circ}\ \sup_{x\in\mathcal B_R}
\big|\alpha_{k0}w_{k0}(x)-\alpha_{k1}w_{k1}(x)\big| \\
&\le \sup_{\theta\in\Theta_\circ}
\Big(
|\alpha_{k0}-\alpha_{k1}|\,\|w_{k0}\|_\infty
+\alpha_{\max}\,\|w_{k0}-w_{k1}\|_\infty
\Big),
\end{align*}
using $0\le \alpha_{k\ell}\le \alpha_{\max}$ on $\Theta_\circ$ and the triangle inequality.
Thus Assumption~\ref{ass:col-sim} reduces to a simple ``parameter closeness'' condition on
$(\alpha_{k0},\alpha_{k1},w_{k0},w_{k1})$.
\end{remark}

\begin{proposition}[Uniform score alignment under column similarity]
\label{prop:hawkes-G4}
Consider the truncated Hawkes model of \cref{app:hawkes} with nonnegative kernels supported on
$(0,h]\times \mathcal B_R$.
Assume \cref{ass:G1,ass:oracle-positivity,ass:col-sim}.
Work either under the empty-history condition \eqref{eq:Hpre} or, more generally, conditional on the
marked history of $(N_0,N_1)$ on $(t^\ast-h,t^\ast]\times\mathcal S$ (so that the marked past entering the Hawkes functional is fixed under the conditioning).
Work on the window-count event $\Omega_{n_*}(h,R)$ from \cref{ass:G2prime}.
Then for every $\theta\in\Theta_\circ$, every deterministic labelling $r$, and for a.e.\ $\tau\in D$,
\[
\big|\tilde s_\theta(\tau;r)-s^{\rm or}_\theta(\tau)\big|
\ \le\ \Delta_s,
\qquad
\Delta_s := \frac{n_*(|D|)}{\underline\mu_{\min}}\varepsilon_{\rm col}.
\]
In particular, Assumption~\ref{ass:G4} holds with
\[
E^{\rm align}_{|D|}:=\Omega_{n_*}(h,R),
\qquad
\Delta_s(|D|):=\frac{n_*(|D|)}{\underline\mu_{\min}}\varepsilon_{\rm col}.
\]
Moreover, since
\[
\mathbb P\big((E^{\rm align}_{|D|})^c\big)\le c_2|D|^{-\eta_{\rm win}},
\]
after possibly decreasing the exponent and enlarging the large-\(|D|\) threshold, this matches the probability requirement in Assumption~\ref{ass:G4}.
\end{proposition}

\begin{proof}
Fix $\theta\in\Theta_\circ$ and a deterministic labelling $r$.
Let $r^\star$ denote the true labelling on $D$.
Under the Hawkes model specification,
\begin{align*}
\widehat\lambda_k^r(\tau;\theta)
=&
\mu_k(\tau;\theta)+\sum_{\ell=0}^1\int g_{k\ell}(\tau-\tau';\theta)\,d\widehat N_\ell^r(\tau'),\\
\lambda_k(\tau\mid\mathcal G_{t-};\theta)
=&
\mu_k(\tau;\theta)+\sum_{\ell=0}^1\int g_{k\ell}(\tau-\tau';\theta)\,d N_\ell(\tau').
\end{align*}

Under the stated conditioning (empty history or marked pre-$t^\ast$ history fixed),
the marked past entering the Hawkes functional is the same in $\widehat\lambda^{r^\star}$ and in the oracle intensity,
and moreover $\widehat N_\ell^{r^\star}=N_\ell$ on $D$.
Therefore $\widehat\lambda_k^{r^\star}(\tau;\theta)=\lambda_k(\tau\mid\mathcal G_{t-};\theta)$ for a.e.\ $\tau\in D$, hence
$\tilde s_\theta(\tau;r^\star)=s^{\rm or}_\theta(\tau)$.

Fix $\tau=(t,x)\in D$ and $k\in\{0,1\}$. Because $g_{k\ell}$ is supported on $(0,h]\times \mathcal B_R$,
only events in the truncated neighbourhood $(t-h,t)\times B_R(x)$ contribute. On $\Omega_{n_*}(h,R)$ there are at most
$n_*(|D|)$ such events.
For any event $\tau_i$ in this neighbourhood, if $r_i=r_i^\star$ then its contribution is the same under $r$ and $r^\star$;
if $r_i\neq r_i^\star$, then its contribution changes from the column $r_i^\star$ kernel to the column $r_i$ kernel, and
\cref{ass:col-sim} yields a pointwise bound by $\varepsilon_k$.
Therefore,
\[
\big|\widehat\lambda_k^r(\tau;\theta)-\widehat\lambda_k^{r^\star}(\tau;\theta)\big|
\ \le\ n_*(|D|)\,\varepsilon_k.
\]
Using $\widehat\lambda_k^r(\tau;\theta)\ge \underline\mu_k$ and $\widehat\lambda_k^{r^\star}(\tau;\theta)\ge \underline\mu_k$
from \cref{ass:G1}, the mean-value bound $|\log a-\log b|\le |a-b|/\underline\mu_k$ gives
\[
\big|\log\widehat\lambda_k^r(\tau;\theta)-\log\widehat\lambda_k^{r^\star}(\tau;\theta)\big|
\ \le\ \frac{n_*(|D|)\varepsilon_k}{\underline\mu_k}
\ \le\ \frac{n_*(|D|)\varepsilon_k}{\underline\mu_{\min}}.
\]
Finally,
\[
\big|\tilde s_\theta(\tau;r)-s^{\rm or}_\theta(\tau)\big|
=
\Big|\big(\log\widehat\lambda_1^r-\log\widehat\lambda_0^r\big)
-\big(\log\widehat\lambda_1^{r^\star}-\log\widehat\lambda_0^{r^\star}\big)\Big|
\le
\sum_{k=0}^1\big|\log\widehat\lambda_k^r-\log\widehat\lambda_k^{r^\star}\big|,
\]
so the claim follows by summing the previous display over $k=0,1$.
\end{proof}

\begin{assumption}[CSC and baseline comparability]\label{ass:CSC}
There exist weights $u_0,u_1>0$ with $u_0+u_1=1$, and constants $0<a_{\rm csc}\le b_{\rm csc}<\infty$ such that, for each column $\ell$ and a.e.\ point $(t,x)$ with $G_{\bullet\ell}(t,x)>0$,
\[
a_{\rm csc} u_k\ \le\ \frac{g_{k\ell}(t,x)}{G_{\bullet\ell}(t,x)}\ \le\ b_{\rm csc} u_k,\qquad k\in\{0,1\}.
\]
In addition, there exist $0<c_\mu\le C_\mu<\infty$ such that
\[
c_\mu\ \le\ \frac{\mu_1(\tau;\theta)}{\mu_0(\tau;\theta)}\ \le\ C_\mu\qquad \text{for all }\tau\in D,\ \theta\in\Theta_\circ.
\]
All constants are uniform over $\theta\in\Theta_\circ$.
\end{assumption}

\begin{lemma}[Lipschitz control for the oracle score (Assumption~\ref{ass:G5})]
\label{lem:hawkes-Lip-avg}
Let $N=(N_0,N_1)$ be a bivariate spatio-temporal Hawkes process on $\mathbb{R}\times\mathbb{R}^d$.
We observe it on $D=(t^\ast,T]\times \mathcal{S}$, where $\mathcal{S}\subset\mathbb{R}^d$ is bounded and $0<\mu(\mathcal{S})<\infty$.
Write $\tau=(t,x)$ and $d\tau:=dt\otimes d\mu(x)$.
For $k\in\{0,1\}$, the $\mathcal{G}$-predictable conditional intensities are
\[
\lambda_k(\tau;\theta)
=
\mu_k(\tau;\theta)
+
\sum_{\ell=0}^1\int_{(t-h,t)}\int_{B_R(x)}
g_{k\ell}(t-s,x-y;\theta)\,N_\ell(ds,dy),
\]
with compactly supported bounded kernels. Denote by $\lambda^\star(\tau):=\lambda_0(\tau;\theta^\star)+\lambda_1(\tau;\theta^\star)$ the true total
intensity and let
\[
s^{\rm or}_\theta(\tau):=\log\frac{\lambda_1(\tau;\theta)}{\lambda_0(\tau;\theta)}
\]
be the oracle log-likelihood ratio.

Assume:

\smallskip
\noindent{\bf (A1) Uniform positivity and relative comparability.}
There exist constants \(0<\nu_-\le \nu_+<\infty\) and \(C_\mu^{\rm rel},C_g<\infty\) such that, for all
\(\theta\in\Theta_\circ\), all \(k,\ell\), and a.e. \((u,z)\in(0,h]\times B_R(0)\),
\[
\nu_-\le \mu_k(\tau;\theta)\le \nu_+,\qquad
\mu_k(\tau;\theta^\star)\le C_\mu^{\rm rel}\mu_k(\tau;\theta),\qquad
g_{k\ell}(u,z;\theta^\star)\le C_g\,g_{k\ell}(u,z;\theta).
\]

\smallskip
\noindent{\bf (A2) Uniform component-ratio bounds on $\Theta_\circ$.}
There exist constants $0<c_{\rm rat}\le C_{\rm rat}<\infty$ such that for all $\theta\in\Theta_\circ$ and
$\tau\in D$,
\[
c_{\rm rat}\ \le\ \frac{\lambda_1(\tau;\theta)}{\lambda_0(\tau;\theta)}\ \le\ C_{\rm rat}.
\]
(For example, this holds under CSC + baseline comparability.)

\smallskip
\noindent{\bf (A3) Activity bounds on schedule enlargements.}
Let the (deterministic) schedule windows be $\{D^{(m)}\}_{m=1}^{\Tmax}$ and $\{D_m\}_{m=1}^{\Tmax}$ as in \S\ref{sec:setup}.
For any window $W=(a,b]\times\mathcal S$ in
\[
\mathcal W_{\rm sch}:=\{D\}\cup\{D^{(m)}:\ 1\le m\le \Tmax\}\cup\{D_m:\ 1\le m\le \Tmax\},
\]
define its enlargement
\[
W^+:=(a-h,b]\times\mathcal S^{+R},
\qquad
\mathcal S^{+R}:=\{y\in\bbR^d:\mathrm{dist}(y,\mathcal S)\le R\}.
\]
Assume there exist constants $K_{\rm act},\eta>0$ such that for all large $|D|$,
\[
\Pr\!\Big(\bigcap_{W\in\mathcal W_{\rm sch}}\{\sum_{\ell=0}^1 N_\ell(W^+)\le K_{\rm act}|W^+|\}\Big)\ \ge\ 1-|D|^{-\eta}.\label{eq:hawkes-activity}
\]
Define the event
\[
E^{\rm Lip}_{|D|}
:=
\bigcap_{W\in\mathcal W_{\rm sch}}\left\{\sum_{\ell=0}^1 N_\ell(W^+)\le K_{\rm act}|W^+|\right\}.
\]

\smallskip
Then Assumption~\ref{ass:G5} holds: there exists a deterministic quantity $L_s(|D|)<\infty$ such that, on
$E^{\rm Lip}_{|D|}$, for all $\theta,\theta'\in\Theta_\circ$,
\begin{equation}
\label{eq:G5-hawkes}
\int_D \big|s^{\rm or}_\theta(\tau)-s^{\rm or}_{\theta'}(\tau)\big|\,\lambda^\star(\tau)\,d\tau
\le
L_s(|D|)\,|D|\,\|\theta-\theta'\|_2.
\end{equation}
The same argument applies with $D$ replaced by any $W\in\mathcal W_{\rm sch}$, yielding the schedule-uniform bounds required in Assumption~\ref{ass:G5}. Moreover, $\Pr((E^{\rm Lip}_{|D|})^c)\le |D|^{-\eta}$ by \eqref{eq:hawkes-activity}, hence Assumption~\ref{ass:G5} holds with the same exponent.

\end{lemma}

\begin{proof}
We first develop an explicit relative-comparability constant. Under (A1), for every \(\theta\in\Theta_\circ\), every $\tau$ and every $k$
\begin{align*}
\lambda_k^\star(\tau)
=&\mu_k(\tau;\theta^\star)+\sum_{\ell=0}^1\int g_{k\ell}(\tau-\tau';\theta^\star)\,\dd N_\ell(\tau')\\
\le&
C_{\rm rel}\Bigl(\mu_k(\tau;\theta)+\sum_{\ell=0}^1\int g_{k\ell}(\tau-\tau';\theta)\,\dd N_\ell(\tau')\Bigr)\\
=&
C_{\rm rel}\lambda_k(\tau;\theta),
\end{align*}
where \(C_{\rm rel}:= \max\{C_\mu^{\rm rel}, C_g\}\). In particular, \(\lambda^\star(\tau)\le C_{\rm rel}\bigl(\lambda_0(\tau;\theta)+\lambda_1(\tau;\theta)\bigr)\).

Next, we find a mean-value bound for $s^{\rm or}$ and removal of denominators. Fix $\theta,\theta'\in\Theta_\circ$. Since
$s^{\rm or}_\theta(\tau)=\log\lambda_1(\tau;\theta)-\log\lambda_0(\tau;\theta)$ and $\lambda_k(\tau;\vartheta)\ge \nu_->0$,
the mean value theorem gives, for each $\tau$,
\[
\big|s^{\rm or}_\theta(\tau)-s^{\rm or}_{\theta'}(\tau)\big|
\le
\|\theta-\theta'\|_2\,
\sup_{\vartheta\in[\theta,\theta']}
\sum_{k=0}^1\frac{\|\nabla_\theta \lambda_k(\tau;\vartheta)\|_2}{\lambda_k(\tau;\vartheta)}.
\]
Multiplying by $\lambda^\star(\tau)$ and using $\lambda^\star\le C_{\rm rel}(\lambda_0+\lambda_1)$ yields
\[
\lambda^\star(\tau)\sum_{k=0}^1\frac{\|\nabla_\theta \lambda_k(\tau;\vartheta)\|_2}{\lambda_k(\tau;\vartheta)}
\le
C_{\rm rel}\sum_{k=0}^1 \frac{\lambda_0(\tau;\vartheta)+\lambda_1(\tau;\vartheta)}{\lambda_k(\tau;\vartheta)}\,
\|\nabla_\theta \lambda_k(\tau;\vartheta)\|_2.
\]
By the ratio bounds (A2), for all $k$,
\[
\frac{\lambda_0+\lambda_1}{\lambda_k}
=1+\frac{\lambda_{1-k}}{\lambda_k}
\le 1+\max\{C_{\rm rat},1/c_{\rm rat}\}
=:C_{\rm mix}.
\]
Hence, uniformly in $\vartheta\in\Theta_\circ$,
\[
\lambda^\star(\tau)\sum_{k=0}^1\frac{\|\nabla_\theta \lambda_k(\tau;\vartheta)\|_2}{\lambda_k(\tau;\vartheta)}
\le
C_{\rm rel}C_{\rm mix}\sum_{k=0}^1 \|\nabla_\theta \lambda_k(\tau;\vartheta)\|_2.
\]
Therefore,
\begin{equation}
\label{eq:key-reduction}
\int_D \big|s^{\rm or}_\theta-s^{\rm or}_{\theta'}\big|\,\lambda^\star\,d\tau
\le
C_{\rm rel}C_{\rm mix}\,\|\theta-\theta'\|_2
\int_D \sup_{\vartheta\in\Theta_\circ}\sum_{k=0}^1\|\nabla_\theta \lambda_k(\tau;\vartheta)\|_2\,d\tau.
\end{equation}

We now obtain a linear bound on gradients in terms of local counts. Because $\Theta_\circ$ is compact and the kernel family is truncated-exponential with bounded $w_{k\ell}$,
there exist deterministic constants $C_0,C_1<\infty$ (depending only on $\Theta_\circ,h,R,w_\infty$ and the baseline
smoothness on $\Theta_\circ$) such that, for all $\vartheta\in\Theta_\circ$ and $\tau=(t,x)$,
\begin{equation}
\label{eq:grad-linear-fixed}
\sum_{k=0}^1 \|\nabla_\theta \lambda_k(t,x;\vartheta)\|_2
\le
C_0 + C_1\sum_{\ell=0}^1 N_\ell\!\big((t-h,t)\times B_R(x)\big).
\end{equation}

We proceed to develop a spatio-temporal Fubini bound for the averaged local counts. Let
\[
v_R:=\sup_{y\in\bbR^d}\mu(B_R(y))<\infty
\]
(which equals $\mu(B_R)$ when $\mu$ is Lebesgue measure). Here and below we write $D^+:=(t^\ast-h,T]\times \mathcal S^{+R}$, i.e.\ $D^+=W^+$ with $W=D$. For each $\ell\in\{0,1\}$,
\begin{align}
\int_{t^\ast}^T\int_{\mathcal{S}} N_\ell\!\big((t-h,t)\times B_R(x)\big)\,\mu(dx)\,dt
&=
\int_{(t^\ast-h,T)}\int_{\bbR^d}
\Big(\int_{t^\ast}^T \mathbf{1}\{s\in(t-h,t)\}\,dt\Big)\\ \quad&
\cdot \Big(\int_S \mathbf{1}\{y\in B_R(x)\}\,\mu(dx)\Big)\,N_\ell(ds,dy)
\nonumber\\
&\le
\int_{D^+} h\,v_R\,N_\ell(ds,dy)
=
h\,v_R\,N_\ell(D^+),
\label{eq:fubini-fixed}
\end{align}
since $\int_{t^\ast}^T \mathbf{1}\{s\in(t-h,t)\}\,dt\le h$ and
$\int_S \mathbf{1}\{y\in B_R(x)\}\,\mu(dx)=\mu(\mathcal{S}\cap B_R(y))\le v_R$,
and the integrand vanishes outside $D^+$. The same computation with $t$ integrated over $(a,b]$ (instead of $(t^\ast,T]$) yields
\[
\int_a^b\int_{\mathcal S} N_\ell((t-h,t)\times B_R(x))\,\mu(dx)\,dt
\ \le\ h v_R\,N_\ell(W^+),
\qquad W=(a,b]\times\mathcal S,
\]
so on $E^{\rm Lip}_{|D|}$ we have $N_\ell(W^+)\le K_{\rm act}|W^+|$ uniformly over $W\in\mathcal W_{\rm sch}$.

We now conclude (Assumption~\ref{ass:G5}) on $E^{\rm Lip}_{|D|}$. Combine \eqref{eq:key-reduction}, \eqref{eq:grad-linear-fixed} and \eqref{eq:fubini-fixed} to get
\[
\int_D \big|s^{\rm or}_\theta-s^{\rm or}_{\theta'}\big|\,\lambda^\star\,d\tau
\le
C_{\rm rel}C_{\rm mix}\,\|\theta-\theta'\|_2
\Big(
C_0|D| + C_1 h v_R \sum_{\ell=0}^1 N_\ell(D^+)
\Big).
\]
On $E^{\rm Lip}_{|D|}$ we have $\sum_{\ell=0}^1 N_\ell(D^+)\le K_{\rm act}|D^+|$, hence
\[
\int_D \big|s^{\rm or}_\theta-s^{\rm or}_{\theta'}\big|\,\lambda^\star\,d\tau
\le
C_{\rm rel}C_{\rm mix}\,\|\theta-\theta'\|_2
\Big(
C_0|D| + C_1 h v_R K_{\rm act}|D^+|
\Big).
\]
Define
\[
L_s(|D|)
:=C_{\rm rel}C_{\rm mix}\left(C_0 + C_1 h v_R K_{\rm act}\,\frac{|D^+|}{|D|}\right),
\]
which is deterministic given the window geometry and model constants. Then \eqref{eq:G5-hawkes} follows.
Finally, by \eqref{eq:hawkes-activity} we have $\Pr((E^{\rm Lip}_{|D|})^c)\le |D|^{-\eta}$, completing the verification of Assumption~\ref{ass:G5}.
\end{proof}


\begin{proposition}[Label-to-score Lipschitz for truncated Hawkes]
\label{prop:hawkes-G7prime}
Assume \cref{ass:G0,ass:G0pred,ass:G1,ass:G2prime,ass:G3}.
Assume the Hawkes process kernels are supported on $(0,L]\times B_R(0)$ and are nonnegative.
Define
\[
V_{L,R}:=L\,\mu(\mathcal B_R),
\qquad\text{so that}\qquad
|(t,t+L]\times B_R(x)|\le V_{L,R}\ \ \text{for all }(t,x).
\]
Redefine (if needed) $K_{\rm win}\leftarrow K_{\rm win}\vee 1$, where $K_{\rm win}$ is the deterministic envelope from
\cref{ass:G2prime}.

For an event $\gamma_j=(t_j,x_j)$ and a schedule endpoint $u\in\{u_m\}_{m\le\Tmax}$ define the forward cone on $D^{[u]}$,
\[
C_j^{[u]}:=\big((t_j,u]\cap(t_j,t_j+L]\big)\times B_R(x_j).
\]
For $z_{\rm Fr}>0$, define the cone-count event
\[
E^{\rm cone}(z_{\rm Fr}):=
\bigcap_{u\in\{u_m\}_{m\le\Tmax}}\ \bigcap_{j\le N^{[u]}}
\left\{
N(C_j^{[u]})
\ \le\
V_{L,R}\,K_{\rm win}
+\sqrt{2V_{L,R}\,K_{\rm win}\,z_{\rm Fr}}
+\frac{2}{3}z_{\rm Fr}
\right\}.
\]

Define the deterministic constants
\[
L_{\log}:=\frac{S_\infty}{\underline\mu_{\min}},\qquad
C_{\log}:=\frac{\overline B_\infty^{\rm comp}}{\underline\mu_{\min}}
+\frac{S_\infty\,B_\infty^{\rm comp}}{\underline\mu_{\min}^2},
\]
and
\[
\bar C_0:=V_{L,R}\,C_{\log}+2L_{\log}+\overline B_1^\Sigma,\qquad
\bar C_1:=\sqrt{2V_{L,R}}\,C_{\log},\qquad
\bar C_2:=\frac{2}{3}C_{\log}.
\]

Then on $\Omega_{n_*}(L,R)\cap E^{\rm cone}(z_{\rm Fr})$, for every update $m\in\{0,\dots,\Tmax-1\}$ with $u=u_{m+1}$
and $r=r^{(m+1)}$, uniformly over $\theta\in\Theta_{\rm sc}$,
\[
\left\Vert \nabla f_{r}^{[u]}(\theta) - \nabla f_{r^\star}^{[u]}(\theta)\right\Vert
\ \le\ 
\frac{d_H^{[u]}(r,r^\star)}{|D^{[u]}|}
\Big(\bar C_0 K_{\rm win} + \bar C_1\sqrt{K_{\rm win} z_{\rm Fr}} + \bar C_2 z_{\rm Fr}\Big).
\]
In particular (by uniformity in $\theta$), the same bound holds at any restricted M-step maximiser
$\hat\theta_r^{[u]}\in\arg\max_{\theta\in\Theta_{\rm sc}}\ell_r^{[u]}(\theta)$, which is exactly \cref{ass:G7prime}.

Moreover, by Lemma~\ref{lem:predictable-cone} and Lemma~\ref{lem:Freedman}, for each $z_{\rm Fr}>0$,
\[
\mathbb P\!\left(\big(E^{\rm cone}(z_{\rm Fr})\big)^c\cap\Omega_{n_*}(L,R)\cap E_{\rm act}\right)
\ \le\ (1+\eta_{\rm act})C_{\rm act}|D|\,\Tmax\,e^{-z_{\rm Fr}}.
\]
\end{proposition}

\begin{proof}
Fix $u$ and $\theta\in\Theta_{\rm sc}$. Let $(r^{(j)})_{j=0}^{\mathfrak I}$ be the canonical telescoping path from $r$ to
$r^\star$ from \cref{def:telescoping-path}, where $\mathfrak I=d_H^{[u]}(r,r^\star)$ and each adjacent pair
differs by a single flip at some index $i_j\le N^{[u]}$. By the triangle inequality,
\[
\big\|\nabla f_r^{[u]}(\theta)-\nabla f_{r^\star}^{[u]}(\theta)\big\|
\le
\sum_{j=1}^{\mathfrak I}
\big\|\nabla f_{r^{(j-1)}}^{[u]}(\theta)-\nabla f_{r^{(j)}}^{[u]}(\theta)\big\|.
\]
It therefore suffices to bound the change in $\nabla f^{[u]}$ under a single flip.

Fix a single flip at $\gamma_{i}=(t_{i},x_{i})$ and denote the two labellings before/after the flip by $r$ and $r'$.
Write $W:=D^{[u]}$ for brevity. By definition,
\[
\nabla f_r^{[u]}(\theta)
=\frac{1}{|W|}\left\{\sum_{j:\,t_j\le u}\nabla_\theta\log\widehat\lambda_{r_j}^r(\gamma_j;\theta)
-\int_W \nabla_\theta \widehat\lambda^r(\tau;\theta)\,d\tau\right\}.
\]

With respect to the event-sum part, at the flipped event $\gamma_i$, predictability implies
$\widehat\lambda_k^{r'}(\gamma_i;\theta)=\widehat\lambda_k^{r}(\gamma_i;\theta)$ for $k\in\{0,1\}$, so the only change
is the component index inside $\log(\cdot)$. Hence
\[
\left\|\nabla\log\widehat\lambda_{r'_{i}}^{r'}(\gamma_i;\theta)
-\nabla\log\widehat\lambda_{r_{i}}^{r}(\gamma_i;\theta)\right\|
\le 2\max_{k\in\{0,1\}}\|\nabla\log\widehat\lambda_k^{r}(\gamma_i;\theta)\|
\le 2L_{\log}.
\]

For any other event $\gamma_j$ with $t_j\le u$, changing the label at $\gamma_{i}$ can only affect
$\widehat\lambda_\cdot(\gamma_j;\theta)$ if $\gamma_j\in C_{i}^{[u]}$ (kernel support $(0,L]\times B_R$).
On $C_{i}^{[u]}$, for each $k$,
\begin{align*}
\big\|\nabla\log\widehat\lambda_k^{r'}(\tau;\theta)-\nabla\log\widehat\lambda_k^{r}(\tau;\theta)\big\|
&=
\left\|\frac{\nabla\widehat\lambda_k^{r'}(\tau;\theta)}{\widehat\lambda_k^{r'}(\tau;\theta)}
-\frac{\nabla\widehat\lambda_k^{r}(\tau;\theta)}{\widehat\lambda_k^{r}(\tau;\theta)}\right\|\\
&\le
\frac{\|\Delta(\nabla \widehat\lambda_k^r)(\tau;\theta)\|}{\underline\mu_{\min}}
+\frac{\|\nabla \widehat\lambda_k^{r}(\tau;\theta)\|}{\underline\mu_{\min}^2}\,|\Delta\widehat\lambda_k^r(\tau;\theta)|
\ \le\ C_{\log},
\end{align*}
by \cref{ass:G1,ass:G3}. Therefore the event-sum contribution is bounded by
\[
2L_{\log}+C_{\log}\,N(C_{i}^{[u]}).
\]

With respect to the compensator part, by the $L^1$ gradient influence bound in \cref{ass:G3},
\[
\left\|\int_W \nabla\widehat\lambda^{r'}(\tau;\theta)\,d\tau - \int_W \nabla\widehat\lambda^{r}(\tau;\theta)\,d\tau\right\|
\le \int_W \|\Delta(\nabla\widehat\lambda^{r})(\tau;\theta)\|\,d\tau
\le \overline B_1^\Sigma.
\]

Combining (i) and (ii),
\[
\big\|\nabla f_{r'}^{[u]}(\theta)-\nabla f_{r}^{[u]}(\theta)\big\|
\le
\frac{1}{|W|}\Big(2L_{\log}+C_{\log}N(C_{i}^{[u]})+\overline B_1^\Sigma\Big).
\]
On $E^{\rm cone}(z_{\rm Fr})$, substitute the bound on $N(C_i^{[u]})$ and use $K_{\rm win}\ge 1$ to absorb
$2L_{\log}+\overline B_1^\Sigma$ into the $K_{\rm win}$ term, giving
\[
\big\|\nabla f_{r'}^{[u]}(\theta)-\nabla f_{r}^{[u]}(\theta)\big\|
\le
\frac{1}{|W|}\Big(\bar C_0 K_{\rm win} + \bar C_1\sqrt{K_{\rm win} z_{\rm Fr}} + \bar C_2 z_{\rm Fr}\Big).
\]

We now sum over the telescoping path. Summing the single-flip bound over $\mathfrak I=d_H^{[u]}(r,r^\star)$ flips yields the stated inequality for
$\|\nabla f_r^{[u]}(\theta)-\nabla f_{r^\star}^{[u]}(\theta)\|$.
Uniformity in $\theta$ allows evaluation at $\theta=\hat\theta_r^{[u]}$.

Fix $(u,j)$ and consider $C_j^{[u]}$. By Lemma~\ref{lem:predictable-cone},
$\mathbf 1_{C_j^{[u]}}$ is $\mathcal G$-predictable, so Freedman inequality applies to
\(
N(C_j^{[u]})-\int_{C_j^{[u]}}\lambda^\star(\tau)\,d\tau.
\)
On $\Omega_{n_*}(L,R)$ we have $\|\lambda^\star\|_\infty\le K_{\rm win}$, hence
$\int_{C_j^{[u]}}\lambda^\star\le K_{\rm win}|C_j^{[u]}|\le V_{L,R}K_{\rm win}$.
Therefore
\[
\mathbb P\!\left(
N(C_j^{[u]})>V_{L,R}K_{\rm win}+\sqrt{2V_{L,R}K_{\rm win}z_{\rm Fr}}+\frac23 z_{\rm Fr},\ \Omega_{n_*}(L,R)
\right)\le e^{-z_{\rm Fr}}.
\]
On $E_{\rm act}$, the number of cones in the defining intersection of $E^{\rm cone}(z_{\rm Fr})$ is at most
$\sum_{m=1}^{\Tmax} N^{[u_m]}\le \Tmax N(D)\le (1+\eta_{\rm act})C_{\rm act}|D|\,\Tmax$, so the union bound yields
the displayed probability bound in the proposition.
\end{proof}

\subsection{Specialization: (in)homogeneous Poisson with dispersion}\label{sec:apppoisson}

We next consider the Poisson limit of Hawkes processes, i.e.\ \(\xi(A)=0\), so that \(g_{k\ell}\equiv 0\) and \(\lambda_k=\mu_k\). The same no-feedback simplification holds for a common predictable dispersion field \(V(\tau)\) in a Cox-type specification \(\lambda_k(\tau)=V(\tau)\mu_k(\tau)\), provided \(V\) is controlled so that the bounded-intensity and concentration assumptions continue to hold. In this case there is no label feedback: \(\widehat\lambda_k^r\equiv \lambda_k\) for every labelling \(r\). Hence Assumption~\ref{ass:G3} holds with zero flip-influence constants, Assumption~\ref{ass:G4} holds with \(\Delta_s=0\), and Assumption~\ref{ass:G2prime} reduces to a bounded-intensity condition with \(K_{\rm win}=O(1)\). The remaining requirements reduce to ordinary regularity conditions for the two-component Poisson family, namely positive bounded intensities, local curvature/identifiability of the Poisson log-likelihood, and the decisive-set mass condition in Assumption~\ref{ass:G8}. Under those conditions, the hard-EM floor in \cref{thm:contract} has no feedback-propagation term and is driven only by oracle-minority mass on decisive sets together with the ambiguous-band mass \(C_b(\theta^\star)\). Thus hard assignments may plateau at this classification floor, whereas SEM-style averaging can still remain consistent because label errors do not propagate through time.

\begin{lemma}[Ambiguous-band mass for Poisson]\label{lem:poisson-Cb}
Let
\[
S_0(\tau):=s^{\rm or}_{\theta^\star}(\tau)
=\log\frac{\lambda_1(\tau\mid\mathcal H_{t-};\theta^\star)}{\lambda_0(\tau\mid\mathcal H_{t-};\theta^\star)}.
\]
Suppose the $\lambda^\star$-weighted distribution of $S_0$ has a density bounded by $L_0$ on $[-2b,2b]$. Then
\[
C_b(\theta^\star)
=\frac{1}{|D|}\int_D \1\{|S_0(\tau)|\le 2b\}\,\lambda^\star(\tau)\,\dd\tau
\ \le\ 4L_0 b\,K_{\rm win}.
\]
\end{lemma}

\begin{proof}
Immediate from the definition, $\lambda^\star(\tau)\le K_{\rm win}$, and the density bound.
\end{proof}

\end{document}